\documentclass[11pt,preprint]{amsart}

\usepackage{amsmath, latexsym, amsfonts, amssymb, amsthm}
\usepackage{graphicx, color,hyperref,dsfont,epsfig,caption,wrapfig,subfig}
\usepackage{mathtools}
\usepackage{movie15}
\hypersetup{
	linktoc=page,
	linkcolor=red,          
	citecolor=blue,        
	filecolor=blue,      
	urlcolor=cyan,
	colorlinks=true           
}

\numberwithin{equation}{section}

\newtheorem{theorem}{Theorem}[section]
\newtheorem{lemma}[theorem]{Lemma}
\newtheorem{proposition}[theorem]{Proposition}

\newtheorem{remark}[theorem]{Remark}

\newcounter{thmc}

\newtheorem{thmcite}[thmc]{Theorem}

\theoremstyle{definition}

\let\oldtocsection=\tocsection
\let\oldtocsubsection=\tocsubsection
\let\oldtocsubsubsection=\tocsubsubsection
\renewcommand{\tocsection}[2]{\hspace{0em}\oldtocsection{#1}{#2}}
\renewcommand{\tocsubsection}[2]{\hspace{1em}\oldtocsubsection{#1}{#2}}
\renewcommand{\tocsubsubsection}[2]{\hspace{2em}\oldtocsubsubsection{#1}{#2}}

\renewcommand{\tilde}{\widetilde}          
\DeclareMathSymbol{\leqslant}{\mathalpha}{AMSa}{"36} 
\DeclareMathSymbol{\geqslant}{\mathalpha}{AMSa}{"3E} 
\DeclareMathSymbol{\eset}{\mathalpha}{AMSb}{"3F}     
\renewcommand{\leq}{\;\leqslant\;}                   
\renewcommand{\geq}{\;\geqslant\;}                   
 
\newcommand{\bigslant}[2]{{\raisebox{.2em}{$#1$}\left/\raisebox{-.2em}{$#2$}\right.}}

\newcommand{\C}{\mathbb{C}}
\newcommand{\R}{\mathbb{R}}
\newcommand{\Z}{\mathbb{Z}}
\newcommand{\N}{\mathbb{N}}
 
\newcommand{\D}{\mathbb{D}} 
\newcommand{\E}{\mathbb{E}} 
\def\T{\mathbb{T}}

\def\A{\bm{\mathrm{A}}}
\def\vac{\vert 0\rangle}
\def\cav{\langle 0\vert}
\def\prim{\vert \alpha\rangle}
\def\vesp{\mc C_\infty}
\def\endv{\text{End}(\vesp)}
\def\endov{\text{End}(\V_{\bm c})}
\def\endovp{\text{End}(\V_{+,\bm c})}

\def\vir{\text{Vir}}

\renewcommand{\P}{\mathbb{P}}
\renewcommand{\L}{\bm{\mathrm{L}}}

\newcommand{\ps}[1]{\langle #1 \rangle}
\newcommand{\qt}[1]{\quad\text{#1}\quad}
\newcommand{\MW}[1]{\prescript{#1}{}{\mc M_+}}
\newcommand{\Wmod}{\prescript{\g}{}{\mc M_\alpha}}
\newcommand{\Verma}{\prescript{\g}{}{\mc V_\alpha}}

\newcommand{\eqlaw}{\overset{\text{(law)}}{=}}

\newcommand{\W}[1]{\bm{\mathrm W}^{(#1)}}

\def\sl{\mathfrak{sl}}
\def\V{\bm{\mathrm V}}

\def\Wb{\bm{\mathrm W}}

\def\X{\bm{\mathrm X}}
\def\P{\bm{\mathrm P}}
\def\A{\bm{\mathrm A}}

\def\L{\bm{\mathrm L}}
\def\Hil{\mathcal{H}_{\T}}
\def\Hild{\mathcal{H}_{\D}}

\def\Lro{L^2(\a\times\Omega_\T)}

\def\d{\mathrm d}
\def\i{\bm{\mathrm i}}
\def\vecgen{\vert\bm u,\bm n\rangle}
\def\h{\mathfrak h}

\newcommand{\eps}{\epsilon}

\newcommand{\mc}{\mathcal}

\def\eps{\varepsilon}
\def\S{\mathbb{S}}
\def\T{\mathbb{T}}

\def\g{\mathfrak g}
\def\a{\mathfrak a}

\def\bi{\begin{itemize}}
	\def\ei{\end{itemize}}

\def\bnum{\begin{enumerate}}
	\def\enum{\end{enumerate}}

\def\<#1{\langle #1 \rangle}

\renewcommand{\P}{\mathbb{P}}

\newcommand{\norm}[1]{\left\lvert#1\right\rvert}
\newcommand{\nnorm}[1]{\left\lVert#1\right\rVert}
\newcommand{\expect}[1]{\mathbb{E}\left[#1\right]}

\newcommand{\ephi}[1]{\mathbb{E}_\varphi\left[#1\right]}
\usepackage{bm}
\usepackage{bbm}

\title[W-algebras, Gaussian Free Fields and $\g$-Dotsenko-Fateev integrals]{W-algebras, Gaussian Free Fields and $\g$-Dotsenko-Fateev integrals}

\author{Baptiste Cercl\'e}
\email{baptiste.cercle@epfl.ch}
\address{EPFL SB MATH RGM, MA B2 397, Station 8, CH-1015 Lausanne, Switzerland.}

\begin{document}

	\begin{abstract}
		Based on the intrinsic connection between Gaussian Free Fields and the Heisenberg vertex algebra, we study some aspects of the correspondence between probability theory and $W$-algebras. This is first achieved by providing a construction of the $W$-algebra associated to a complex simple Lie algebra $\mathfrak g$ by means of Gaussian Free Fields.
		
		This correspondence in turn allows to translate algebraic statements into actual constraints for free-field correlation functions. This leads to new integrability results for Dotsenko-Fateev integrals associated to $\g$, such as Ward identities and the derivation of a new Fuchsian differential equation for deformations of $B_2$-Dotsenko-Fateev integrals arising from the Mukhin-Varchenko conjecture. 
		
		Along the proof of this statement we also provide new results on representation theory of $W$-algebras such as the description of some singular vectors for the $W$-algebra associated to $\mathfrak g=B_2$. 
	\end{abstract}
	
	\maketitle
    
	\section{Introduction}
	
	\subsection{Different perspectives on Conformal Field Theory}
	Two-dimensional Conformal Field Theory (CFT) is a topic of research that has attracted a lot of attention since its inception by Belavin-Polyakov-Zamolodchikov (BPZ) in 1984~\cite{BPZ}. The study of this notion has led to major breakthroughs in many different areas of mathematics and physics, ranging from representation theory to geometric quantization, conformal geometry or statistical physics.
	
	However even answering the seemingly simple question \lq\lq What is a Conformal Field Theory?" is not as straightforward as one might expect and many different approaches have been developed to address this issue.  The coexistence of several viewpoints on CFT is exemplified by one of its most simplest objects: the free-field or free boson. This apparently simple object does indeed exhibit several non-trivial features, and can already be thought of in many different ways for which it may not be clear at first sight that they do describe the very same object. 
	
	For instance a free field can be defined as a scalar field, or put differently a formal measure $\mu$ over a set $\mc F$ of functions $\varphi:\Sigma\to\R$. This formal measure is built using a path integral:
	\begin{equation}\label{eq:free_boson}
		\mu(D\varphi)=\frac1{\mc Z}e^{-S(\varphi)}D\varphi,\qt{with} S(\varphi)=\frac{1}{4\pi}\int_\Sigma \norm{\nabla_g \varphi}_g^2\mathrm {dv}_g
	\end{equation}
	where $D\varphi$ would stand for a \lq\lq uniform measure" over $\mc F$ and which as such doesn't make sense mathematically speaking. The rigorous definition of this formal measure comes from probability theory, where the law of a Gaussian Free Field (GFF) plays the role of $\mu$.
	
	On another perspective, a free boson is often referred to as the current that generates the $\hat u(1)$ symmetry (see \textit{e.g.}~\cite[Section 6.1]{Gin_CFT}), which may be understood by saying that the free boson $\Phi$ is the generating function for the Heisenberg algebra. This means that we can write (again formally)
	\begin{equation}\label{eq:Phi_intro}
		\partial\Phi(z)=\sum_{n\in\Z} \A_{n}z^{-n-1}
	\end{equation}
	where the modes $\A_n$ satisfy the commutation relations of the Heisenberg algebra (see Equation~\eqref{eq:comm_he}). The language of \textit{Vertex Operator Algebras} (VOAs) provides the natural framework to make sense of this heuristic, and the above writing allows to define the \textit{Heisenberg vertex algebra}.
	
	And besides these two viewpoints there are yet other ways to make sense of this free boson, such as viewing it as an infinite-dimensional harmonic oscillator (see \textit{e.g.}~\cite[Section 2.1.1]{yellow_book}), a perspective whose connection with the approach based on VOA may seem more natural since both are based on the consideration of the Heisenberg algebra. But before discussing the connections between these different approaches let us first present in more details the probabilistic approaches to CFT based on Gaussian Free Fields and the one based on Vertex Operator Algebra.
	
	\subsubsection{Probabilistic approaches to Conformal Field Theory}
	The probabilistic interpretation in terms of a GFF of the formal path integral~\eqref{eq:free_boson} is at the heart of the recent mathematical works aimed at studying in a rigorous way \textit{Liouville CFT}~\cite{Pol81}. Initiated by David, Kupiainen, Rhodes and Vargas in 2014~\cite{DKRV}, this mathematical program has led to major breakthroughs such as a rigorous derivation of the structure constants of Liouville CFTs~\cite{KRV_DOZZ, ARS, ARSZ} (building on the approach proposed in~\cite{Teschner_DOZZ, FZZ, Hos, PT02})  and more generally a mathematical implementation of the conformal bootstrap procedure~\cite{KRV_loc,GKRV, GRW} (see also~\cite{Teschner_revisited}). These works have implications beyond Liouville CFT as they allow for a deeper comprehension of CFT, as exemplified by the derivation~\cite{GKRV_Segal} of Segal's axioms~\cite{Seg04} and the better understanding of Virasoro conformal blocks~\cite{BGKR} it allows for. 
	
	Conformal invariance can be supplemented by an additional level of symmetry, often referred to as $W$ or \textit{higher-spin} symmetry (discussed in more details in the next subsection).
    The analogs of Liouville theory in this context are Toda CFTs, whose algebras of symmetry are no longer the Virasoro algebra but rather \textit{$W$-algebras}. A mathematical approach to these models~\cite{Toda_construction,CH_construction} based on the probabilistic interpretation of Equation~\eqref{eq:Phi_intro} led to the rigorous derivation of some of the predictions made in the physics literature~\cite{FaLi1} for the $\g=\sl_3$ Toda CFT~\cite{Toda_OPEWV, Toda_correl1, Toda_correl2}, the simplest instance of a Toda CFT for which the algebra of symmetry strictly contains Virasoro algebra. 
	
	There are yet other approaches to 2d CFT based on probabilistic tools, one key example being the connections between CFT, Schramm-Loewner Evolutions (SLEs) and Conformal Loop Ensembles~\cite{Sh_CLE, SW, MSW}.  These links allow to derive striking results for CLEs, such as exact formulas for certain CLE observables~\cite{AS_DOZZ} or for the backbone exponent for Bernoulli percolation~\cite{Backbone}. The construction in~\cite{BJ_SLE} of a representation of the Virasoro algebra in the setting of SLE loop measures and its implications are another example of the strong interplays between probability and CFT. 
	
	\subsubsection{Conformal Field Theory and Vertex Operator Algebras}
	Providing a rigorous meaning to the formalism introduced in~\cite{BPZ} can be achieved in different ways. For instance the approach developed by BPZ features formal expressions similar to Equation~\eqref{eq:Phi_intro}: giving a sense to the latter can be achieved by understanding~\eqref{eq:Phi_intro} as a formal power series in $z$ and $z^{-1}$ whose coefficients are linear operators acting on some vector space $V$ called \textit{space of states}. For this writing to make sense several assumptions are required, leading to the definition of \textit{Vertex Operator Algebra}.
	
	Indeed, a powerful framework to make sense of the formalism of BPZ, especially in the setting of \textit{minimal models}, is that of VOAs. Originally introduced by Borcherds~\cite{Borcherds} in the context of the \lq\lq monstrous moonshine conjecture" (see also~\cite{FLM89}), it allows for an algebraic and axiomatic study of the holomorphic (or chiral) part of a CFT and in particular provides a natural setting where to make sense of some of the key tools used in the physics such as \textit{Operator Product Expansions} or representation theory of the Virasoro algebra. Likewise, geometric approaches to the notion of VOA~\cite{Hua_CFT} have proved to be fundamental in conjunction with more geometric approaches to CFT as introduced in the landmark article~\cite{FS87} or in relation with Segal's axioms for CFT~\cite{Seg04}. 
	
	The formalism of VOA also proves to be crucial in the rigorous description of the symmetries of models such as Toda CFTs. The corresponding VOAs are then $W$-\textit{algebras}, which unlike the Virasoro algebra are \textit{not} Lie algebras: as such the language of VOA is necessary to make sense of them. We will discuss in more details $W$-algebras and their mathematical study in Subsection~\ref{sec:W_intro}.

	\subsubsection{Connections between the two approaches}
	The initial motivation for writing this document was to try to conciliate these different viewpoints, and to understand to what extent it was possible to implement techniques and methods coming from the language of VOA in the setting of Liouville and Toda CFTs. Indeed the existence of links between the path integral approach defining (formally) the free boson is not new as it already appears in several works. For instance such a correspondence has already been established at the lattice level in~\cite{Tsukada}, while the intrinsic connection between the GFF and the Heisenberg algebra has been worked out in several different frameworks (see \textit{e.g.}~\cite{Kang-Makarov}) and has proved to be key in various settings, for instance to understand scaling limits of a discrete GFF in a sense that would preserve the structure of bosonic CFT~\cite{ACBK}.
	
	However and to the best of our knowledge there are no formulations of this correspondence that would fit in the setting of Liouville and Toda CFTs, and that would allow to translate methods coming from VOAs into algebraic constraints on these models. Our first result, Theorem~\ref{thm:VOA_Proba}, is meant to overcome this issue as it fits exactly in the Hilbert space picture inherent to these theories. Informally speaking (more details are given in Section~\ref{sec:GFF}) this Hilbert space $\Hil$ is a space of states $(c,\varphi)$ with $c\in\R$ and $\varphi:\T\to\R$ a map over the unit circle $\S^1=\T$. This Hilbert space can be identified with a set $\Hild$ of functionals of a field $\X:\D\to\R$ using the map (introduced in~\cite{GKRV})
	\begin{equation}
		\begin{matrix}
			U_0:& \Hild&\to & \Hil\\
			& F& \mapsto & U_0F
		\end{matrix}\qt{with}U_0F:(c,\varphi)\mapsto e^{-Qc}\ephi{F(\X)}
	\end{equation}
	where, under $\d\P_\varphi$, $\X$ is a Dirichlet GFF inside $\D$ with boundary condition over $\T$ given by $\varphi$ (see Section~\ref{sec:GFF}). Here $Q$ is the so-called \textit{background charge}, and is parametrized in terms of a coupling constant $\gamma\in\C$ via $Q=\frac\gamma2+\frac2\gamma$. The central charge $\bm c$ of the theory is then given by $\bm c=1+6Q^2$.
    
	A Fock space can be extracted out of this picture by considering the algebraic direct sum
	\begin{equation}
		\V_{+,\bm c}\coloneqq \bigoplus_{p\geq 0}\V_{+,\bm c}^{(p)},\quad \V_{+,\bm c}^{(p)}\coloneqq\bigoplus_{\substack{n_1,\cdots,n_k\geq 1\\ n_1+\cdots+n_k=p}}\text{span}\left\{U_0F,\quad F:\X\mapsto :\prod_{k=1}^p\frac{\partial^{n_k}\X(0)}{n_k!}:\right\}
	\end{equation}
    where the notation $:\cdot:$ stands for a Wick product and is defined using a limiting procedure based on a regularization $\X_\eps$ of $\X$ (see Subsection~\ref{subsec:field_obs}). Thanks to this observation, our first result is the probabilistic construction of a vertex algebra isomorphic to the Heisenberg vertex algebra, given by the data of $(\V_{+,\bm c},\vac,T,Y)$. Here the vacuum vector is $\vac=U_01\in\V_{+,\bm c}$ and the vertex operators $Y(\cdot,z):\V_{+,\bm c}\to\text{End}(\V_{+,\bm c})[[z,z^{-1}]]$ are the assignment for $u,v\in\V_{+,\bm c}$ and $z\in\C\setminus\{0\}$ of
	\begin{equation}\label{eq:voa_prob1}
		\begin{split}
			&Y\left(u,z\right)v\coloneqq	U_0\left(:\prod_{k=1}^p\frac{\partial^{n_k}\X(z)}{n_k!}::\prod_{l=1}^q\frac{\partial^{m_l}\X(0)}{m_l!}:\right)\qt{if}\\
            &u=U_0\left(:\prod_{k=1}^p\frac{\partial^{n_k}\X(0)}{n_k!}:\right)\qt{and}v=U_0\left(:\prod_{l=1}^q\frac{\partial^{m_l}\X(0)}{m_l!}:\right).
		\end{split}
	\end{equation}
    We refer to Theorem~\ref{thm:VOA_Proba} for a more detailed statement.
	This specific formulation proves to be important especially in the context of the recent probabilistic works on Liouville and Toda CFTs. Indeed one important application of this statement is the translation in a probabilistic language of the technique of \textit{Operator Product Expansions} (see Proposition~\ref{prop:ope_prob}), which is fundamental in the VOA approach and that allows to turn \textit{a priori} probabilistic computations into algebraic ones. More generally we expect to be able to use this approach to translate more directly statements coming from representation theory of the Virasoro and $W$-algebras into probabilistic constraints.

	\subsection{From VOAs to Dotsenko-Fateev integrals via representation theory of $W$-algebras}
	Let us illustrate this by deriving some consequences of this correspondence on free-field correlation functions associated to Toda CFTs, and more precisely \textit{Ward identities}. We also discuss to what extent it allows for a systematic study of $\g$-Dotsenko-Fateev integrals.
	
	First of all we need a straightforward generalization of the above correspondence where we consider, instead of a scalar field $\varphi:\T\to\R$, a vectorial GFF having values in a $r$-dimensional Euclidean space $\a\simeq\R^r$. We obtain a similar probabilistic construction of the rank $r$ Heisenberg vertex algebra, from which we can consider special vertex subalgebras called \textit{$W$-algebras}.
	
	\subsubsection{W-algebras and Conformal Field Theory}\label{sec:W_intro}
	Despite lacking the mathematical formalism of VOA, $W$-algebras were first introduced in the physics literature by Zamolodchikov~\cite{Za85}, and subsequently developed by Fateev-Lukyanov~\cite{FaLu}, to extend the conformal bootstrap procedure of BPZ to the case where the theory under consideration enjoys an enhanced symmetry. $W$-algebras have proved to be fundamental in many areas of physics, should it be thanks to their links with Kac-Moody algebras~\cite{BFFOrW3}, gauge theories via the AGT correspondence~\cite{AGT} or statistical physics (one unexpected such connection being the one between Ising Model in a Magnetic Field at criticality and a Toda theory associated to the Lie algebra $E_8$~\cite{Za_E8}, see also~\cite{magnet}). We refer to~\cite{BouSch} for a more in-depth review of some of the physics aspects of $W$-algebras.
	
	On the mathematical side, the study of $W$-algebras is a thriving field of research. Following the works of Feigin-Frenkel~\cite{FF_QG, FF_DS} and Kac-Roan-Wakimoto~\cite{KRW}, several approaches have been developed to understand mathematically these objects. One tool that proves to be powerful in this respect is the so-called \textit{Drinfeld-Sokolov reduction}, which thanks to its functoriality turns out to be an efficient approach to obtain results concerning the representation theory of $W$-algebras~\cite{FKW,Arakawa_rep}, see also~\cite{Arakawa_intro}. There are yet other definitions for $W$-algebras, and a particularly meaningful one in the context of actual models of CFT such as Toda CFTs is that based on screening operators~\cite{FF_KM, FF_QG}. As we will now explain, the consideration of such operators is very natural from the point of view of the probabilistic approach to (Toda) CFTs. As such the equivalence~\cite{FF_KM} between these two definitions for $W$-algebras is key in the prospect of translating representation theoretical results into actual statements for CFTs with such symmetries.
	
	\subsubsection{From the Heisenberg vertex algebra to $W$-algebras}
	To obtain a correspondence between GFFs and $W$-algebras, let  $\g$ be a finite-dimensional, complex simple Lie algebra of rank $r$ (we recall some basics on this topic in Subsection~\ref{subsec:lie}), $\mathfrak{h}$ be a Cartan subalgebra and $\mathfrak{h}^*$ be the dual of $\h$. For the sake of simplicity $(\h^*,\kappa)$ is identified with $(\C^r,\ps{\cdot,\cdot})$, where $\kappa$ is inherited from the Killing form of $\g$. The coupling constant $\gamma\in\C$ defines the background charge $Q\in\C^r$ and the central charge $\bm c$ via\footnote{The difference with the usual convention in Liouville theory where we have instead $Q=\frac\gamma2+\frac2\gamma$ stems from the convention that the longest simple roots of $\g$ have squared norm $2$ rather than $1$, see Section~\ref{sec:W-algebra}.}(here $\rho,\rho^\vee\in\C^r$ stand for the Weyl and co-Weyl vector of $\g$, see~\eqref{eq:def_rho})
	\begin{equation}
		Q=\gamma \rho+\frac2\gamma\rho^\vee\qt{and}\bm c=r+6\ps{Q,Q}.
	\end{equation} 
    For convenience, the correspondence with the parametrization in terms of a \textit{level} $k$ often encountered in the $W$-algebra literature is given by $k+h^\vee=-\frac2{\gamma^2}$, with $h^\vee$ the dual Coxeter number.
	
	The $W$-algebra associated to $\g$ can be constructed~\cite{FF_QG,FBZ} out of the rank $r$ Heisenberg vertex algebra by means of \textit{screening operators}, a family of operators $(Q_i)_{1\leq i\leq r}$ acting on $\V_{+,\bm c}$. They are formally defined based on \textit{bosonic Vertex Operators} $\mc V_{\gamma e_i}^+$ (introduced in Section~\ref{sec:VOA}) by setting
	\begin{equation}
		Q_i\coloneqq \oint \mc V_{\gamma e_i}^+(z)dz
	\end{equation}
	where the $(e_i)_{1\leq i\leq r}$ are elements of $\R^r$ stemming from the \textit{simple roots} of $\g$ and the integration contour surrounds the origin.
	Then for generic values\footnote{That is for $\bm c$ in a dense open subset of $\C$. Outside of this set we can still construct currents $\W{s}$ but the dimension of $\MW{\g}^{(n)}$ may increase~\cite[Lemma 4.4.2]{FF_QG}.} of the central charge $\bm c$, the restriction of the above probabilistic construction of the Heisenberg vertex algebra to the subspace $\MW{\g}\coloneqq \bigcap_{i=1}^r \ker_{\V_{+,\bm c}}\left(Q_i\right)$ of $\V_{+,\bm c}$ 
    defines a vertex algebra isomorphic to the $W$-algebra associated to $\g$.
	
	This allows us to construct, in Section~\ref{sec:W-algebra}, $W$-algebras within our probabilistic framework. To do so we consider Young diagrams, \textit{i.e.} finite families of positive integers $\lambda=(\lambda^1,\cdots,\lambda^l)$ with $\lambda_1\geq\cdots\geq\lambda_l$. Here $l=l(\lambda)$ is the length of $\lambda$ while its level is given by $\norm{\lambda}=\sum_{k=1}^{l(\lambda)}\lambda_k$. In the rest of this document we will denote, for a positive integer $p$, by $\mc T_p$ the set of Young diagrams such that $\lambda_{l}\geq p$.
    In the following statement, we will say that $\Wb$ is a $p$-multilinear differential operator when there are $l(\lambda)$-linear forms $\Wb_\lambda$ over $\C^r$ such that it admits an expansion of the form
    \[
        \Wb[\X]= \sum_{\lambda\in\mc T_1,\text{ }\norm{\lambda}=p}:\Wb_\lambda(\partial^{\lambda_1}\X,\cdots,\partial^{\lambda_l}\X):.
	\]
    Using the probabilistic construction of the Heisenberg vertex algebra described above, the statement of~\cite[Theorem 4.6.9]{FF_QG} then translates as follows (see Theorem~\ref{thm:def_W} for more details):
	\begin{theorem}\label{thm:W_free1}
		For generic values of $\bm c$, there exist $s_i$-multilinear differential operators $\W{s_i}$ for $i=1,\cdots, r$ such that $\MW{\g}=\bigoplus_{n\geq0}\left(\MW{\g}\right)^{(n)}$ where
        \begin{equation}
            \left(\MW{\g}\right)^{(n)}\coloneqq\bigoplus_{\substack{\lambda_1\in\mc T_{s_1},\cdots\lambda_r\in\mc T_{s_r}\\ \norm{\lambda_1}+\cdots+\norm{\lambda_r}=n}}\text{span}\left\{e^{-\ps{Q,c}}\ephi{:\prod_{i=1}^r\prod_{j=1}^{l(\lambda_i)}\frac{\partial^{\lambda_i^j-s_i}\Wb^{(s_i)}[\X](0)}{(\lambda_i^j-s_i)!}:}\right\}.
		\end{equation}
	\end{theorem}
	The defining property of the $\W s$ comes from their link with screening operators. This seemingly abstract definition has strong implications for free-field correlation functions as we now explain.
	
	\subsubsection{Free-field correlation functions and $\g$-Dotsenko-Fateev integrals}
    Free-field correlation functions are described by Dotsenko-Fateev type integrals~\cite{DF1,DF2} as soon as a so-called \textit{neutrality condition} is satisfied. The correspondence between GFFs and $W$-algebras in the form of Theorem~\ref{thm:W_free1} thus naturally allows to translate results from VOAs to such integrals.
    
	Namely, let $\alpha_1,\cdots,\alpha_N\in\C^r$ and assume that for some non-negative integers $n_1,\cdots,n_r$
	\begin{equation}\label{eq:neutrality_gen1}
		\sum_{k=1}^N\alpha_k-2Q=-\sum_{i=1}^r n_i \gamma e_i.
	\end{equation}
	The correlation functions of Vertex Operators\footnote{We use this notation in analogy with Toda correlation functions. Indeed Dotsenko-Fateev-type integrals naturally arise from Toda correlation functions when the weights satisfy~\eqref{eq:neutrality_gen1}, see for instance~\cite[Subsection 5.1]{DKRV} and~\cite{GKR_CILT} in the Liouville case. As such it is expected that the free-field correlation functions considered here and Toda correlation functions both satisfy the same Ward identities.} are then defined for $z_1,\cdots,z_N\in\C$ distinct by
	\begin{equation}\label{eq:dot_fat1}
		\begin{split}
			&\ps{\prod_{k=1}^NV_{\alpha_k}(z_k)}_{\gamma}\coloneqq \prod_{k<l}\norm{z_k-z_l}^{-\ps{\alpha_k,\alpha_l}}\mc I_n(\bm\alpha,\bm z),\qt{with}\\
			&\mc I_n(\bm\alpha,\bm z)\coloneqq\int_{\C^{n}}\prod_{i=1}^n \prod_{k=1}^N\norm{x_i-z_k}^{-\ps{\alpha_k,\gamma e_{x_i}}}\prod_{i< j}\norm{x_i-x_j}^{-\ps{\gamma e_{x_i},\gamma e_{x_j}}}\mathrm d^2x_1\cdots \mathrm d^2x_n, 
		\end{split}
	\end{equation}
    provided it makes sense, where $n=\sum_{i=1}^rn_i$ and $e_{x_j}=e_l$ for $l$ with $\sum_{i=1}^{l-1}n_i< j\leq \sum_{i=1}^ln_i$.

	The one-dimensional version of such integrals (whose connection with Dotsenko-Fateev integrals is well-known~\cite{Aomoto, Neretin}), that is where the integrals range over $(0,1)$ (or a suitable contour) instead of $\C$, and for which $N=3$ with $z_1=0$, $z_2=1$ and $z_3=\infty$, are given by $\mathfrak g$-Selberg integrals :
	\begin{equation}
		\mc S_{\bm n,\g}(\alpha_1,\alpha_2)=\int_{(0,1)^{n}}\prod_{i=1}^n t_i^{-\ps{\alpha_1,\gamma e_{t_i}}}(1-t_i)^{-\ps{\alpha_2,\gamma e_{t_i}}}\prod_{i< j}\norm{t_i-t_j}^{-\ps{\gamma e_{x_i},\gamma e_{x_j}}}\mathrm dt_1\cdots \mathrm dt_n.
	\end{equation}
	In the case where $\g=\sl_2$, this is the usual Selberg integral~\cite{Selberg,Importance}, given by (with $e_1=1$)
	\begin{equation}
		\mc S_{n,\g}(\alpha_1,\alpha_2)=\prod_{k=0}^{n-1}\frac{\Gamma\left(\gamma\alpha_1-k\frac{\gamma^2}2\right)\Gamma\left(\gamma\alpha_2-k\frac{\gamma^2}2\right)}{\Gamma\left(\gamma(\alpha_1+\alpha_2)-(k+n-1)\frac{\gamma^2}2\right)\Gamma\left(-\frac{\gamma^2}2\right)}\cdot
	\end{equation}
	However in general the question of the evaluation of such integrals remains a very important yet unsolved problem, often called the \textit{Mukhin-Varchenko conjecture}~\cite{MuVa}. More precisely, it is conjectured in~\cite{MuVa} that under specific assumptions made on $\alpha_1$ and $\alpha_2$, stemming from representation theory of $\g$, the corresponding $\mathfrak{g}$-Selberg integral can be expressed as a ratio of Gamma functions. 
	
	A first answer to the Mukhin-Varchenko conjecture was brought a few years later~\cite{TaVa} where the case $\g=\sl_3$ was solved. This result was then extended in 2006 by Warnaar~\cite{War_Sel}, where a formula describing $\mathfrak{g}$-Selberg integrals was unveiled in the case where the underlying Lie algebra was of the form $\g=A_n$. To this end the author relied on the theory of symmetric functions and more precisely Macdonald polynomials. 
	However and apart from very specific cases such as the one considered in~\cite{MT} (derived directly from the usual Selberg integral), there are no known results when the Lie algebra under consideration is not of type $A$ (that is $\sl_n$), and the question of providing explicit expressions or at least integrability results for $\g$-Selberg integrals is still an open one.

	\subsubsection{Representation theory of $W$-algebras and (Fuchsian) differential equations}
	We propose to address this issue by describing a method inspired by the study of Toda CFTs. The first step in this perspective is the proof that the correlation functions inherit symmetries from the underlying structure of $W$-algebras, which in turn put strong constraints on these quantities. Such constraints are expressed by means of local and global \textit{Ward identities}, encoding respectively the local and global symmetries of the model. In the context of Selberg integrals these Ward identities can be understood as integrability properties for (deformations) of these quantities. We prove in Section~\ref{sec:last} that such Ward identities are indeed valid and that they take the following form:	
	\begin{theorem}\label{thm:ward_gen1}
		Assume the neutrality condition~\eqref{eq:neutrality_gen1} to hold, and that the corresponding integral~\eqref{eq:dot_fat1} is absolutely convergent. Then for $z_1,\cdots,z_n\in\C$ distinct and $s$ any spin of $\g$: 
		\begin{equation}
			\ps{\W s_{-s}V_{\alpha_1}(z)\prod_{k=2}^NV_{\alpha_k}(z_k)}_{\gamma}=\sum_{k=2}^N\sum_{i=1}^s\frac{1}{(z-z_k)^i}\ps{V_\alpha(z)\W s_{i-s}V_{\alpha_k}(z_k)\prod_{l\neq k}^NV_{\alpha_l}(z_l)}_{\gamma}.
		\end{equation}
        We also have global Ward identities for all $0\leq n\leq 2s-2$:
		\begin{equation}
			\sum_{k=1}^N\sum_{i=0}^{s-1}\binom{n}{i}z_k^{n-i}\ps{\W s_{i+1-s}V_{\alpha_k}(z_k)\prod_{l\neq k}V_{\alpha_l}(z_l)}_{\gamma}=0.
		\end{equation}
	\end{theorem}
	In the above the $\W s_{-i}V_{\alpha}$ (defined in Proposition~\ref{prop:ward}) are called \textit{descendant fields} and the correlation functions containing them are defined via a limiting procedure based on their probabilistic representation. More explicitly, we formally have $\W s_{-i}V_\alpha(z)=\lim\limits_{\eps\to0}:\W s_{-i,\alpha}[\X_\eps]e^{\ps{\alpha,\X_\eps}}(z):$ for $\X_\eps$ regularization of $\X$, where $\W s_{-i,\alpha}$ is a $i$-multilinear differential operator constructed from $\W s$.
	
	The above Ward identities are general constraints put on a theory that enjoys a symmetry prescribed by $W$-algebras. There are however additional restrictions stemming from the particular representation of the $W$-algebra under consideration. For the free-field representation considered in this document such constraints are prescribed by the existence of \textit{singular vectors}. These are non-trivial linear relations between descendants $\W {s}_{-\lambda}V_\alpha$ for $s$ ranging over the spins of $\g$ and with $\norm{\lambda}$ fixed. The computation of explicit expressions for these singular vectors thus provides additional constraints on correlation functions containing such descendant fields. And under suitable assumptions this gives rise to Fuchsian differential equation for certain correlation functions.
    
	This method led to the computation of the structure constants for the Liouville and the $\g=\sl_3$ Toda CFTs. Beyond these cases, it is predicted~\cite{FaLi1} that in the $\g=A_n$ case a family of four-point correlation functions solves an hypergeometric differential equation of order $n+1$.
	
	One of our result is an application of this method to the $\g=B_2$ case. Namely based on new results concerning the representation theory of the $W$-algebra associated to $B_2$ (see Theorem~\ref{thm:Verma}) we prove the following in Section~\ref{sec:last} (to which we refer for additional details and notations): 
	\begin{theorem}\label{thm:B2}
		Let $\alpha^*=-\frac\gamma2\omega_2^\vee$ and $\alpha_2=-\gamma\omega_1^\vee$ with $(\omega_1^\vee,\omega_2^\vee)$ the basis dual to $(e_1,e_2)$. Then under the neutrality condition~\eqref{eq:neutrality_gen} and provided that it is well-defined, the correlation function $\ps{V_{\alpha^*}(z)V_{\alpha_1}(0)V_{\alpha_2}(1)V_{\alpha_3}(\infty)}_\gamma$ is a solution of a Fuchsian differential equation of order $4$:
		\begin{equation}\label{eq:PDE1}
			\mc D(\bm\alpha,z)\mc H(z)=0,\quad\mc H(z)\coloneqq \norm{z}^{\ps{\alpha^*,\alpha_1}}\norm{z-1}^{\ps{\alpha^*,\alpha_2}}\ps{V_{\alpha^*}(z)V_{\alpha_1}(0)V_{\alpha_2}(1)V_{\alpha_3}(\infty)}_\gamma.
		\end{equation}
		Here $\mc D(\bm\alpha,z)$ is the Fuchsian differential operator (for some explicit coefficients $A_i$, $B_i$ and $\tilde A_i$, $\tilde B_i$)
		\begin{equation}\label{eq:fuchs}
			\begin{split}
				\mc D(\bm\alpha, z)&=z^2\prod_{k=1}^4\left(z\partial_z+A_k\right)+\prod_{k=1}^4\left(z\partial_z+B_k-1\right)\\
				&-z\left(\prod_{k=1}^4\left(z\partial_z+\tilde A_k\right)+\prod_{k=1}^4\left(z\partial_z+\tilde B_k-1\right)\right).
			\end{split}
		\end{equation}
	\end{theorem}
	More generally, we show that when $\alpha+\beta+\gamma\sum_{i=1}^r n_i e_i-2Q-\gamma\omega_1\in\left\{\gamma\omega_1^\vee,\frac\gamma2\omega_2^\vee,\frac2\gamma\omega_1^\vee,\frac2\gamma \omega_2^\vee\right\}$, a deformation of the following $B_2$-Dotsenko-Fateev integrals
    \begin{equation}
		\begin{split}
			\int_{\C^n}\prod_{i=1}^n x_i^{-\ps{\alpha,\gamma e_{x_i}}}\norm{1-x_i}^{-\ps{\beta,\gamma e_{x_i}}}\prod_{i<j}\norm{x_i-x_j}^{-\ps{\gamma e_{x_i},\gamma e_{x_j}}}\d^2x_1\cdots \d^2x_n
		\end{split}
	\end{equation}
    gives rise to a Fuchsian differential equation analogous to Equation~\eqref{eq:fuchs}.
	To the best of our knowledge, this is the first integrability result to address the Mukhin-Varchenko conjecture when the underlying algebra is not of type $A$. We sketch in Section~\ref{sec:last} a method of proof, systematic and based on CFT-inspired techniques, to infer from this fact a potential expression for such integrals.

	\subsection{Some additional results and outlooks}
	
	\subsubsection{Representation theory of $W$-algebras}
	The previous statements provides probabilistic implications of the correspondence between GFFs and $W$-algebras. Conversely by exploiting the connection between the two objects we can derive some properties on the $W$-algebra side. For instance we show that the usual change rule for the GFF under conformal transformations $\X\to\X\circ\psi+Q\ln\norm{\psi'}$ has a natural counterpart in the setting of $W$-algebras: this is the content of Proposition~\ref{prop:mobius}. 
    
	We also show (Proposition~\ref{prop:unitary}) that the free-field representation considered in this document is unitary in the sense that if $\gamma\in (0,\sqrt 2)$, then for any $n\in\Z$ and $s$ a spin of $\g$
    \begin{equation}
        \left(\W{s}_n\right)^*=(-1)^s\W{s}_{-n}.
    \end{equation} 
    This property is actually valid as soon as the Heisenberg generators satisfy~\eqref{eq:adj_A}.
    This statement is often referred to as a \textit{duality} result in the VOA literature (see \textit{e.g.}~\cite[Section 5]{FHV} for more details about this notion). In the case of the Virasoro algebra one can use the Segal-Sugawara representation to show its validity by explicit computations. For $W$-algebras a similar statement appears in~\cite[Equation (279)]{Arakawa_rep} where the adjoint corresponds to the image by the so-called \textit{duality functor}. It is however not clear to the author of the present paper that this action of the duality functor does indeed correspond to the described action on the generators of the Heisenberg algebra. 
	
	We also describe in Theorem~\ref{thm:Verma} some new results concerning the representation theory of the $W$-algebra associated to $B_2$. To be more specific we prove that for special values of the highest-weight $\prim$, the free-field module $\Wmod$ generated by acting using the generators of the $W$-algebra on $\prim$ is not isomorphic to the corresponding Verma module $\Verma$ but embeds into a quotient of this Verma module by some proper submodule, which we describe explicitly in terms of \textit{singular vectors}. For instance we prove that (see Section~\ref{sec:last} for more details) if $\alpha=\kappa\omega_1^\vee$ for some generic $\kappa\in\R$ then we have a singular vector at level $1$, given by $\left(\frac\kappa\gamma\left(-3 \gamma ^2+\gamma  \kappa -8\right) \L_{-1} +\Wb_{-1}\right)\prim$.
	This means that we have an embedding $\Wmod\xhookrightarrow{} \bigslant{\Verma}{\mc I}$, with $\mc I$ left ideal generated from this singular vector.
	Besides we describe three singular vectors for $-\gamma\omega_1,-\frac2\gamma\omega_2^\vee$ and four singular vectors for $-\gamma\omega_2, -\frac2\gamma\omega_1^\vee$. Remarkably the above description of singular vectors features a phenomenon of duality in that the change $\gamma\leftrightarrow \frac2\gamma$
	amounts to interchanging the role of the roots of $B_2$, and thus switching $B_2$ and its \textit{Langlands dual} $C_2$. This is a particular case of a more general duality property for $W$-algebras, stating that for generic $\bm c$ there is an isomorphism between the $W$-algebra associated to $\g$ and the one associated to its Langlands dual $\prescript{L}{}{\g}$ (see \textit{e.g.}~\cite[Proposition 15.4.16]{FBZ}).
    
	\subsubsection{Marginal operators: from the free-field to Liouville and Toda CFTs}
	Toda CFTs can be thought of as perturbations of the free-field theory by \lq\lq marginal operators" $(I_i)_{1\leq i\leq r}$, informally speaking given by two-dimensional analogs of the screening operators and that would take the form $I_i=\int_{\D}V_{\gamma e_i}(x)d^2x,$ with $V_{\gamma e_i}(x)=e^{\ps{\gamma e_i,\X(x)}}$.
	The corresponding correlation functions have been rigorously defined in~\cite{Toda_construction, CH_construction} using GFF and the theory of \textit{Gaussian Multiplicative Chaos}~\cite{Kah}. 
	
	Due to the specific form of these marginal operators $I_i$, strongly reminiscent of the screening operators $Q_i$, we expect the $W$-symmetries enjoyed by the free-field theory to remain preserved when we consider Toda CFTs instead. Like in the present work this would amount to proving Ward identities for Toda correlation functions, which would in turn allow to translate results coming from representation theory of $W$-algebras into actual constraints for Toda correlation functions. And along the same lines as the derivation of the Fateev-Litvinov formula~\cite{Toda_OPEWV,Toda_correl1,Toda_correl2} in the case where $\g=\sl_3$ this would pave the way to a rigorous derivation of Toda structure constants in a more general setting. The computation of Selberg intergals arising from the Mukhin-Varchenko conjecture would then become a corollary of the exact expressions for such Toda structure constants. 
	
	\begin{itemize}
		\item We plan to explore this direction in a future work. To be more specific we first plan to address the case of $\g=B_2$ and aim to show that Ward identities and therefore~\ref{thm:B2} hold true where we consider Toda correlation functions rather than free-field ones. This would hopefully lead to the derivation of some three-point correlation functions along the method that we describe in Subsection~\ref{subsec:Mukhin_Varchenko}, and recover Selberg integrals as special cases. 
		\item Though we have described here how to proceed when $\g=B_2$, we expect this approach to work in a more general case. Indeed the method we propose is systematic and paves the way to an algorithmic derivation of $\g$-Selberg integrals arising from the Mukhin-Varchenko conjecture. For instance when $\g=B_n$ (and likewise for its Langlands dual $C_n$) we believe that four-point correlation functions containing two degenerate insertions $\alpha^*=-\gamma\omega_n$ and $\alpha_2=-\gamma\omega_1$ are also solutions of a Fuchsian differential equation of order $2n$.
		\item It is also a challenging problem to understand what is the structure of Vertex Operator Algebra that is naturally associated with Liouville and Toda CFTs. Being able to figure out in which form we could make sense of this object is fundamental in building a bridge between the two approaches, and would hopefully lead to a better understanding of objects coming from both sides, one striking example being the convergence of conformal blocks. 
	\end{itemize}

	\vspace{0.3cm}
	\textit{\textbf{Acknowledgements:}}
	The author would like to thank Juhan Aru, Colin Guillarmou, Eveliina Peltola and Vincent Vargas for their interest and support in this work. The author is also grateful to Yi-Zhi Huang for his hospitality and for sharing his knowledgeable viewpoint, to Philémon Bordereau for his helpful comments on a first version of the document, and to the anonymous referee for suggestions that helped to improve the paper. 
	Finally the author would like to acknowledge his attendance at different programs from which this work has benefited, among them: Probability in Conformal Field Theory (Bernoulli Center, EPFL), Geometry, Statistical Mechanics, and Integrability (IPAM, UCLA), Random paths to QFT: New probabilistic approaches to field theory (SCGP, Stony Brook University).
	The author has been supported by Eccellenza grant 194648 of the Swiss National Science Foundation and is a member of NCCR SwissMAP.

	

	\section{The free-field theory: a probabilistic perspective}\label{sec:GFF}
	We introduce here the setting where we will define the free-field theory of interest, and more specifically Gaussian Free Fields together with a family of observables of these random functions. Some of the material presented in this section can be found in~\cite[Section 3]{GKRV} in the case $r=1$.
	
	
	\subsection{The general setting and Gaussian Free Fields}
	
	Before presenting the framework we will rely on throughout this article, we first introduce basic notions related to conformal geometry on the (Riemann) sphere and (Gaussian) Free Fields. Hereafter we denote by $\a\simeq\R^r$ an Euclidean space of dimension $r\geq 1$, and its complexification by $\mathfrak h\coloneqq \a\otimes_\R\C\simeq\C^r$. We assume that $\a$ comes equipped with an inner product $\ps{\cdot,\cdot}$ and an orthonormal basis $(\bm v_i)_{1\leq i\leq r}$.	
	
	
	\subsubsection{Conformal geometry on the Riemann sphere}
	The two-dimensional sphere $\S^2$ is conformally equivalent, by means of stereographic projection, to the Riemann sphere $\hat\C\coloneqq \C\cup\{\infty\}$, where the latter comes equipped with the Riemannian metric $g(z)\coloneqq\norm{z}^{-4}_+|\d z|^2$ where $\norm{z}_+\coloneqq\max(\norm{z},1)$.
	As will be made clear later (see Subsection~\ref{subsec:refl_pos} below), this metric is particularly relevant from the perspective of reflection positivity (or Osterwalder-Schrader positivity), but can also be seen from the following observation. Let $\theta:\hat\C\to\hat\C$ be the reflection over the unit circle $\S^1=\T$, $\theta(z)\coloneqq \frac{1}{\bar z}\cdot$
	Then  $\theta$ sends the unit disk $\D$ to $\D^c\coloneqq\hat\C\setminus\bar\D$, and the pushforward on $\D^c$ of the flat metric $\norm{\d z}^2$ over $\D$ is precisely given by $\norm{z}^{-4}\norm{\d z}^2$. Put differently the Riemannian metric $g$ is preserved by the reflection $\theta$.
	This statement has a counterpart at the level of Green's kernels. Namely the Green's function of the Laplace operator in the metric $g$ is given for $x\neq y\in\hat\C$ by
	\begin{equation}\label{eq:green}
		G(x,y)\coloneqq \ln\frac{1}{\norm{x-y}}+\ln\norm{x}_++\ln\norm y_+.
	\end{equation}
	It is such that $G(\theta(x),\theta(y))=G(x,y)$. More generally if $\psi$ is a M\"obius transform of the sphere:
	\begin{equation}\label{eq:Green_Mobius}
		G(\psi(x),\psi(y))= G(x,y)-\frac14(\phi(x)+\phi(y)),\quad e^{\phi(z)}g(z)=g_\psi(z)\coloneqq \norm{\psi'(z)}^2\norm{\psi(z)}^{-4}_+\norm{\d z}^2.
	\end{equation}
	
	The Riemannian surface $(\hat\C,g)$ has a natural $L^2$ scalar product (with $\d v_g$ the volume form) 
	\[
	\ps{f,h}\coloneqq \int_{\C} \ps{f(x),h(x)}\d v_g(x)
	\]
	for $f,h\in C^\infty_c(\hat\C\to\a)$ the space of smooth, compactly supported maps from $\C$ to $\a\simeq\R^r$. The Sobolev space $\mathrm H^1(\hat\C,g)$ is defined as the closure of $C^\infty_c(\hat\C\to\a)$ with respect to the Hilbert-norm
	\begin{equation}\label{Hnorm}
		\int_{\C}\norm{h}^2\d v_{g}+\int_{\C}\ps{h,-\Delta_g h}\d v_g
	\end{equation}
    ($\Delta_g$ is the Laplacian in the metric $g$).
	The continuous dual of $\mathrm H^1(\hat\C,g)$ will be denoted $\mathrm H^{-1}(\hat\C,g)$.

	
	\subsubsection{Gaussian Free Fields}
	In this document, we will work with a random (generalized) function with values in the Euclidean space $\a$. It is defined based on a vectorial GFF $\X$, which is a random Gaussian distribution that satisfies the property that for any $u,v\in\a$ and $x\neq y\in\C$:
	\begin{equation}\label{eq:cov_GFF_plane}
		\expect{\ps{u,\X(x)}\ps{v,\X(y)}}=\ps{u,v}G(x,y)
	\end{equation}
	where $G$ is the Green kernel introduced in Equation~\eqref{eq:green}. It is a standard fact~\cite{dubedat,She07} that $\X$ almost surely belongs to the distributional space $\mathrm H^{-1}(\hat\C,g)$.	By definition it satisfies $\X\eqlaw \X\circ\theta$. Besides for any M\"obius transform of the sphere $\psi$ (see~\cite[Proposition 2.3]{DKRV})
	\begin{equation}\label{eq:GFF_Mobius}
		\X\circ\psi\eqlaw\X-m_{g_\psi}(\X),\quad m_{g_\psi}(\X)\coloneqq \int_\C \X(x)\d v_{g_\psi}(x).
	\end{equation} 
	 
	Alternatively this GFF can be realized via the decomposition
	\begin{equation}\label{eq:GFF_decomp}
        \X= P\varphi+\X_{\D}+\X_{\D^c}
	\end{equation}
	where $P\varphi$ denotes the harmonic extension to $\C$ of a free field $\varphi:\T\to\a$, while $\X_{\D}$ and $\X_{\D^c}$ are two independent (vectorial) GFFs, taking values respectively in $\D$ and $\D^c$, and subject to Dirichlet boundary conditions. This means that $\X_{\D}$ has covariance kernel given by 
	\begin{equation}\label{eq:cov_GFF_disc}
		\expect{\ps{u,\X_\D(x)}\ps{v,\X_\D(y)}}=\ps{u,v}G_{\D}(x,y),\quad G_{\D}(x,y)\coloneqq\ln\frac{\norm{1-x\bar y}}{\norm{x-y}}
	\end{equation}
	for any $u,v$ in $\a$, while $\X_{\D^c}\eqlaw \X_{\D}\circ\theta$. By setting $\bm e_{n,j}(e^{\i\theta})\coloneqq e^{\i n\theta}\bm v_j$, the field $\varphi$ is given by
	\begin{equation}\label{eq:cov_GFF_circle}
		\varphi\coloneqq \sum_{n\in\mathbb{Z}^*}\sum_{i=1}^r \varphi_{n,i}\bm e_{n,i}.
    \end{equation}
	Here for positive $n$ the modes $\varphi_{n,i}$, $i=1,\cdots,r$, are centered, complex independent Gaussian variables with variance $\frac{1}{2n}$, and with $\varphi_{n,i}=\bar\varphi_{-n,i}$. More explicitly, by writing $\varphi_{n,j}\coloneqq \frac{x_{n,j}+i y_{n,j}}{2\sqrt n}$, these modes are distributed according to the probability measure
	\begin{equation}
		\P_\T\coloneqq \bigotimes_{\substack{n\geq1,\\j=1,\cdots,r}}\frac1{2\pi}e^{-\frac12(x_{n,j}^2+y_{n,j}^2)}\d x_{n,j}\d y_{n,j}
	\end{equation} 
	over the space $\Omega_{\T}\coloneqq \left(\R^{2r}\right)^{\N^*}$, equipped with a cylinder sigma-algebra $\Sigma_{\T}\coloneqq \mathcal{B}^{\otimes \N^*}$.
	With these notations the harmonic extension of $\varphi$ over the complex plane $\C$ takes the form 
		\begin{equation}\label{eq:harm_ext}
			P\varphi(z)\coloneqq \sum_{n\geq 1}\sum_{i=1}^r \left(\varphi_{n,i}z^n+\bar\varphi_{n,i}\bar z^n\right)\bm v_{i}.
		\end{equation}
		This field has zero mean over the circle: we will often shift $\varphi$ by a constant $c\in\a$ in the sequel.
		
		Hereafter we will assume that such GFFs live in a probability space $(\Omega,\Sigma,\P)$. The above decompositions translate as $\Omega=\Omega_\T\times\Omega_\D\times\Omega_{\D^c}$ (and likewise for $\Sigma$ and $\P$). We will denote by $\ephi{\cdot}$ the conditional expectation with respect to $\varphi$, meaning that the randomness comes from $\X_\D$, $\X_{\D^c}$.
		
		
		\subsubsection{Additional tools}
		A key property of these GFFs is \textit{Gaussian integration by parts}. Namely for a centered Gaussian vector $(X,Y_1,\dots,Y_N)$ and $f$ smooth over $\R^N$ with bounded derivatives:
		\begin{equation}\label{eq:Gauss_IPP}
			\E\left[Xf(Y_1,\dots,Y_N)\right]=\sum_{k=1}^N\E\left[XY_k\right]\E\left[\partial_{Y_k}f(Y_1,\dots,Y_N)\right].
		\end{equation}
		From this statement one can infer the following equality for the GFF $\X$
		\begin{equation}\label{eq:IPP}
			\expect{\ps{u,\X(z)}e^{\ps{\X,f}_\D}}=\int_\D G(z,x)\ps{u,f(x)}\expect{e^{\ps{\X,f}_\D}}\d^2x
		\end{equation}
		for $f\in C_0^\infty\left(\D\to\a\right)$, $z\in\D$ and $u\in\a$, as well as its counterpart statement for $\X_\D$:
		\begin{equation}\label{eq:IPP_D}
			\ephi{\ps{u,\X_\D(z)}e^{\ps{\X,f}_\D}}=\int_\D G_\D(z,x)\ps{u,f(x)}\ephi{e^{\ps{\X,f}_\D}}\d^2 x,
		\end{equation}
		where we have set $\ps{f,h}_\D\coloneqq \int_\D\ps{f(z),h(z)}\d^2z$.
		These statements can be recovered thanks to the Girsanov (or Cameron-Martin) theorem, see \emph{e.g.}~\cite[Chapter VIII]{RY91}:
		\begin{thmcite}\label{thm:girsanov}
			Let $D$ be a subdomain of $\C$, $(\bm X(x))_{x\in D}\coloneqq (X_1(x),\cdots,X_{n-1}(x))_{x\in D}$
			 a family of smooth centered Gaussian fields, and $Z$ a Gaussian variable belonging to the $L^2$ closure of the subspace spanned by $(\bm X(x))_{x\in D}$. Then, for $F$ bounded over the space of continuous functions,
			\[
			\expect{e^{Z-\frac{\expect{Z^2}}{2}}F(\bm X(x))_{x\in D}}=\expect{F\left(\bm X(x)+\expect{Z\bm X(x)}\right)_{x\in D}}.
			\]
		\end{thmcite}

		
		\subsection{Hilbert space of states}\label{subsec:hil}
		Using this probabilistic framework we can define the space of states of the free-field theory, the latter being viewed as a map from the unit circle $\T$ to $\a\simeq\R^r$. We follow here the approach developed in~\cite[Section 3]{GKRV} when $r=1$.
		
		
		\subsubsection{The Hilbert space of states}
		The Hilbert space of states of the free-field theory can be realized based on a space of maps $\S^1=\T\to\a$. Namely, via Equation~\eqref{eq:cov_GFF_circle}, functionals of such a field $\varphi$ can be encoded by functions $F=F\left((\varphi_{n,i})_{\substack{n\in\Z\\ i=1,\cdots, r}}\right)$ of its modes. The space of states is then the Hilbert space $\mc H_0\coloneqq L^2(\Omega_\T,\P_\T)$, the latter being equipped with the Hermitian inner product
		\begin{equation}
			\ps{F\vert H}_{0,2}=\expect{F(\varphi)\overline{H(\varphi)}}.
		\end{equation}
		Here $\E$ stands for the probability measure $\P_\T$ introduced above. 
		
		As mentioned before, this space does not take into account the fact that the free field $\varphi$ has zero mean over the circle. To remedy this issue we will also consider the Hilbert space $\Hil$ corresponding to the $L^2$ space of $\a\times\Omega_{\T}$ when equipped with the measure $\mu\coloneqq dc\otimes\P_\T$, where $\d c$ is the Lebesgue measure on $\a\simeq\R^r$ (which is \emph{not} a probability measure). The associated scalar product is $\ps{\cdot\vert\cdot}_2$:
		\begin{equation}
			\ps{F\vert G}_2=\int_{\a}\expect{F(\varphi+c)\overline{G(\varphi+c)}}\d c\coloneqq\int_{\R^r}\expect{F\overline G\left(\varphi+\sum_{i=1}^rc_i\bm v_i\right)}\d c_1\cdots\d c_r.
		\end{equation}
		
		\subsubsection{Structure of graded space}
		The Hilbert space $\mc H_0$ naturally admits a structure of graded space. Namely let us set for any integers $p$ and $q$:
		\begin{equation}
			\V_{p,q}\coloneqq \bigoplus_{\substack{n_1+\cdots+n_k=p\\ m_1+\cdots +m_l=q}}\left\{F:\varphi\mapsto \prod_{i=1}^k\ps{u_i,\varphi_{n_i}}\prod_{j=1}^l\ps{v_j,\varphi_{-m_j}},\quad \substack{u_1,\cdots, u_k\in\h\\v_1,\cdots, v_l\in\h}\right\}
		\end{equation}
		where the sum ranges over tuples of (strictly) positive integers. Then we have the equality
		\begin{equation}
			\mc H_0=\overline{\V},\qt{where}\V\coloneqq\bigoplus_{p,q\geq0}\V_{p,q}
		\end{equation}
		and with the closure understood as the Hilbert space closure, with respect to the scalar product induced by $\P_\T$, of the algebraic direct sum (see \textit{e.g.}~\cite[Theorem 3.29]{Janson}).
		
		Since our model features two chiralities (that is depends on both holomorphic and anti-holomorphic coordinates) it is natural to introduce the corresponding graded spaces $\V_+$ and $\V_-$ by setting
		\begin{equation}\label{eq:def_V+}
			\V_+\coloneqq\bigoplus_{p\geq0}\V_+^{(p)}\qt{and}\V_-\coloneqq\bigoplus_{q\geq0}\V_-^{(q)},
		\end{equation}
		where $\V_+^{(p)}\coloneqq \V_{p,0}$ and $\V_-^{(q)}\coloneqq\V_{0,q}$. More specifically
		\begin{equation}
			\V_+^{(p)}\coloneqq \bigoplus_{\substack{n_1,\cdots,n_k>0\\n_1+\cdots+n_k=p}}\left\{F:\varphi\mapsto \prod_{i=1}^k\ps{u_i,\varphi_{n_i}},\quad u_1,\cdots, u_k\in\h\right\}.
		\end{equation}

		\subsubsection{The Hilbert space of field observables}\label{subsec:refl_pos}
		We now identify the Hilbert space $\Hil$ with a space of functionals of a field $\X:\D\to\a$. Let $\mathcal{A}_\D$ be the sigma-algebra on $\a\times\Omega$ generated by maps of the form $\X\mapsto\ps{\X,f}_\D$ for $f\in C_0^\infty(\D\to\a)$. Let $\mathcal{F}_\D$ be the set of $\C$-valued, $\mathcal{A}_\D$-measurable functions, equipped with the sesquilinear form
		\begin{equation}\label{eq:def_bil0}
			(F,G)_\D\coloneqq \ps{\theta F \bar G}_{\gamma},\quad\text{where }\ps{F}_{\gamma}\coloneqq \int_{\a}e^{-2\ps{Q,c}}\expect{F\left(\X+c-2 Q\ln \norm{\cdot}_+\right)}\d c.
		\end{equation}
		In the above, the action of $\theta$ on $\mathcal F_\D$ has been defined by $\theta F(\phi)\coloneqq F\left(\phi\circ\theta-2Q\ln\norm{\cdot}\right)$.  The \textit{background charge} $Q$ is an element of $\a$ that will be fixed later on and that is defined from the \textit{coupling constant} $\gamma\in\C$\footnote{The role of this extra parameter will be made clear later on. As we will see it allows to extend the correspondence between free fields and vertex operator algebras for any value of the central charge.}.
		The decomposition~\eqref{eq:GFF_decomp} of the GFF gives
		\begin{equation}\label{eq:bil_phi}
			\begin{split}
				(F,G)_\D= \int_{\a}e^{-2\ps{Q,c}}\expect{\ephi{F\left(\X_1+P\varphi+c\right)}\ephi{\overline{G\left(\X_2+P\varphi+c\right)}}}\d c
			\end{split}
		\end{equation}
		where $\X_1$ and $\X_2$ are two independent GFFs with same law as $\X_\D$. The associated seminorm is denoted $\nnorm{\cdot}_\D$, so that $\nnorm{F}_\D^2= (\cdot,\cdot)_\D$. Using the Cauchy-Schwartz inequality, it satisfies
		\begin{equation}\label{eq:ineq_semi1}
			\nnorm{F}^2_\D \leq \ps{F^2}_{\gamma}.
		\end{equation}
		We denote by $\mathcal F_\D^{2}$ the subspace of $\mathcal F_\D$ made of functionals $F$ for which  $\nnorm{F}_2<\infty$.
		
		The sesquilinear form defined via Equation~\eqref{eq:def_bil0} is not positive definite over $\mathcal F_\D$. To remedy this issue, let $\mathcal{N}\coloneqq \left\lbrace F\in\mathcal{F}_\D^{2},\quad (F,F)_\D=0\right\rbrace$ be the \textit{null space}. 
		The Hilbert space of field observables $\Hild$ is then defined as the completion with respect to $(\cdot,\cdot)_\D$ of the quotient space $\mathcal{F}^{2}_\D/\mathcal{N}$. 
		
		\subsubsection{Identification between the two Hilbert spaces}
		As we now show these two Hilbert spaces $\Hild$ and $\Hil$ can actually be identified. For this purpose we define a map $U_0:\mc F_\D^{2}\to\Hil$ by setting
		\begin{equation}
			U_0F (c,\varphi) \coloneqq   e^{-\ps{Q,c}}\E_\varphi\left[F(\X_\D+P\varphi+c)\right].
		\end{equation}
		The connection between $\Hild$ and $\Hil$ can now be properly stated:	
		\begin{proposition}\label{prop:refl_pos0}
			The map $U_0:\mc F_\D^{2}\to\Hil$ descends to a unitary map from $\Hild$ to $\Hil$.
		\end{proposition}
		Non-negativity of the sesquilinear form is often referred to as \emph{reflection positivity}. The map $U_0$ was introduced in~\cite[Equation (3.19)]{GKRV} where a similar statement~\cite[Proposition 3.1]{GKRV} was proved.
		\begin{proof}
            For the sake of completeness we reproduce the arguments in~\cite[Proposition 3.1]{GKRV}.
			By Equation~\eqref{eq:bil_phi}, over $\mathcal F_\D^{2}$ we have $(F,G)_\D=\ps{U_0F\vert U_0G}_2$. In particular the sesquilinear form $(\cdot,\cdot)_\D$ is non-negative and descends to an isometry on $\Hild$. Hence proving that it is unitary boils down to the statement that it is onto, and for this it is enough to find a subset of $\Lro$ whose linear span is dense in $\Lro$ and that lies in the image of $U_0$. This is achieved by considering elements of $\Lro$ of the form $e^{-\ps{Q,c}}\rho(c)e^{\ps{\varphi,h}_\T}$ where $\rho\in C_c^\infty(\a)$ and $h\in C^\infty(\T\to\a)$ has zero mean over $\T$.
			Indeed the linear span of this set is dense in $\Lro$, and besides if
			\[
			F_\eps(\X+c)\coloneqq \rho\left(\ps{\X,g_\eps}_\D+c\right)e^{\ps{\X+c,f_\eps}_\D-\frac12\ps{f_\eps,G_\D f_\eps}}\text{ with }g_\eps(z)\coloneqq \eps^{-1}\eta\left(\frac{1-\norm z}{\eps}\right)\text{ and }f_\eps=(Ph)g_\eps
			\]
			with $\eta$ a smooth mollifier with support in $[1,2]$, then $\lim\limits_{\eps\to 0}U_0(F_\eps)(c,\varphi)=e^{-\ps{Q,c}}\rho(c)e^{\ps{\varphi,h}_\T}$ where the limit is taken in $\Lro$. Indeed
			\begin{align*}
				U_0(F_\eps)(c,\varphi)&=e^{-\ps{Q,c}}e^{\ps{P\varphi+c,f_\eps}_\D}\ephi{\rho\left(\ps{\X_\D+P\varphi,g_\eps}_\D+c\right)e^{\ps{\X_\D,f_\eps}_\D-\frac12\ps{f_\eps,G_\D f_\eps}}}\\
				&=e^{-\ps{Q,c}}e^{\ps{P\varphi,f_\eps}_\D}\ephi{\rho\left(\ps{\X_\D,g_\eps}_\D+c\right)}
			\end{align*}
			where we have used the fact that $\ps{c,f_\eps}_\D=\ps{P\varphi,g_\eps}_\D=0$ (since both $f_\eps$ and $P\varphi$ have zero mean on circles), together with independence of the Gaussian random variables $\ps{\X_\D,g_\eps}_\D$ and $\ps{\X_\D,f_\eps}_\D$. Like in~\cite[Proposition 3.1]{GKRV} the latter converges in $\Lro$ to $e^{-\ps{Q,c}}\rho(c)e^{\ps{\varphi,h}_\T}$.
		\end{proof}
		
		
		\subsection{Some field observables}\label{subsec:field_obs}
		We describe some specific functionals that we will use in the sequel to understand properties of the free-field theory: derivatives and exponentials of the free-field.
		
		\subsubsection{Wick products and derivatives of the free-field}
		We have identified the Hilbert space $\Hil$ with a space of functionals of a field on the disk $\Hild$. We provide here an explicit equality relating the elements of the graded space $\V$ and such functionals using Wick products. For this purpose we write, inside $\D$, $\X=\X_{\D}+P\varphi+c$.	
        To overcome the lack of regularity of $\X_\D$, we define a smooth approximation of $\X$ by setting, with $\eta_\eps\coloneqq \frac1{\eps^2}\eta(\frac{\cdot}{\eps})$ a smooth, compactly supported mollifier,
		\begin{equation}\label{eq:regularization}
			\X_\eps\coloneqq \X*\eta_\eps=\int_\C \X(\cdot-z)\eta_\eps(z)\d^2z.
		\end{equation}		
        
		We now pick any positive integers $n_1,\cdots,n_p$ and $u_1,\cdots,u_p\in\h$. For $z\in\D$ and $\eps>0$, let us denote $\xi^\eps_i\coloneqq \ps{u_i,\partial^{n_i}\X_{\D,\eps}(z)}$ where $\X_{\D,\eps}=\X_\D*\eta_\eps$. The Wick product is defined by setting
		\begin{equation}
			:\partial^{n_1}\ps{u_1,\X(z)}\cdots \partial^{n_p}\ps{u_p,\X(z)}:\quad\coloneqq \lim\limits_{\eps\to 0}\sum_{\Gamma }(-1)^{\norm{\Gamma_2}}\prod_{(i,j)\in\Gamma_2}\expect{\xi_i^\eps\xi_j^\eps}\prod_{i\in\Gamma_1}\left(\xi_i^\eps+\partial^{n_i}P\varphi(z)\right)
		\end{equation}
		where the sum ranges over Feynman diagrams, that is to say partitions $\Gamma=\Gamma_1\cup\Gamma_2$ of $\{1,\cdots,p\}$ such that $\Gamma_1$ contains only singletons while $\Gamma_2$ contains pairs $(i,j)$. Then:
		\begin{proposition}\label{prop:wick}
			For any positive integers $n_1,\cdots,n_p$ and $u_1,\cdots,u_p$ in $\h$ we have the equality
			\begin{equation}
				\ps{u_1,\varphi_{n_1}}\cdots \ps{u_p,\varphi_{n_p}}=\ephi{:\frac{\partial^{n_1}\ps{u_1,\X(0)}}{n_1!}\cdots\frac{\partial^{n_p}\ps{u_p,\X(0)}}{n_p!}:}.
			\end{equation}
		\end{proposition}
	   Recall that the expectation is conditional with respect to $\varphi$ (randomness comes from $\X_\D$). 
		\begin{proof}
			A manipulation of the partitioning of $\{1,\cdots,p\}$ allows to show that:
			\begin{align*}
				&:\partial^{n_1}\ps{u_1,\X_\eps(0)}\cdots \partial^{n_p}\ps{u_p,\X(0)}:\quad=\sum_{\Gamma_1\subset\{1,\cdots,p\}}\prod_{k\not\in\Gamma_1}\ps{u_k,n_k!\varphi_{k}}  :\prod_{l\in\Gamma_1}\xi_{l}^\eps:.
			\end{align*}
			Since as soon as $\Gamma_1$ is non-empty the corresponding term in the sum is $\ephi{:\prod_{l\in\Gamma_1}\xi_{l}^\eps:}=0$ while for $\Gamma_1=\emptyset$ it is given by $\ps{u_1,\varphi_{n_1}}\cdots \ps{u_p,\varphi_{n_p}}$ we can conclude. 
		\end{proof}
		This statement gives an explicit identification between $\V_+$ and a space of functionals defined by taking holomorphic derivatives of the field at the origin.

		For future reference we also have the following property: take $p_1,\cdots,p_m\geq 2$ and $z_{1},\cdots,z_{m}$ distinct.  For $1\leq k\leq p_1+\cdots+ p_m$, set $\xi_{k}^\eps\coloneqq \ps{u_{k},\partial^{n_{k}}\X_{\D,\eps}(x_k)}$ for some $u_k\in\h$ and where $x_{k}=z_l$ for $l$ such that $\sum_{i=1}^{l-1}p_i<k\leq \sum_{i=1}^lp_i$. Then~\cite[Theorem 3.15]{Janson}
		\begin{equation}\label{eq:wick_gen}
			\begin{split}
				&\prod_{l=1}^m:\prod_{k=1+\sum_{i=1}^{l-1}p_i}^{\sum_{i=1}^{l}p_i}\xi^\eps_{k}:=\sum_{\Gamma}:v(\Gamma):\qt{where}:v(\Gamma):\quad=\prod_{(i,j)\in\Gamma_2}(-1)^{\norm{\Gamma_2}}\expect{\xi^\eps_{i}\xi^\eps_{j}}:\prod_{i\in\Gamma_1}\xi^\eps_{k}:.
			\end{split}
		\end{equation}
		Here the sum  ranges over all Feynman diagrams whose pairs $(i,j)$ are such that $i$ and $j$ do not belong to the same interval $\{p_l+1,\cdots,p_{l+1}\}$.
		
		
		\subsubsection{Vertex Operators and exponentials of the free-field}\label{subsec:correl_GFF}
		We now turn to \textit{Vertex Operators}. These functionals of the free-field depend on a point $z\in\C$ and a weight $\alpha\in\a$. They are formally
        \[
		V_\alpha(z)\left[\phi\right]\coloneqq e^{\ps{\alpha,\phi(z)}}.
		\]
		
		The Vertex Operators are then defined, with  the \textit{conformal weights} $\Delta_\alpha\coloneqq \ps{\frac\alpha2,Q-\frac\alpha2}$, by
		\begin{equation}\label{eq:def_V_eps}
			V_{\alpha,\eps}(z)\coloneqq \norm{z}_+^{-4\Delta_\alpha}e^{\ps{\alpha,\X_\eps(z)+\bm c}-\frac{\expect{\ps{\alpha,\X_\eps(z)}^2}}2}
		\end{equation}
		where the prefactor $\norm{z}_+^{-4\Delta_\alpha}$ accounts for the underlying metric on $\hat\C$.		
		We then define (these quantities play the role of the correlation functions in the interacting theory)
		\begin{equation}\label{eq:def_correl_GFF}
			\expect{\prod_{k=1}^NV_{\alpha_k}(z_k)} \coloneqq \lim\limits_{\eps\to0} \expect{ \prod_{k=1}^NV_{\alpha_k,\eps}(z_k)}
		\end{equation}
		where $z_1,\cdots,z_N\in\C$ while $\alpha_1,\cdots,\alpha_N\in\a$. Gaussian computations then give
		\[
		\expect{\prod_{k=1}^NV_{\alpha_k}(z_k)}=\prod_{k=1}^N\norm{z_k}_+^{-4\Delta_{\alpha_k}}\prod_{k<l}e^{\ps{\alpha_k,\alpha_l}G(z_k,z_l)}.
		\]

		
		
		\section{The free-field theory: an algebraic viewpoint}\label{sec:VOA}
		We have presented in the previous section the probabilistic setting where we define the free-field theory. Here we describe a more algebraic view on the latter and establish a correspondence between these two perspectives. To this end we first provide a Fock representation of the rank $r$ Heisenberg algebra 
        acting on the Hilbert space of states $\Hil$, which we then promote to a structure of vertex algebra and relate to the probabilistic framework from the previous section. 

		\subsection{A representation of the Heisenberg algebra}
		Having defined the Hilbert state associated to the free-field theory, we now introduce a representation of the Heisenberg algebra whose elements are viewed as operators acting on (a subspace of) this Hilbert space.
		
		\subsubsection{Heisenberg algebra and the Hilbert space}
        Here we recall and extend some of the content of~\cite[Section 4]{GKRV}.
		Let $\mathcal{S}$ be the subset of $L^2(\Omega_\T)$ defined as the linear span of smooth maps that depend only on a finite number of coordinates $\varphi_{n,i}$. Further define a dense subspace $\vesp$ of $\Hil$ by
		\begin{equation}\label{eq:def_hilb}
			\vesp\coloneqq\text{Span}\left\lbrace (c,\varphi)\mapsto\psi(c)F(\varphi)\in\mc H_\T,\quad\psi\in C^\infty(\a)\text{ and }F\in\mathcal{S}\right\rbrace.
		\end{equation} 
		We now introduce a family $(\partial_{n,i})_{n\in\Z,1\leq i\leq r}$ of differential operators acting on $\endv$ by setting
		\begin{equation}
			\partial_{n,i}F\left((\varphi_{m,j})_{\substack{\norm{m}\leq N \\ j=1,\cdots,r}}\right)\coloneqq \frac{\partial F}{\partial \varphi_{n,i}}\left((\varphi_{m,j})_{\substack{\norm{m}\leq N \\ j=1,\cdots,r}}\right).
		\end{equation}
		Using these operators we can form  $\partial_nF\coloneqq \sum_{i=1}^r\frac{\partial F}{\partial\varphi_{n,i}}\bm v_i$ and likewise set $\varphi_n\coloneqq\sum_{i=1}^r\varphi_{n,i}\bm v_i$	so that $\partial_n$ is the gradient with respect to the variable $\varphi\in\a\simeq\R^r$.
		Besides $\partial_{n,i}F=\ps{\partial_n F,\bm v_i}$.
		
		We now present the creation and annihilation operators that will allow to define a Fock representation of the Heisenberg algebra. 
		They are defined by setting for any $i=1,\cdots,r$:
		\begin{equation}\label{eq:def_An}
			\begin{split}
				&\A_{n,i}\coloneqq \frac{\i}{2}\partial_{n,i}\quad\text{for positive } n\quad\text{(the annihilation operators)}\\
				&\A_{0,i}\coloneqq \frac{\i}{2}\left(\partial_{0,i}+\ps{Q,\bm v_i}\right)\quad\text{for } n=0\\
				&\A_{n,i}\coloneqq \frac{\i}{2}\left(\partial_{n,i}+2n\varphi_{-n,i}\right)\quad\text{for negative } n\quad\text{(the creation operators)}.
			\end{split}
		\end{equation}
        These are unbounded and densely defined operators. Moreover, their adjoint with respect to the $\Lro$ Hermitian inner product are given by (well-defined since the $\A_{n,i}$ are densely defined)
		\begin{equation}\label{eq:adj_A}
			\A_{n,i}^*=\A_{-n,i}\text{ for }n> 0\quad\text{and}\quad\A_{0,i}^*=\A_{0,i}-\i\ps{\mathfrak{Re}(Q),\bm v_i}.
		\end{equation}
        In particular the operators $\A_{n,i}$ are closable. More algebraically, these operators form $r$ independent copies of the Heisenberg algebra in that for $1\leq i,j\leq r$ and $n,m\in\Z$\footnote{The commutation relations differ by a factor $2$ from the usual ones for the Heisenberg algebra. This is to conform with the notations used in~\cite{GKRV} and subsequent works.}:
		\begin{equation}\label{eq:comm_he}
			\left[\A_{n,i},\A_{m,j}\right]=\frac{n}{2}\delta_{n,-m}\delta_{i,j}.
		\end{equation}
		Like before we will also be interested in the vectorial version of these operators, defined via
		\begin{equation}
			\A_{n}=\sum_{i=1}^r\A_{n,i}\bm v_i.
		\end{equation}
		It is such that for any $u$ and $v$ in $\a$, $\left[\ps{u,\A_{n}},\ps{v,\A_{m}}\right]=\ps{u,v}\frac{n}{2}\delta_{n,-m}$.
		
		We can likewise define other Fock representations of the Heisenberg algebra by taking \textit{anti-holomorphic} derivatives rather than holomorphic ones, that is we set $\tilde \A_{0,i}\coloneqq \A_{0,i}$ while for $n\geq 1$:
		\begin{equation}
				\tilde\A_{n,i}\coloneqq \frac{\i}{2}\partial_{-n,i},\quad \tilde\A_{n,i}\coloneqq \frac{\i}{2}\left(\partial_{-n,i}+2n\varphi_{n,i}\right).
		\end{equation}
		By doing so we have the same commutation relations as in Equation~\eqref{eq:comm_he}; besides the two families of Heisenberg algebras commute in the sense that 
		\begin{equation}
			[\A_{n,i},\tilde\A_{m,j}]=0\quad\text{for all}\quad n,m\in\Z\text{ and }1\leq i,j\leq r.
		\end{equation}
		
		\subsubsection{Fock representation of the Heisenberg algebra}
		We now wish to define a set of operators acting on (a subset of) $\Hil$ and that would play the role of counterpart, in this algebraic setting, of the space $\V$. For this purpose we introduce for non-negative integers $p,q$:
		\begin{equation}
			\mc O_{p,q}\coloneqq\bigoplus_{\substack{n_1+\cdots+n_k=p\\m_1+\cdots+m_l=q}}\left\{\prod_{i=1}^k\ps{u_i,\A_{-n_i}}\prod_{j=1}^l\ps{v_j,\tilde\A_{-m_j}},\quad \substack{u_1,\cdots, u_k\in\h\\v_1,\cdots, v_l\in\h}\right\}.
		\end{equation}
		This allows to define a graded space of endomorphisms of $\mc C_\infty$ by setting $\mc O\coloneqq\bigoplus_{p,q\geq0}\mc O_{p,q}$.
		It is naturally identified with $\C\left[\A_{n,i}\right]_{\substack{n<0\\1\leq i\leq r}}$, the Fock representation of the Heisenberg algebra~\cite[Chapter 2.1.3]{FBZ} (see also Appendix~\ref{appendix:VOA}). It is naturally identified with the space of states $\V$:
		\begin{proposition}
			There is an isomorphism between the vector spaces $\mc O_{p,q}$ and $\V_{p,q}$ given by
			\begin{equation}
				\prod_{i=1}^k\ps{u_i,\A_{-n_i}}\prod_{j=1}^l\ps{v_j,\tilde\A_{-m_j}}\mapsto\prod_{i=1}^k\ps{u_i,\A_{-n_i}}\prod_{j=1}^l\ps{v_j,\tilde\A_{-m_j}}\mathds 1.
			\end{equation}
			This map lifts to an isomorphism of graded spaces between $\mc O$ and $\V$. 
		\end{proposition}
		\begin{proof}
			This follows from the fact $\prod_{i=1}^k\ps{u_i,n_i\varphi_{n_i}}=\prod_{i=1}^k\ps{u_i,\i\A_{-n_i}}\mathds 1$.
		\end{proof}
		In the same spirit we can identify a space of \textit{chiral operators} $\mc O_+\mc O_+\coloneqq\bigoplus_{p\geq0}\mc O_{p,0}$ with the space $\V_+$ of chiral states.
		This is reminiscent of the so-called \textit{state-operator correspondence} advocated in physics, though we will make this more precise later on (see Subsection~\ref{subsec:st_op}).
		
		\subsubsection{Normally ordered product}
		The annihilation operators do not actually appear in the state-operator correspondence as stated above; however they play a crucial role in the vertex algebra structure. They appear through the analog of the Wick product in the space of operators, which correspond here to the notion of \textit{normally ordered product}.
		It is defined by setting
		\begin{equation}\label{eq:normal_prod}
			:\ps{u_1,\A_{n_1}}\cdots\ps{u_p,\A_{n_p}}:\quad \coloneqq \prod_{n_i<0}\ps{u_i,\A_{n_i}}\prod_{n_i\geq0}\ps{u_i,\A_{n_i}}.
		\end{equation}
		Put differently we put the annihilation operators on the right and the creation ones on the left. There is no ambiguity since operators with same sign commute. 
		
		Remarkably enough, the normally ordered product for the Heisenberg algebra satisfies the same type of recursive formulas as Wick products~\eqref{eq:wick_gen}. 
		Namely for $p\geq 2$ (see \textit{e.g.}~\cite[Theorem 3.3]{Kac_VOA})
		\begin{align*}
			&:\prod_{k=1}^p\ps{u_k,\A_{m_k}}::\prod_{k=p+1}^{p+q}\ps{u_k,\A_{m_k}}:=\sum_{\Gamma}\prod_{(i,j)\in\Gamma_2}(-1)^{\mathds{1}_{m_i\geq0}}\Big[\ps{u_i,\A_{m_i}},\ps{u_j,\A_{m_j}}\Big]:\prod_{i\in\Gamma_1}\ps{u_k,\A_{m_k}}:.
		\end{align*} 
		More generally a proper (and immediate) generalization of~\cite[Theorem 3.3]{Kac_VOA} yields:
		\begin{equation}\label{eq:normal_ord_gen}
			\begin{split}
				&:\prod_{k=1}^{p_1}\ps{u_k,\A_{m_k}}:\cdots :\prod_{k=p_l+1}^{p_{l+1}}\ps{u_k,\A_{m_k}}:=\sum_{\Gamma}:v(\Gamma):\qt{where}\\
				&v(\Gamma)=\prod_{(i,j)\in\Gamma_2}(-1)^{\mathds{1}_{m_i\geq0}}\Big[\ps{u_i,\A_{m_i}},\ps{u_j,\A_{m_j}}\Big]:\prod_{i\in\Gamma_1}\ps{u_k,\A_{m_k}}:,
			\end{split}
		\end{equation}
		the sum ranging over Feynman diagrams with pairs not belonging to the same interval $(p_i+1,p_{i+1})$.		
		
		\subsection{The Heisenberg vertex algebra from the free-field}\label{subsec:st_op}
		We now enrich the above representation of the Heisenberg algebra to define a structure of vertex algebra. For convenience, some elementary notions about vertex algebras are gathered in Appendix~\ref{appendix:VOA}.
        
		\subsubsection{A first take at the state-operator correspondence}
		The initial step in the state-operator correspondence is the definition of the counterpart in $\endv$ of the field $\varphi(z)$, where $\mc C_\infty$ has been defined in Equation~\eqref{eq:def_hilb}. Recall the harmonic extension~\eqref{eq:harm_ext} to $\C$ of $\varphi\in L^2(\T\to\a)$:
		\begin{equation}\label{eq:varphi_mode}
			\varphi(z)=c+\varphi_+(z)+\varphi_-(\bar z),\qt{with}\varphi_+\coloneqq\sum_{n> 0} \varphi_nz^n\qt{and}\varphi_-(\bar z)\coloneqq\sum_{n\geq 0} \bar\varphi_n\bar z^n.
		\end{equation}
		In view of this expansion we define for $z\in\C\setminus\{0\}$ the formal power series in $\endv[[z,z^{-1},\log\norm{z}]]$:
		\begin{equation}\label{eq:A_mode}
			\Phi(z)\coloneqq c+\i\A(z)+\i\bar\A(\bar z),\quad\A(z)\coloneqq \A_0\log z-\sum_{n\in\Z^*} \frac1n\A_nz^{-n}.
		\end{equation}
		In this section and unless specified, we will mostly consider $\Phi$ and its complex derivatives, which are single-valued. As such we don't need to specify a branch cut for the logarithm.
	    Moreover, for any $F\in\mc C_\infty$ the sum that appears in $\Phi(z)F$ is actually finite, and defines an element in $\mc C_\infty$. As such, $\Phi(z)$ can also be thought of as an unbounded, densely defined and closable operator.
		
		It is often more natural to consider the holomorphic derivative $J(z)$ of $\Phi$:
		\begin{equation}
			J(z)\coloneqq\partial\Phi(z)=\i\sum_{n\in\Z}\A_nz^{-n-1}\qt{and}\bar J(z)\coloneqq\bar\partial\Phi(z)=\i\sum_{n\in\Z}\tilde\A_n\bar z^{-n-1}.
		\end{equation} 
		Like before $J(z)\in\endv$ for any $z\in\C$. It can also be viewed as the generating series of the generators of the Heisenberg algebra. Up to a multiplicative constant $J$ corresponds to the current $b$ used to define the Heisenberg vertex algebra~\cite[Chapter 2.2]{FBZ}. It is also sometimes referred to as \textit{chiral bosonic field} in the physics literature~\cite{yellow_book} or as the current generating the $\hat u(1)$ symmetry.
		
		\subsubsection{The Heisenberg vertex algebra}
		We now turn to the definition of a vertex algebra structure. Due to the fact that our model has two chiralities we will define two representations of the Heisenberg vertex algebra based on the two vector spaces $\V_+$ and $\V_-$ defined in Equation~\eqref{eq:def_V+}.
		Moreover in order to take into account the dependence in the central charge $\bm c$ (which will be defined in terms of the background charge $Q$ in the next section) we first introduce
		\begin{equation}
			\V_{+,\bm c}\coloneqq e^{-\ps{Q,c}}\V_+\qt{as well as}\V_{-,\bm c}\coloneqq e^{-\ps{Q,c}}\V_-
		\end{equation} 
		where with a slight abuse of notation we identify $e^{-\ps{Q,c}}F(\varphi)$ with the map $(c,\varphi)\mapsto e^{-\ps{Q,c}}F(\varphi)$. We also set for any positive integer $p$, $\bm u=(u_1,\cdots,u_p)\in\h^p$ and $\bm n=(n_1,\cdots,n_p)\in\left(\Z^r\right)^p$
        \begin{equation}
            \vert \bm u,\bm n\rangle\in\V_{+,\bm c},\quad\vert \bm u,\bm n\rangle:(c,\varphi)\mapsto e^{-\ps{Q,c}}\prod_{i=1}^p\ps{u_i,\varphi_{n_i}}.
        \end{equation}
		
		We can then define a structure of vertex algebra as follows:
		\begin{itemize}
			\item \emph{space of states}: we consider $\V_{+,\bm c}$ as above;
			\item \emph{vacuum vector}: we set $\vac\coloneqq (c,\varphi)\mapsto e^{-\ps{Q,c}}$ viewed as an element of $\V_{+,\bm c}$;
			\item \emph{translation operator}: we define $T:\V_{+,\bm c}^{(n)}\to\V_{+,\bm c}^{(n+1)}$ by setting
			\begin{equation}\label{eq:def_T}
					T\left(\vert \bm u,\bm n\rangle\right)\coloneqq -e^{-\ps{Q,c}}\sum_{i=1}^{p}(n_i+1)\ps{u_i,\varphi_{n_i+1}}\prod_{1\leq j\neq i\leq p}\ps{u_i,\varphi_{n_j}};
			\end{equation} 
			\item \emph{vertex operators}: we define a map $Y(\cdot,z):\V_{+,\bm c}\to\endovp[[z,z^{-1}]]$ by setting
			\begin{equation}\label{eq:def_VOA_H}
				\begin{split}
					&Y\left(\vert \bm u,\bm n\rangle,z\right)\coloneqq \quad :\prod_{i=1}^p\frac{\partial^{n_i}\ps{u_i,\Phi(z)}}{n_i!}:.
				\end{split}
			\end{equation}
		\end{itemize}
        Like for $\Phi$, $Y(\vecgen;z)$ also defines an unbounded, closable operator with dense domain $\mc C_\infty$.
		\begin{remark}
			We can define in the very same fashion a representation of the Heisenberg vertex algebra by picking $\V_-$ as vector space and slightly adapting the definition. We can also extend the Vertex Operators into a map $Y(\cdot,z):\V_{\bm c}\to\endov[[z^{\pm1},\bar z^{\pm 1}]]$ by setting
			\begin{equation}\label{eq:def_VOA_H_gen}
					Y(\vert \bm u,\bm n\rangle,z)\coloneqq \quad :\prod_{n_i>0}\frac{\partial^{n_i}\ps{u_i,\Phi(z)}}{n_i!}\prod_{n_i<0}\frac{\bar\partial^{-n_i}\ps{u_i,\Phi(z)}}{(-n_i)!}:.
			\end{equation}
		\end{remark}
		\begin{remark}
			This construction is valid for any $\bm c\in\C$ since it only appears via the multiplicative factor $e^{-\ps{Q,c}}$. Likewise Theorem~\ref{thm:VOA_Proba} is actually valid for the whole range of $\bm c\in\C$. 
		\end{remark}
		Based on this definition we have the following (see Appendix~\ref{appendix:VOA} for more background):
		\begin{proposition}
			The data of $\left(\V_{+,\bm c},\vac,T,Y(\cdot,z)\right)$ as above is isomorphic (in the sense of vertex algebra) to the (rank $r$) Heisenberg vertex algebra.
		\end{proposition}
		\begin{proof}
			This follows from the fact (see for instance~\cite[Section 2.1]{FBZ} with a slight adaptation since this is $J$ and not $\Phi$ that corresponds to the $b$ employed there) that $\V_+$ is generated by the action of creation operators on the vacuum state. The only things that have to be checked are the fact  that the vacuum state satisfies $\A_n\vac=0$ for $n\geq0$  and that with this definition of the translation operator we do have the property that $\left[T,J(z)\right]\varphi_n=\partial J(z)\varphi_n$ for any positive integer $n$.
			The first point is easily seen, while for the second one this is a consequence of the equality $\left[T,\A_m\right]\varphi_n=-m\A_{m-1}\varphi_n$ for any integer $m$, proved by elementary computations.
		\end{proof}
		
		This construction is rather standard, but it is crucial to establish a translation between the language of vertex algebras and a more probabilistic one as exemplified by Theorem~\ref{thm:VOA_Proba} below.
		
		\subsection{Probabilistic interpretation of the vertex algebra structure}\label{subsec:proba_ope}
		We now explain how the structure of vertex algebra fits within our probabilistic framework. For this purpose we will make sense of the so-called \textit{state-operator correspondence} in our setting and explain how some of the methods employed in vertex algebra translate as purely probabilistic computations.
		
		\subsubsection{The state-operator correspondence}
		The starting point is the observation that the vertex operators from the above vertex algebra construction can be expressed in probabilistic terms:
		\begin{theorem}\label{thm:VOA_Proba}
			The probabilistic description of the Heisenberg vertex algebra give above can be expressed in terms of the two Hilbert spaces $\Hil$ and $\Hild$ as follows: 
			\begin{enumerate}
				\item \emph{state-operator correspondence}: for any $F\in\V_{\bm c}$ we have $F=Y(F,0)\vac;$
				\item \emph{states}: take any positive integers $n_1,\cdots,n_p$ as well as $u_1,\cdots,u_p$ in $\h$. Then
				\begin{equation}
					\prod_{k=1}^p\ps{u_k,\varphi_{n_k}}=\ephi{:\prod_{k=1}^p\frac{\partial^{n_k}\ps{u_k,\X(0)}}{n_k!}:};
				\end{equation}
				\item \emph{vertex operators}: under similar assumptions, for $z\neq 0$:
				\begin{equation}\label{eq:voa_prob}
						Y\Big(\vert \bm u,\bm n\rangle,z\Big)\vert \bm v,\bm m\rangle=U_0\left(:\prod_{k=1}^p\frac{\partial^{n_k}\ps{u_k,\X(z)}}{n_k!}::\prod_{l=1}^q\frac{\partial^{m_l}\ps{v_m,\X(0)}}{m_l!}:\right)\cdot
				\end{equation}
			\end{enumerate}
		\end{theorem}
		This statement establishes a correspondence between the structure of vertex algebra attached to the free-field theory and the probabilistic representation of some functionals of a free-field on the disk. 	
		The notion of normal ordered product for vertex operators in the language of vertex algebra thus naturally translates as the notion of Wick product for field observables in that we have the equality (over $\V$ but more generally for suitable $F$):
		\begin{equation}
			\begin{split}
				&:\prod_{k=1}^p\frac{\partial^{n_k}\ps{u_k,\Phi(z)}}{n_k!}\prod_{l=1}^q\frac{\bar\partial^{m_l}\ps{v_l,\Phi(z)}}{m_l!}:U_0F\\
				=&U_0\left(:\prod_{k=1}^p\frac{\partial^{n_k}\ps{u_k,\X(z)}}{n_k!}\prod_{l=1}^q\frac{\bar\partial^{m_l}\ps{v_l,\X(z)}}{m_l!}:F(\X)\right)\cdot
			\end{split}
		\end{equation}
		
		This result is actually a consequence of a more general lemma which we now give. To provide a meaningful statement let us define for $0<\delta<1$ a subset $\mathcal{F}_\delta$ of $\mathcal{F}_\D$ by setting
		\begin{equation}
			\begin{split}
				\mathcal{F}_\delta\coloneqq\span\left\{F:\X+c\mapsto\prod_{i=1}^l\ps{\X+c,g_i}_\D e^{\ps{\X+c,f}};\quad l\geq0\text{ and }f,g_i\in\mathcal{E}_\delta\right\}\quad\text{where}
			\end{split}
		\end{equation}
		\[
		\mathcal E_\delta\coloneqq \left\{f\left(e^{-t+i\theta}\right)=\sum_{\norm{n}\leq N}f_n(t)e^{in\theta}\text{ with }f_n\in C_0^{\infty}\left((-\ln\delta,+\infty)\to\a\right), N\geq0\right\}.
		\]
		\begin{lemma}\label{lemma:to_prove1}
			In the setting of Theorem~\ref{thm:VOA_Proba}, for any $F\in\mc F_\delta$ and as soon as $\norm{z}\geq\delta$:
			\begin{equation}\label{eq:voa_prob_Fd}
				\begin{split}
					Y\left(\vecgen,z\right)U_0F=		U_0\left(:\prod_{k=1}^p\frac{\partial^{n_k}\ps{u_k,\X(z)}}{n_k!}:F(\X)\right)\cdot
				\end{split}
			\end{equation}
		\end{lemma}
		The proof of this statement is not very informative and rather computational. As a consequence we do not include it there but refer to Appendix~\ref{appendix:proofs} where it can be found.

		\subsubsection{Probabilistic interpretation of Operator Product Expansions}
		We now wish to understand the counterpart of Operator Product Expansions within our probabilistic framework.
		\begin{proposition}\label{prop:ope_prob}
			Take $n_1,\cdots,n_p$, and $m_1,\cdots,m_q$ two sets of positive integers, as well as $u_1,\cdots,u_p$ and $v_1,\cdots,v_q$ in $\h$. Then for any $F\in\mc F_\delta$ and $1>\norm{z}>\norm{w}>\delta>0$ we have
			\begin{equation}\label{eq:ope_prob}
                U_0\left(:\prod_{k=1}^p\frac{\partial^{n_k}\ps{u_k,\X(z)}}{n_k!}::\prod_{l=1}^q\frac{\partial^{m_l}\ps{v_l,\X(w)}}{m_l!}:F(\X)\right)=Y\left(\vert \bm u,\bm n\rangle,z\right)Y\left(\vert \bm v,\bm m\rangle,w\right)U_0F.
			\end{equation}
			More generally under similar assumptions and as soon as $\norm{z_1}>\cdots>\norm{z_l}\geq\delta$
			\begin{equation}\label{eq:ope_prob_gen}
				\begin{split}
					&U_0\left(\prod_{k=1}^l:\prod_{j=1}^{l(n_k)}\ps{u_{k,j},\partial^{\bm n_{k,j}}\X(z_k)}:F(\X)\right)=Y\left(\vert \bm u_1,\bm n_1\rangle,z_1\right)\cdots Y\left(\vert \bm u_l,\bm n_l\rangle,z_l\right)U_0F.
				\end{split}
			\end{equation}
		\end{proposition}
		Like before we postpone to the Appendix~\ref{appendix:proofs} the proof of this statement.
		
		
		\subsection{Bosonic Vertex Operators}
		Up to now we have only considered functionals depending on \textit{derivatives} of the free-field. We now turn to its exponentials: (bosonic) Vertex Operators . 
		
		\subsubsection{Definition of the bosonic Vertex Operators}
		In the probabilistic free-field theory, the primary fields are observables of the field $\X$ formally defined by $V_\alpha(z)=e^{\ps{\alpha,\X(z)}}$ for some \textit{weight} $\alpha\in\h$ and $z$ any point in $\C$. Based on the state-operator correspondence and more precisely item $(2)$ in Theorem~\ref{thm:VOA_Proba} we define an element of $\endv[[\log\norm{z},z^{\pm1},\bar z^{\pm 1}]]$ by setting, for any $\alpha\in\h$,
        \begin{equation}
			\mc V_{\alpha}(z)\coloneqq \sum_{n\in\N}\frac{:\ps{\alpha,\Phi(z)}^n:}{n!}=\quad :e^{\ps{\alpha,\Phi(z)}}:.
		\end{equation}
		More precisely the normal ordering that appears above takes the form
		\begin{equation}
			\mc V_\alpha(z)=e^{\ps{\alpha,c}}e^{\ps{\alpha,\i\A_-(z)+\i\bar\A_-(\bar z)}}e^{\ps{\alpha,\i\A_+(z)+\i\bar\A_+(\bar z)}}
		\end{equation}
		where in $\A_-$ (resp. $\A_+$) we have gathered the creation (resp. annihilation) operators:
		\begin{equation}
			\A_-(z)\coloneqq-\sum_{n<0}\frac1n\A_n z^{-n},\quad\A_+(z)\coloneqq\A_0\log z-\sum_{n>0}\frac1n\A_n z^{-n}.
		\end{equation}
		These are closely related to the \textit{bosonic Vertex Operators} that arise in the vertex algebra literature. Namely the bosonic Vertex Operators considered \textit{e.g.} in~\cite{FBZ,AF} contain only one chirality: 
		\begin{equation} 
			\mc V_{\alpha}^+(z)\coloneqq e^{\ps{\alpha, c}} :e^{\ps{\alpha,\i\A(z)}}:\qt{and}\mc V_{\alpha}^-(\bar z)\coloneqq e^{\ps{\alpha, c}}:e^{\ps{\alpha,\i\bar \A(\bar z)}}:.
		\end{equation}
        Defining them requires to fix a determination of the logarithm (which we assume in Proposition~\ref{prop:VO_prim_chiral} below).
		We can then write (the order is irrelevant since the two Heisenberg algebras commute)
		\begin{equation}
			\mc{V}_\alpha(z)=e^{-\ps{\alpha, c}}\mc V_{\alpha}^+(z)\mc V_{\alpha}^-(\bar z).
		\end{equation}
		At this stage it may be unclear how to make sense of such objects. 
        However as we now discuss they admit a probabilistic representation, allowing to make sense of $\mc V_\alpha(z)$ when acting on $\V_{+,\bm c}$.
		
		\subsubsection{Connection with the probabilistic Vertex Operators}
		Exponentials of the free-field  allow to define the correlation functions of Vertex Operators. Let $\alpha_1,\cdots,\alpha_N$ be in $\h$ and $z_1,\cdots,z_N$ be distinct inside $\D$. Correlation functions on the Riemann sphere are formally given by writing
		\begin{equation*}
			\ps{\prod_{k=1}^NV_{\alpha_k}(z_k)}_{\gamma}=\left(\bm 1,\prod_{k=1}^NV_{\alpha_k}(z_k)\right)_\D,\qt{with}V_{\alpha}(z)=:e^{\ps{\alpha,\X(z)}}:.
		\end{equation*}
		  The bosonic Vertex Operators defined above provide an alternative representation of the above:
		\begin{proposition}\label{prop:VO_prim}
			For any $F\in\mc F_\delta$ and $1> \norm{z_1}>\cdots>\norm{z_N}\geq\delta$ we have
			\begin{equation}
				\mc V_{\alpha_1}(z_1)\cdots\mc V_{\alpha_N}(z_N)U_0F=U_0\left(\prod_{k=1}^NV_{\alpha_k}(z_k)F(\X)\right).
			\end{equation}
			In particular $\mc V_{\alpha}(z)$ is well-defined over $\V_{+,\bm c}$ and we have for $z\neq0$:
			\begin{equation}
				\mc V_{\alpha}(z)\prod_{k=1}^p\ps{u_k,\i m_k\A_{-m_{k}}}\vac=e^{\ps{\alpha,c}}:e^{\ps{\alpha,P\varphi(z)}}:\ephi{:\prod_{k=1}^p\frac{\partial^{m_k}\ps{u_k,\X+\alpha\ln\frac1{\norm{z-\cdot}}}}{m_k!}(0):}.
			\end{equation}
		\end{proposition}
		The proof can be found in Appendix~\ref{appendix:proofs}. From it we highlight that we get that for $\norm{z}>\norm{w}$:
		\begin{equation}
			\mc V_\alpha(z)\mc V_\beta(w)=\norm{z-w}^{-\ps{\alpha,\beta}}:\mc V_\alpha(z)\mc V_\beta(w):.
		\end{equation}
		It also shows that we can keep only one chirality in the following sense:
		\begin{proposition}\label{prop:VO_prim_chiral}
			For any $F\in\mc F_\delta$ and $1> \norm{z_1}>\cdots>\norm{z_N}\geq\delta$ we have
			\begin{equation}
				\begin{split}
					&\mc V^+_{\alpha_1}(z_1)\cdots\mc V_{\alpha_N}^+(z_N)U_0F\\
					&=\prod_{k<l}\left(z_k-z_l\right)^{-\frac{\ps{\alpha_k,\alpha_l}}{2}}e^{\sum_{k=1}^N\ps{\alpha_k,c+P\varphi_+(z_k)}}U_0\left[F\left(\X(w)+\sum_{k=1}^N\frac{\alpha_k}2\ln\left(\frac{1-z_k\bar w}{z_k-w}\right)\right)\right].
				\end{split}
			\end{equation}
			Here the notation $\ln(z-w)$ has to be understood as $\ln z-\sum_{n\geq1}\frac1nw^nz^{-n}$ for $\norm{z}>\norm{w}$.
		\end{proposition}
        The above might depend on the determination of the logarithm fixed, however this ambiguity will be lifted in the applications considered in the present document, namely in Subsections~\ref{subsec:vir} and~\ref{subsec:screen}.
		Before moving on let us describe the OPE between a vertex operator from the Heisenberg vertex algebra and one of the bosonic Vertex Operators defined above. 
		\begin{proposition}\label{prop:VO_prim_bis}
			Let $\vecgen\in\V_{\bm c,+}$ and $\alpha\in\h$. Then for $\norm{z}>\norm{w}$
			\begin{equation}
				Y(\vecgen,z)\mc V_\alpha(w)=:F\left[\Phi+\alpha\ln\frac1{\norm{\cdot-w}}\right](z)\mc V_\alpha(w):\qt{with}F[\X]\coloneqq \prod_{k=1}^p\ps{u_k,\partial^{n_k}\X(z)}.
			\end{equation}
		\end{proposition}
		\begin{proof}
			The statement can be proved by explicit computations by splitting $\Phi$ as $\Phi_-+\Phi_+$ between creation and annihilation operator. More elegantly, we can use Girsanov's theorem~\ref{thm:girsanov} together with the correspondence between probabilistic and bosonic Vertex Operators from Proposition~\ref{prop:VO_prim}.
		\end{proof}
		
		\subsubsection{Radial ordering}
		We conclude this section with a synthetic way of writing the above correspondences using \textit{radial ordering}. Let $z\in\D$, $\alpha\in\a$, and $\bm u,\bm n$ like before. With a slight abuse of notation and in agreement with what has been said above we set
		\begin{equation}
			Y\left(\vert \alpha;\bm u,\bm n\rangle,z\right)\coloneqq \quad :Y\left(\vecgen,z\right)\mc V_\alpha(z):
		\end{equation}
		where with the normal ordering annihilation operators are on the right and creation operators on the left. Then a straightforward generalization of the above gives for $F\in\mc F_\delta$ and $\norm{z}>\delta$:
		\begin{equation}
			Y\left(\vert \alpha;\bm u,\bm n\rangle,z\right)U_0\left(F(\X)\right)=U_0\left(:\prod_{k=1}^p\ps{u_k,\partial^{n_k}\X(z)}V_\alpha(z): F(\X)\right)
		\end{equation}
		More generally for $1>\norm{z_1}>\cdots>\norm{z_N}>\delta$ as well as $(\alpha_1,\bm u_1,\bm n_1),\cdots,(\alpha_N,\bm u_N,\bm n_N)$ as above
		\begin{equation}\label{eq:radial_ordering}
			U_0\left(\prod_{k=1}^N:\prod_{i=1}^{l(u_k)}\ps{u_{k,i},\partial^{n_{k,i}}\X(z_k)} V_{\alpha_k}(z_k): F(\X)\right)=\mc R\left(\prod_{k=1}^NY\left(\vert \alpha_k;\bm u_k,\bm n_k\rangle,z_k\right)\right)U_0\left(F(\X)\right)
		\end{equation}
		for any $F\in\mc F_\delta$, where $\mc R\Big(\cdot\Big)$ denotes the radial ordering, defined by setting for $\norm{z_1}>\cdots>\norm{z_N}$:
		\begin{equation}
			\mc R\Big(\prod_{k=1}^NY\left(\vert \alpha_k;\bm u_k,\bm n_k\rangle,z_k\right)\Big)=Y\left(\vert \alpha_1;\bm u_1,\bm n_1\rangle,z_1\right)\cdots Y\left(\vert \alpha_N;\bm u_N,\bm n_N\rangle,z_N\right).
		\end{equation}

		

		\section{From the free-field to $W$-algebras}\label{sec:W-algebra}
		In the previous sections we have demonstrated that the Heisenberg algebra was the natural vertex algebra to describe the free-field theory. We are interested here in some implications of this correspondence for subalgebras of the Heisenberg algebra given by \textit{$W$-algebras}. To do so we will represent the $W$-algebra associated to a simple and complex Lie algebra $\g$ as a family of operators acting on $\Hil$ and defined using the Heisenberg vertex algebra introduced above. 
		
		This will be achieved thanks to \textit{screening operators}, which provide an explicit description of the $W$-algebras that turns out to be very well fitted for our purpose (see the next section). 
		We mention that the Appendix~\ref{appendix:Walgebra} is dedicated to presenting two additional constructions of these $W$-algebras.
		
		\subsection{On the Virasoro vertex algebra}
		Before considering $W$-algebras, we first define the Virasoro vertex algebra associated to the probabilistic free-field theory. We also introduce the primary fields, which are shown to coincide with their probabilistic definition, and of the \textit{screening operators}.
		
		\subsubsection{The Virasoro vertex algebra}
		The two representations of the Heisenberg vertex algebra constructed in the previous section actually define \textit{conformal} VOAs. Namely let
		\begin{equation}
			w_+\coloneqq (c,\varphi)\mapsto e^{-\ps{Q,c}}\left(\ps{Q,\varphi_{2}}-\norm{\varphi_1}^2\right)\in\V_{+,\bm c}
		\end{equation}
        be the \textit{conformal element}\footnote{It differs from the usual convention (by a normalization constant) made in the literature due to our convention on the commutation relations of the Heisenberg algebra. The extra exponential factor comes from the definition of the $\A_0$ operator to take into account the background charge $Q$.}. Then we have the following:
		\begin{proposition}
			The $\Z$-graded vertex algebra $\V_{+,\bm c}=\bigoplus_{n\geq0}\V_{+,\bm c}^{(n)}$ is a conformal vertex algebra with conformal vector $w_+\in\V_+^{(2)}$. Its central charge is given by $\bm c=r+6\ps{Q,Q}^2$.	
		\end{proposition}
		\begin{proof}
			This corresponds to the usual construction of the Virasoro vertex algebra from the Heisenberg one, see \textit{e.g.}~\cite[Section 2.5.9]{FBZ}.
		\end{proof}
		The corresponding Vertex Operator is the \textit{stress-energy tensor} of the theory, defined by
		\begin{equation}
			\L(z)\coloneqq Y(w_+,z)=\sum_{n\in\Z}\L_nz^{-n-2}
		\end{equation}
		where the $(\L_n)_{n\in\Z}$ satisfy the commutation relations of a Virasoro algebra with central charge $\bm c$:
		\begin{equation}
			[\L_n,\L_m]=(n-m)\L_{n+m}+\frac {\bm c}{12}(n^3-n)\delta_{n,-m}.
		\end{equation}
        Like the $\A_{n,i}$, the $\L_n:\mc C_\infty\to\Hil$ define unbounded, densely defined and closable operators. We refer \textit{e.g.} to~\cite[Section 2]{BGKRV} for more analytical properties of these $\L_{n}$. 
		The modes of the stress-energy tensor can be expressed using the \textit{Segal-Sugawara} construction using the Heisenberg algebra:
		\begin{equation}
			\L_n=-\i Q(n+1)\A_n+\sum_{m\in\Z}:\ps{\A_{n-m}\A_m}:.
		\end{equation}
		
		In the rest of this section we denote by $\vir_+$ the representation of the Virasoro vertex algebra thus defined. That is to say $\vir_+$ is the data of $(\mc M_{+,\bm c},\vac,T,Y)$ where $\vac, T$ and $Y$ are defined like before but where the space of states is now given by
		\begin{equation}
			\mc M_{+,\bm c}\coloneqq \bigoplus_{n\geq 0} \mc M_{+,\bm c}^{(n)},\qt{where}\mc M_{+,\bm c}^{(n)}\coloneqq \bigoplus_{\substack{m_1,\cdots m_n\geq 2\\ m_1+\cdots+m_n=n}}\text{span}\left\{\L_{-m_1}\cdots \L_{-m_n}\vac\right\}.
		\end{equation}
		A standard result is that $ \mc M_{+,\bm c}^{n}$ is actually spanned by a sum ranging over Young diagrams:
		\begin{equation}
			\mc M_{+,\bm c}^{(n)}\coloneqq \bigoplus_{\substack{\lambda\in\mc T_2\\ \norm{\lambda}=n}}\text{span}\left\{\L_{-\lambda}\vac\right\},\qt{where}\L_{-\lambda}\coloneqq\L_{-\lambda_l}\cdots\L_{-\lambda_1}.
		\end{equation}

		\subsubsection{Primary fields and Virasoro modules}\label{subsec:vir}
		The counterparts of the Vertex Operators in the space of states $\Hil$ are the \textit{primary fields}
		\begin{equation}
			\begin{split}
				\vert\alpha\rangle &\coloneqq (c,\varphi)\mapsto e^{-\ps{Q,c}}\sum_{n\in\N}\frac1{n!}\ps{\alpha,\varphi_{0}}^n=e^{\ps{\alpha-Q,c}}.
			\end{split}
		\end{equation}
		They do \textit{not} belong to the space of states $\Hil$. Still they define Virasoro modules associated to the vertex algebra $\vir_+$: for any such $\alpha\in\h$ let us introduce the graded vector space
		\begin{equation}\label{eq:vir_mod}
			\mc M_\alpha^+\coloneqq \bigoplus_{n\geq 0} \mc M_\alpha^{+,n},\qt{where}\mc M_\alpha^{+,n}\coloneqq \bigoplus_{\substack{\lambda\in\mc T_1\\ \norm{\lambda}=n}}\text{span}\left\{\L_{-\lambda_1}\cdots \L_{-\lambda_r}\prim\right\}.
		\end{equation}
		A consequence of Proposition~\ref{prop:Vir_prim} below is that for any $\alpha\in\h$, $\mc M_\alpha^+$ is a conformal module over $\vir_+$ (it may also be called a \textit{Feigin-Fuchs module}) in the sense of~\cite[Definition 5.1.9]{FBZ}. It is \textit{not} a Heisenberg module: these are rather defined by setting
		\begin{equation}
			\V_\alpha^+\coloneqq \bigoplus_{n\geq 0} \V_\alpha^{+,n},\qt{where}\V_\alpha^{+,n}\coloneqq \bigoplus_{\substack{m_1,\cdots m_r\geq 1\\ m_1+\cdots+m_r=n}}\left\{\ps{u_1,\A_{-m_1}}\cdots \ps{u_n,\A_{-m_r}}\prim\right\}.
		\end{equation}
		  Such objects do indeed define primary fields in the usual sense:
		\begin{proposition}\label{prop:Vir_prim}
			For any $\alpha\in\h$ and $z\in\D\setminus\{0\}$:\begin{enumerate}
				\item $\vert\alpha\rangle$ is a primary of conformal dimension $\Delta_\alpha=\ps{\frac\alpha2,Q-\frac\alpha2}$ in the sense that
				\begin{equation}
					\L_n\prim=0\qt{for all}n\geq1\qt{and}\L_0\prim=\Delta\alpha\prim;
				\end{equation}
				\item the following commutation relations hold true for any $n\in\Z$:
				\begin{equation}\label{eq:comm_L_V}
					[\L_n,\mc V_\alpha(z)]=z^n\left(\Delta_\alpha +z\partial_z\right)\mc V_\alpha(z).
				\end{equation}
			\end{enumerate}
		\end{proposition}
         In the last equation, the equality has to be understood in the sense of the $\ps{\cdot,\cdot}_2$ pairing over $U_0(\mc F_\delta)$, that is for any $F,G$ in $\mc F_\delta$, $\ps{[\L_n,\mc V_\alpha(z)]U_0F,U_0G}_2=z^n\left(\Delta_\alpha +z\partial_z\right)\ps{\mc V_\alpha(z)U_0 F,U_0G}_2$.  
          
		\begin{proof}
			For the first item, $\A_m\vert\alpha\rangle=0$ for positive $m$ so $\L_n\prim=-\i(n+1)\ps{Q,\A_n}\prim+\ps{\A_n,\A_0}\prim$ for any $n\in\Z$. As consequence $\L_n\prim=0$ for any $n\geq1$.
			As for $\L_0$ we have
			\[
			\L_0\prim=-\i\ps{Q,\A_0}\prim+\ps{\A_0,\A_0}\prim=\ps{\frac Q2,\alpha+Q}-\frac14\ps{\alpha+Q,\alpha+Q}=\Delta_\alpha.
			\]
			Item $(2)$ follows from explicit computations since $\left[\A_n,\mc V_\alpha(z)\right]= \frac {\i\alpha}2z^{n}\mc V_\alpha(z)$ so that for positive $n$:
			\begin{align*}
				\left[\L_n,\mc  V_\alpha(z)\right]&=\frac{\ps{\alpha, Q}}2(n+1)z^n\mc V_\alpha(z)+\frac12\sum_{m\geq 1}\mc V_\alpha(z)\ps{\alpha,\i\A_{m}}z^{n-m}+\frac12\sum_{m\geq 1}\ps{\alpha,\i\A_{n-m}}z^{m}\mc V_\alpha(z).
			\end{align*}
			Now for $0\leq m\leq n$ we have $\ps{\alpha,\i\A_{n-m}}z^{m}\mc V_\alpha(z)=\quad :\ps{\alpha,\i\A_{n-m}}z^{m}\mc V_\alpha(z):-\frac{\ps{\alpha,\alpha}}2z^n\mc V_\alpha(z)$. We deduce the result for positive $n$. The proof follows from the same argument for $n\leq0$.
		\end{proof}

		We now explain how to give a meaning to the bosonic Vertex Operators previously introduced. First, applying $\mc{V}^+_\alpha(z)$ to some $w\in\V_\beta^+$ may be ill-defined due to the logarithm that enters the definitions of $\A(z)$, making it possibly multi-valued. Likewise it is not clear, beyond this obstruction, what meaning should one give to $\mc{V}^+_\alpha(z)w$ for $w\in\V_{+,\bm c}$. To remedy this issue we formally set
		\begin{equation}
			\hat {\mc V}_{\alpha}^+(z)\coloneqq \exp\left(\sum_{n\in\Z^*}\frac 1{2n}\ps{\alpha,\partial_n}z^{-n}\right)\qt{so that}{\mc V}_{\alpha}^+(z)=e^{\ps{\alpha,c+P\varphi_+(z)}}e^{\ps{\alpha,\i\A_0\log z}}\hat{\mc V}_{\alpha}^+(z)
		\end{equation}
		where recall that $P\varphi_+(z)=\sum_{n\geq 1}\varphi_nz^n$.  
		Then over any Heisenberg module $\V_\beta^+$:
		\begin{equation*}
			\mc V^+_\alpha(z) \ps{u_l,\A_{-n_l}}\cdots \ps{u_1,\A_{-n_1}}\vert\beta\rangle=z^{-\frac{\ps{\alpha,\beta}}2}e^{\ps{\alpha,P\varphi_+(z)}}\hat{\mc V}^+_\alpha(z)\ps{u_l,\A_{-n_l}}\cdots \ps{u_1,\A_{-n_1}}\vert\beta+\alpha\rangle.
		\end{equation*}
		In particular $\mc V^+_\alpha(z)$ is single-valued over $\V_\beta^+$ as soon as $\ps{\alpha,\beta}\in2\Z$. Now $\hat{\mc V}^+_\alpha(z)\vert\beta+\alpha\rangle=\vert\beta+\alpha\rangle$ so that we can use the Baker-Campbell-Hausdorff formula to show that
		\begin{align*}
			\hat{\mc V}^+_\alpha(z)\ps{u_l,\A_{-n_l}}\cdots \ps{u_1,\A_{-n_1}}\vert\beta+\alpha\rangle\qt{belongs to}\V_{\alpha+\beta}^+[z^{-1}].
		\end{align*}
		In brief, $\mc V^+_\alpha(z)$ is well-defined when viewed as a linear map:
		\begin{equation}
			\mc V^+_\alpha(z):\V_\beta^+\longrightarrow \V_{\alpha+\beta}^+((z))\qt{as soon as}\ps{\alpha,\beta}\in2\Z.
		\end{equation}
		In particular if we take $\alpha\in\h$ such that $\norm{\alpha}^2\in2\Z$ then we are naturally led to defining Vertex Operators over the module $\V_\alpha^+$ by the assignment
		\begin{equation}
			Y\Big(\ps{u_1,\A_{-m_1}}\cdots\ps{u_r,\A_{-m_r}}\prim,z\Big)=:Y\Big(\ps{u_1,\A_{-m_1}}\cdots\ps{u_r,\A_{-m_r}}\vac,z\Big)\mc V_\alpha^+(z):.
		\end{equation} 
		The vacuum axiom and translation invariance are satisfied, while locality holds between fields $Y(A,z)$ and $Y(B,w)$ with $A\in\V_0^+$ and $B\in \V_\alpha^+$.
		In order to define a vertex algebra we may thus consider $\Lambda_\alpha\coloneqq\bigoplus_{N\in\Z}\V^+_{N\alpha}$:	the data of $(\Lambda_\alpha,\vac,T,Y_\alpha)$ defines a vertex algebra. Likewise the restriction of the above to $\bigoplus_{N\in\Z}\mc M^+_{N\alpha}$ defines a conformal vertex algebra~\cite[Proposition 5.2.5]{FBZ}.
		
		\subsection{$W$-algebras from the screening operators}
		We now turn to the definition of $W$-algebras based on these screening operators, which originates from works by Feigin-Frenkel~\cite{FF_KM,FF_QG} in relation with the celebrated (quantum) Drinfeld-Sokolov reduction~\cite{DS,FKW,FF_DS}.  A $W$-algebra is viewed there as a subalgebra of the (rank $r$) Heisenberg algebra.
		
		\subsubsection{Lie algebra reminders}\label{subsec:lie}
        We first recall some basic Lie algebras notions.
		A simple Lie algebra $\mathfrak{g}$ is a (non-Abelian) Lie algebra that doesn't admit any proper, nonzero ideal. Simple Lie algebras that are finite-dimensional and complex are completely classified: they are either isomorphic to one of the \textit{classical} Lie algebras $(A_n)_{n\geq 1}$ (corresponding to $\mathfrak{sl}_{n+1}$), $(B_n)_{n\geq 2}$ (for $\mathfrak{o}_{2n+1}$), $(C_n)_{n\geq 3}$ ($\mathfrak{sp}_n$) and $(D_n)_{n\geq 4}$ ($\mathfrak{o}_{2n}$), or one of the exceptional Lie algebras $E_6$, $E_7$, $E_8$, $F_4$ and $G_2$.
		
		A finite-dimensional, complex, simple Lie algebra $\mathfrak{g}$ naturally comes with a Euclidean space $(\mathfrak{a},\ps{\cdot,\cdot})$, which is such that the Cartan subalgebra $\mathfrak{h}$ of $\mathfrak{g}$ can be written as the complexification of $\a$. It is equipped with a (positive definite) scalar product $\ps{\cdot,\cdot}$ proportional to the Killing form of $\mathfrak{g}$. This Euclidean space has dimension given by $r$ the \emph{rank} of $\mathfrak{g}$ and is unique up to isomorphism. 
		
		The \emph{simple roots} are linear forms over $\h$. They can be identified (by means of the Killing form) with a basis $(e_i)_{1\leq i\leq r}$ of $\a$ such that, in terms of the Cartan matrix $A$ of $\mathfrak{g}$,
		\begin{equation}
			2\frac{\ps{e_i,e_j}}{\ps{e_i,e_i}}=A_{i,j} \quad\text{for all }1\leq i,j\leq r.
		\end{equation}
        The \textit{coroot} $e_i^{\vee}\coloneqq 2\frac{e_i}{\ps{e_i,e_i}}$ associated to $e_i$ naturally arises from the form of the Cartan matrix.  
		As usually assumed in the physics literature, we choose to normalize the scalar product $\ps{\cdot,\cdot}$ so that the longest roots have squared norm $2$: for such long roots we have $e^\vee=e$. The renormalization constant compared to the Killing form 
		is given by $2h^\vee$, where $h^\vee$ is the so-called \textit{dual Coxeter number}, an explicit positive integer that depends on the underlying Lie algebra.
		
		The \textit{fundamental weights} $(\omega_i)_{1\leq i \leq r}$ form the basis of $\mathfrak{a}$ 
        dual to that of the simple coroots:
		\begin{equation}\label{relweights}
			\omega_i\coloneqq  \sum_{l=1}^{r} (A^{-1})_{i,l}e_l\qt{and}
			\langle \omega_i,e_j^{\vee}\rangle= \delta_{ij}  ,\quad  \langle \omega_i,\omega_j \rangle= (A^{-1})_{i,j}
		\end{equation}
		where $ \delta_{ij}$ is the Kronecker symbol.
		They allow to define the \textit{Weyl vector} $\rho\in\mathfrak{a}$ by setting
		\begin{equation}\label{eq:def_rho}
			\rho\coloneqq \sum_{i=1}^r \omega_i,\qt{so that $\ps{\rho,e_i^{\vee}}=1$ for all $1\leq i\leq r$.}
		\end{equation}
		We will also consider the co-Weyl vector $\rho^\vee\in\a$ associated to the roots, \textit{i.e.} $\rho^\vee=\sum_{i=1}^r\omega_i^\vee$ where the $(\omega_i^\vee)_{1\leq i\leq r}$ satisfy $\ps{e_i,\omega_j^{\vee}}=\delta_{i,j}$ for all $1\leq i,j\leq r$.
		The squared norm of the Weyl vector can be expressed explicitly in terms of the Lie algebra under consideration via the \textit{Freudenthal-de Vries strange formula} for simple Lie algebras \cite[Equation (47.11)]{FdV}\footnote{This equation differs from the one in \cite{FdV} by a multiplicative factor $2h^\vee$. This is due to our normalization convention for the scalar product $\ps{\cdot,\cdot}$ on $\mathfrak{a}$.}, $\norm{\rho}^2=\frac{h^\vee\dim \mathfrak{g}}{12}\cdot$
		
		The exponents of $\g$ (sometimes called Coxeter exponents) are a set of $r$ integers $d_i$ with $1\leq d_i\leq h-1$, where $h$ is the Coxeter number. We will call the $s_i\coloneqq d_i+1$ the \textit{spins} of $\g$. They naturally appear in various settings: these spins are the degrees of the Casimir operators, but also they also arise in the description of the adjoint representation of $\g$. They are given by:
		\begin{multline}\label{central_charge}
			\arraycolsep=1.9pt\def\arraystretch{1.5}
			\begin{array}{c|c}
				\mathfrak{g} & \text{Exponents }d_i \\
				\hline A_{n}& 1,2,\cdots,n \\
				B_n & 1,3,\cdots,2n-1\\
				C_n& 1,3,\cdots,2n-1\\
			\end{array}
			\quad
			\arraycolsep=1.9pt\def\arraystretch{1.5}
			\begin{array}{c|c}
				\mathfrak{g} & \text{Exponents }d_i \\
				\hline D_n &  1,3,\cdots,2n-3,n-1\\
				F_4 & 1, 5, 7, 11\\
				G_2 & 1, 5
			\end{array}
			\quad
			\arraycolsep=1.9pt\def\arraystretch{1.5}
			\begin{array}{c|c}
				\mathfrak{g} & \text{Exponents }d_i \\
				\hline E_6 & 1, 4, 5, 7, 8, 11\\
				E_7 & 1, 5, 7, 9, 11, 13, 17\\
				E_8 & 1, 7, 11, 13, 17, 19, 23, 29.
			\end{array}
		\end{multline}
		
		\subsubsection{Screening operators}\label{subsec:screen}
		Let us consider a bosonic Vertex Operator $V_{\alpha^*}$ with weight $\alpha^*$ either given by $\alpha^*=\gamma e_i$ or $\alpha^*=\frac2\gamma e_i^\vee$ for some $i\in\{1,\cdots,r\}$. Then explicit computations show that $\Delta_{\alpha^*}=1$ so that the commutation relations~\eqref{eq:comm_L_V} now read
		\begin{equation}\label{eq:comm_L_Vi}
			[\L_n,\mc V^+_{\alpha^*}(z)]=\partial_z\Big(z^{n+1}\mc V_{\alpha^*}^+(z)\Big).
		\end{equation}
		In particular if we define the following \textit{screening operators}:
		\begin{align}\label{eq:screening}
			Q_{i}^+\coloneqq \frac{1}{2i\pi}\oint_{\T}\mc V_{\gamma e_i}^+(z)dz,\quad \check Q_{i}^+\coloneqq \frac{1}{2i\pi}\oint_{\T}\mc V_{\frac2\gamma e_i^\vee}^+(z)dz\qt{for all}1\leq i\leq r,
		\end{align}
		then for any integer $n$, $\L_n$ commutes (using Stokes' theorem) with the screening operators:
		\begin{equation}\label{eq:OPE_SET}
			[\L(z),Q_{i}^+]=[\L(z),\check Q_{i}^+]=0.
		\end{equation}
		Informally speaking, these screening operators may be thought of as the constant terms of $z\mc V_{\gamma e_i}^+(z)$. However some care is necessary to make sense of them when acting over the Heisenberg modules $\V_\alpha^+$. 
        Indeed, for $\alpha$ such that $\ps{\gamma e_i,\alpha}\in2\Z$, according to the previous discussion for any $v\in \V_{\alpha}^+$ the quantity $z\mc V_{\gamma e_i}^+(z)v$ is well-defined as an element of $\V_{\alpha+\beta}^+((z))$, and in particular the constant term in this expansion belongs to $\V_{\alpha+\beta}^+$. This shows that the following maps are well-defined
		\begin{equation}
			Q_{i}^+:\V_\alpha^+\to \V_{\alpha+\gamma e_i}^+,\quad \check Q_{i}^+:\V_\alpha^+\to \V_{\alpha+\frac2\gamma e_i^\vee}^+
		\end{equation}
		as soon as $\ps{\alpha,\gamma e_i}$ (resp. $\ps{\alpha,\frac2\gamma e_i^\vee}$) belongs to $2\Z$.
		In particular the $2r$ screening operators $Q_{i}^+,\check Q_{i}^+$ for $1\leq i\leq r$ are well-defined over $\V_{+,\bm c}$.
		As we now explain this allows to define, beyond the Virasoro algebra, a representation of the $W$-algebra associated to the simple Lie algebra $\g$.
		
		\subsubsection{From screening operators to $W$-algebras}
		Take any element $w$ of $\mc M_+$: then Equation~\eqref{eq:OPE_SET} implies that $Q_{i}^+w=0$ for any $1\leq i\leq r$.
		Conversely if one takes for instance the state $\ps{u,\A_{-1}}\vac$ for some non-zero $u\in\h$, then the latter defines an element of $\V_+$ but not of $\mc M_+$. Besides explicit computations show that it satisfies $Q_{i}^+\A_{-1}\vac\neq 0$ as soon as $\ps{u,e_i}\neq 0$. 
		This feature characterizes of the subspace of $\V_{+,\bm c}$ over which the $W$-algebra associated to $\g$ is defined.
		Namely introduce
		\begin{equation}
			\MW{\g}\coloneqq\bigcap_{1\leq i\leq r}\text{Ker}_{\V_{+,\bm c}}\left(Q_{i}^+\right).
		\end{equation}
		\begin{theorem}\label{thm:def_W}
			For generic values of the central charge $\bm c$, the data of $\mc W_{\bm c}(\g)\coloneqq\left(\MW{\g},\vac,T,Y\right)$ defines a vertex algebra isomorphic to the $W$-algebra associated to $\g$. Moreover there exist $s_i$-multilinear differential operators $\W{s_i}$ for $i=1,\cdots, r$ such that $\MW{\g}=\bigoplus_{n\geq0}\left(\MW{\g}\right)^{(n)}$ where
        \begin{equation}
				\left(\MW{\g}\right)^{(n)}\coloneqq\bigoplus_{\substack{\lambda_1\in\mc T_{s_1},\cdots\lambda_r\in\mc T_{s_r}\\ \norm{\lambda_1}+\cdots+\norm{\lambda_r}=n}}\text{span}\left\{(c,\varphi)\mapsto e^{-\ps{Q,c}}\ephi{:\prod_{i=1}^r\prod_{j=1}^{l(\lambda_i)}\frac{\partial^{\lambda_i^j-s_i}\Wb^{(s_i)}[\X](0)}{(\lambda_i^j-s_i)!}:}\right\}.
		\end{equation}
			The currents can be chosen in such a way that $\W 2(z)=\L(z)$, while for any integers $n,m$:
			\begin{equation}
				[\L_n,\W s_m]=\left((s-1)n-m\right)\W s_{n+m}.
			\end{equation}
		\end{theorem}
		\begin{proof}
			The first point follows from~\cite[Theorem 15.4.12]{FBZ}: using the notations employed there we have an identification between the Heisenberg VOA $\pi_0$ and the one constructed here, as well as between the screening operators $\int V_{-\alpha_i/\nu}(z)dz$ and $Q_{i}^+$. The assumption that $k$ is generic there amounts to prescribing $\bm c$ generic in our case. As a consequence for generic values of $\bm c$ the data of $\left(\MW{\g},\vac,T,Y\right)$ is isomorphic to the $W$-algebra associated to $\g$.
            
            Existence of the $s_i$-multilinear differential operators is a consequence of~\cite[Theorem 4.6.9]{FF_QG}. Namely for any spin $s$ of $\g$ there is an element $w_s\in\V_{+,\bm c}^{(s)}$ such that $\MW{\g}=\bigoplus_{n\geq0}\left(\MW{\g}\right)^{(n)}$, with
            \[
                \left(\MW{\g}\right)^{(n)}=\bigoplus_{\substack{\lambda_1\in\mc T_{s_1},\cdots\lambda_r\in\mc T_{s_r}\\ \norm{\lambda_1}+\cdots+\norm{\lambda_r}=n}}\text{span}\left\{\Wb_{-\lambda_r}^{(s_r)}\cdots \Wb_{-\lambda_1}^{(s_1)}\vac\right\}.
            \]
            Here $\Wb_{-\lambda}^{(s)}=\prod_{i=1}^{l(\lambda)}\Wb^{(s)}_{-\lambda_i}$ with $\W s(z)\coloneqq \sum_{n\in\Z}\W s_{n}z^{-n-s}\coloneqq Y(w_s;z)$. Now by definition of $\V_{+,\bm c}^{(s)}$ there exist $l(\lambda)$-linear forms $\Wb_\lambda^{(s)}$ such that we can expand 
            \[
                \Wb^{(s)}_{-s}\vac=\sum_{\substack{\lambda\in\mc T_1\\ \norm{\lambda}=s}}\Wb^{(s)}_\lambda(\lambda_1!\phi_{\lambda_1},\cdots,\lambda_l!\varphi_{\lambda_l}).
            \]
             We thus conclude using Proposition~\ref{prop:wick}.
		\end{proof}
        These currents can be defined explicitly in some special cases: the stress-energy tensor gives
        $\L[\X]=\W 2[X]=\ps{Q,\partial^2\X}-\ps{\partial\X,\partial\X}$, while for $\g=A_n$, $\g=D_n$ or $\g=B_2$ these are described in Appendices~\ref{appendix:Walgebra} and~\ref{appendix:WB2} (see for instance Equation~\eqref{eq:W3_curr} for $\g=\sl_3$). 
        When $\gamma$ is chosen to be real (and thus $Q\in\a$) the coefficients of these multilinear forms can be chosen in $\a$. 
        
        The embedding $\Upsilon:\MW{\g}\to\V_+$ thus defined lifts to an injective homomorphism between the $W$-algebra $\mc W_{\bm c}(\g)$ and the Heisenberg vertex algebra constructed before, the \textit{Miura map}, expressed using the $s$-multilinear differential operators $\Wb^{(s)}$. These are such that the OPE between $\mc V_{\gamma e_j}^+(z)$ and $\Wb^{(s_i)}[\Phi](w)$ can be written, up to regular terms, as a total derivative with respect to $z$.
		
		
		
		\subsection{Some properties of the free-field representation}
		Having defined the $W$-algebra associated to $\g$ using a family of operators acting on $\mc H_\T$, we now provide some basic properties of the $W$-module thus constructed. Namely we show that the representation constructed is \textit{unitary}, and describes its behavior under (infinitesimal) conformal transformations.
		
		\subsubsection{Covariance under conformal transformations}
		A key property of $W$-algebras is that they naturally contain the Virasoro algebra, which corresponds to the algebra of symmetry of a Conformal Field Theory. As such, the currents generating the $W$-algebra associated to $\g$ naturally enjoy conformal covariance under global conformal transformations of the plane, but also infinitesimal deformations. For this purpose we recall some of the material introduced in~\cite[Subsection 2.4]{BGKRV}: such infinitesimal deformations may be encoded using so-called \textit{Markovian} vector fields, that is holomorphic vector fields $\bm{\mathrm v}$ on $\D$ that extend smoothly to $\partial\D$, and for which $\Re(\bar z v(z))<0$ over $\partial\D$ where $\bm{\mathrm v}=v(z)\partial_z$ with $v(z)=-\sum_{n\geq-1}v_nz^{n+1}$. The flow of $\bm{\mathrm v}$ is the unique solution $(f_t(z))_{t\geq 0}$ of $f_0(z)=z$ and $\partial_t f_t(z)=v\left(f_t(z)\right)$ for $t\geq0$, and where $z\in\D$. 
		
		\begin{proposition}\label{prop:vir_deform}
			Let $\bm{\mathrm v}$ be a Markovian vector field with flow $f_t$, and $\L_{\bm{\mathrm v}}\coloneqq\sum_{n\geq -1}v_n\L_n$. Then 
			\begin{equation}
				\begin{split}
					\left[\L_{\bm{\mathrm v}},\W s(z)\right]=-\frac{d}{dt}\left( 
					\left(f_t'(z)\right)^{s}\W s\circ f_t(z)\right)\vert_{ t=0}.
				\end{split}
			\end{equation}
			The same applies to the bosonic Vertex Operators in that:
			\begin{equation}
				\begin{split}
					\left[\L_{\bm{\mathrm v}},\mc V_\alpha(z)\right]=-\left( 
					\frac{d}{dt}\left(f_t'(z)\right)^{\Delta_\alpha}\mc V_\alpha\circ f_t(z)\right)\rvert_{ t=0}.
				\end{split}
			\end{equation}
		\end{proposition}
        Like in Proposition~\ref{prop:Vir_prim}, these equalities hold in the sense of the $\ps{\cdot,\cdot}_2$ pairing over $U_0(\mc F_\delta)$.
		\begin{proof}
			To start with, thanks to the commutation relations between the Virasoro modes $\L_n$ and that of the currents $\W s$ for any integer $n$
			\begin{equation}
				\begin{split}
					\left[\L_n,\W s(z)\right]&=z^n\left((n+1)s+z\partial_z\right)\W s(z).
				\end{split}
			\end{equation}
            Now explicit computations show that
			\begin{align*}
				\left(f_t'(z)\right)^{-s}\frac{d}{dt}\left(f_t'(z)\right)^{s}\W s\circ f_t(z)=s\frac{\partial_t f_t'(z)}{f_t'(z)}+\partial_t f_t(z)\partial_z\W s(f_t(z))
			\end{align*}
			with $f_t'(z)\to 1$ and $\partial_t f_t'(z)\to \partial_z v(z)$ as $t\to0$. As a consequence 
			\begin{align*}
				\left(\frac{d}{dt}\left(f_t'(z)\right)^{s}\W s\circ f_t(z)\right)\rvert_{ t=0}=\left(s\partial_z v(z) +v(z)\partial_z\right)\W s(z),
			\end{align*}
			 allowing to conclude for $\W s$. The proof remains the same for $\mc V_\alpha$ by exchanging $s$ and $\Delta_\alpha$.
		\end{proof}
		
		These describe the variation of the currents under local conformal transformations. It admits a global formulation in terms of global conformal maps of the Riemann sphere:
		\begin{proposition}\label{prop:mobius}
			Let $\psi$ be any M\"obius transform of the sphere. Then for any spin $s$ of $\g$
			\begin{equation}\label{eq:cov_current}
				\Wb^{(s)}[\Phi](\psi(z))=\psi'(z)^{-s}\Wb^{(s)}[\Phi\circ\psi+Q\ln\norm{\psi'}](z).
			\end{equation}
		\end{proposition}
		To the best of our knowledge this statement is new from the perspective of vertex algebras. We note that a somewhat similar statement appears in~\cite[Proposition 6.4.4]{FBZ} though on a less explicit way and in a different language. Remarkably it is the exact analog of the transformation rule for the GFF that appears in the probabilistic study of Liouville and Toda CFTs, and that comes from the covariance property of Gaussian Multiplicative Chaos measures under diffeomorphisms.
		This statement is a direct consequence of the following alternative description of the action of infinitesimal deformations, written directly in terms of the underlying field $\Phi$: 
		\begin{lemma}\label{lemma:vir_deformQ}
			In the same setting as in Proposition~\ref{prop:vir_deform}:
			\begin{equation}
				\begin{split}
					\left[\L_{\bm{\mathrm v}},\Phi(z)\right]=- 
					\frac{d}{dt}\left(\Phi_+\circ f_t+Q\ln\norm{f_t'}\right)(z)\rvert_{ t=0}
				\end{split}
			\end{equation}
			with $\Phi_+\coloneqq c+\i\A(z)$.
			In particular for any spin $s$ of $\g$:
			\begin{equation}
				\begin{split}
					\left[\L_{\bm{\mathrm v}},\W s(z)\right]=-\left( 
					\frac{d}{dt}\W s\left[\Phi\circ f_t+Q\ln\norm{f_t'}\right](z)\right)\rvert_{ t=0}.
				\end{split}
			\end{equation}
		\end{lemma}
		\begin{proof}
			Let us consider the action of the Virasoro operators on $\Phi$: based on the commutation relation $[\L_n,\i\A_m]=Q\frac{n(n+1)}{2}\delta_{n,-m}-m \i\A_{n+m}$ we see that 
			\[
			\left[\L_{\bm{\mathrm v}},\Phi(z)\right]=-\frac Q2\partial v(z)-v(z)\partial \Phi(z).
			\]
			On the other hand
			\begin{align*}
				\frac{d}{dt}\left(\Phi_+\circ f_t+Q\ln\norm{f_t'}\right)(z)=\frac Q2\frac{\partial_t f_t'}{f_t'}+f_t'(z)\partial\Phi(f_t(z))
			\end{align*}
			so that by using the fact that as $t\to0$: $f_t'(z)\to 1$ and $\partial_t f_t'(z)\to \partial v(z)$ we get as desired
			\begin{align*}
				\frac d{dt}\left(\Phi_+\circ f_t+Q\ln\norm{f_t'}\right)(z)_{t=0}=\frac Q2\partial v(z)+v(z)\partial \Phi(z).
			\end{align*}
			This readily extends to the case where we consider $\partial^m \Phi$ instead of $\Phi_+$ for any $m\geq 1$.  To extend this relation to the current $\W s$ it suffices to apply the chain rule, that is write inductively that
			\begin{align*}
				&[\L_{\bm{\mathrm v}},:\partial^mJ(z)Y(w,z):]=:[\L_{\bm{\mathrm v}},\partial^mJ(z)]Y(w,z):+:\partial^mJ(z)[\L_{\bm{\mathrm v}},Y(w,z)]:\qt{and}\\
				&\left(\frac{d}{dt}:\partial^m\left[f_t^*\Phi\right]F(f_t^*\Phi)(z):\right)\rvert_{ t=0}=:\frac{d}{dt}\rvert_{ t=0}\partial^m\left[f_t^*\Phi\right]F(\Phi)(z):+:\partial^m\left[\Phi\right]\frac{d}{dt}\rvert_{ t=0}F(f_t^*\Phi)(z):.
			\end{align*} 
		\end{proof}
		We are now in position to prove Proposition~\ref{prop:mobius} by combining the above statements.
		\begin{proof}[Proof of Proposition~\ref{prop:mobius}]
			From the above we get the equality:
			\begin{equation}
				\left(\frac{d}{dt}\left(f_t'(z)\right)^{s}\W s\circ f_t(z)\right)\rvert_{ t=0}=\left(\frac{d}{dt}\W s\left[\Phi\circ f_t+Q\ln\norm{f_t'}\right](z)\right)\rvert_{ t=0}.
			\end{equation}
			In particular if we take $v(z)=x$, $v(z)=\lambda z$ and $v(z)=xz^2$ for some $x\neq 0$ in $\C$ and $\lambda>0$, the latter is nothing but the infinitesimal formulation of Equation~\eqref{eq:cov_current} in the case where the transformation under consideration is respectively a translation, a dilation or an inversion. We conclude since the group of M\"obius transforms is generated by such transforms.
		\end{proof}
		
		\begin{remark}
			Providing a similar geometric interpretation for the action of the other currents generating the $W$-algebra, beyond the stress-energy tensor, remains an interesting question. 
		\end{remark}
		
		\subsubsection{Unitarity of the representation}
        Assume that $\vecgen\in\V_{+,\bm c}$. We define the Hermitian adjoint $Y(\vecgen;z)^*$ of the vertex operator $Y(\vecgen;z)$ to be the unique element of $\endv[[z,z^{-1}]]$ such that for any $F,G$ in $\mc C_\infty$, we have $\ps{Y(\vecgen;z)F,G}_2=\ps{F,Y(\vecgen;z)^*G}_2$, viewed as an equality in $\C((z))$. The same definition applies to $\Phi$ by a straightforward adaptation.
        \begin{lemma}
            For any $\vecgen\in\V_{+,\bm c}$, $Y(\vecgen;z)^*$ thus defined coincides with the adjoint of $Y(\vecgen;z)$ viewed as an unbounded and densely defined operator. It is obtained by the identity:
            \begin{equation}\label{eq:inv_phi}
			\Phi^*=\Phi\circ\theta-2 \mathfrak{Re}(Q)\ln\norm{\cdot}
		\end{equation}
        \end{lemma}
        \begin{proof}
            This follows from Equation~\eqref{eq:adj_A}.
        \end{proof}
		In particular $\gamma\in(0,\sqrt 2)$ (so that $Q\in\a$), the formal series $\Phi(z)$ satisfies the same rule as the GFF $\X$ under the action of the inversion map $\theta$. Likewise for any positive integer $m$, $\gamma\in(0,\sqrt 2)$, and the M\"obius transform $\psi:z\mapsto\frac1z$, we have the equality $\left(\partial^m\Phi(z)\right)^*=\partial^m\left(\Phi\circ\psi+ Q\ln\norm{\psi'}\right)(\bar z)$.
        Combining this property with naturally leads to unitarity of the representation thus constructed:
		\begin{proposition}\label{prop:unitary}
			Assume that $\gamma\in(0,\sqrt 2)$, let $s$ be any spin of $\g$ and $n$ be an integer. Then
			\begin{equation}
				\left(\W{s}_n\right)^*=(-1)^s\W{s}_{-n}.
			\end{equation} 
		\end{proposition}
		\begin{proof}
			According to the previous discussion, the current $\W {d+1}$ has an expansion of the form
			\[
			\Wb^{(s)}(z)=\sum_{\substack{\lambda\in\mc T_1\\ \norm{\lambda}=d+1}}:\Wb^{(s)}_\lambda(\partial^{\lambda_1}\Phi(z),\cdots,\partial^{\lambda_l}\Phi(z)):,
			\]
			where the multilinear forms $\Wb^{(s)}_\lambda$ have coefficients in $\a$.
			Now in agreement with Equation~\eqref{eq:inv_phi}
			\begin{align*}
				\left(\Wb^{(s)}[\Phi]\right)^*(\bar z)=\Wb^{(s)}[\Phi\circ\psi +Q\ln\norm{\psi'}](\bar z)
			\end{align*}
			where $\psi:z\mapsto\frac1z$ is a M\"obius transform of the sphere.
			By the transformation rule for $\Wb^{(s)}$ under the action of such maps from Proposition~\ref{prop:mobius}, we can expand the result in modes and get
			\begin{align*}
				\sum_{n\in\Z}\left(\W{s}_n\right)^*\bar z^{-n-s}&=\left(\frac{-1}{\bar z^2}\right)^{s}\sum_{n\in\Z}\W s_n\bar z^{n+s}=(-1)^s\sum_{n\in\Z}\W s_{-n}\bar z^{-n-s}.
			\end{align*}
			We conclude that, as desired, $\left(\W s_n\right)^*=(-1)^s\W s_{-n}$.
		\end{proof}
		\begin{remark}
			Up to renormalization and depending on the convention we may rather ask that $\left(\W s_n\right)^*=\W s_{-n}$. We stress that the proof does not rely on the particular representation chosen but is only based on the fact that the Heisenberg algebra satisfies Equation~\eqref{eq:adj_A}.
		\end{remark}
		


		
		\section{Some implications for free-field correlation functions}\label{sec:last}
		In this concluding section we derive some consequences of the framework described above and present some future directions of work as well as heuristics. Namely we first discuss the probabilistic interpretation of the algebraic setting described above. We then discuss free-field correlation functions in relation to Ward identities expressing the $W$-symmetry satisfied by such quantities. As we will show, this has strong implications on $\g$-Dotsenko-Fateev integrals.
		

		\subsection{Probabilistic representation of the $W$-algebra and free-field correlation functions}
		The above definition of the $W$-algebra associated to the free-field theory relies on a characterization of the space of states based on screening operators. Despite being not completely explicit, it still provides deep information on the free-field theory as we now explain.		
		
		\subsubsection{Probabilistic representation of the $W$-algebra}
		We have provided in Theorem~\ref{thm:def_W} a probabilistic construction of the $W$-algebra associated to $\g$. The defining property of the $W$-algebra $\mc W_{\bm c}(\g)$ in terms of kernel of screening operators is not explicit but still allows to derive the following property, instrumental in the derivation of Ward identities for free-field correlation functions:
		\begin{lemma}\label{lemma:derWV}
			For any spin $s$ of $\g$ and $1\leq i\leq r$, there exist multilinear forms $\W{s}_{i,j}$ such that for any $1>\norm{z}\geq\norm{w}\geq\delta$ and $F$ in $\mc F_\delta$: 
			\begin{equation}
				U_0\left(\W {s}(z)V_{\gamma e_i}(w)F(\X)\right)=\partial_w\left[\sum_{j=1}^{s-1}\frac{1}{(z-w)^{j}} U_0\left(:\W{s}_{i,j}(w)V_{\gamma e_i}(w):F(\X)\right)\right]+reg.
			\end{equation} 
			where $reg.$ has a well-defined limit (in $e^{\beta\norm{c}}\Lro$ for $\beta$ large enough) as $z\to w$.
		\end{lemma}
		\begin{proof}
			For any $\alpha$ the OPE between $\mc V_{\alpha}(w)$ and $\W{s}(z)$ can be written for $\norm{z}>\norm{w}$ as
			\[
			\W{s}(z)\mc V_{\alpha}(w)=\sum_{j=1}^{s}\frac{\W{s}_{j-s}\mc V_{\alpha}(w)}{(z-w)^{j}}+reg.
			\]
			where the descendants$\W{s}_{j-s}\mc V_{\alpha}(w)$ can be written in terms of multilinear forms in derivatives of $\Phi$.
			Now it follows from the defining property of the $W$-algebra that if we specialize to $\alpha=\gamma e_i$ then the OPE between $\mc V_{\gamma e_i}(w)$ and $\W{s}(z)$ is a total derivative in $w$ plus a regular term, so that
			\[
			\W{s}(z)\mc V_{\gamma e_i}(w)=\partial_w\left[\sum_{j=1}^{s-1}\frac{:\W{s}_{ij}[\Phi]\mc V_{\alpha}(w):}{(z-w)^{j}}\right]+reg.
			\]
			for some $\W{s}_{ij}$. Our claim is now a direct consequence of the probabilistic interpretation, using \textit{e.g.} Proposition~\ref{prop:ope_prob} and Theorem~\ref{thm:def_W}, of the above equality.
		\end{proof}

		\subsubsection{Definition of the correlation functions} 
		Let $\alpha_1,\cdots,\alpha_p,\beta_1,\cdots,\beta_q$ be in $\h$ and take distinct points $z_1,\cdots,z_p,w_1,\cdots,w_q$ in $\C$ with $\norm{z_1}<\cdots<\norm{z_p}<1<\norm{w_{1}}<\cdots<\norm{w_q}$.
		The probabilistic definition of the free-field correlation functions on the Riemann sphere is formally given by writing
		\begin{equation}
			\ps{\prod_{k=1}^NV_{\alpha_k}(z_k)}_{\gamma}=\left(\psi_{\bm\alpha}(\bm z),\psi_{\bm\beta}(\bm w)\right)_\D\qt{where} \psi_{\bm\alpha}(\bm z)\coloneqq U_0\left(\prod_{k=1}^NV_{\alpha_k}(z_k)\right).
		\end{equation}
		However the latter is actually infinite due to the divergence in the constant mode $c$. This accounts for the fact that we need to impose a \textit{neutrality condition} to make sense of the above writing:
		\begin{equation}\label{eq:neutrality}
			\sum_{k=1}^N\alpha_k-2Q=0.
		\end{equation}
		Under this assumption the free-field correlation functions are defined as the limit:
		\begin{equation}
			\begin{split}
				&\ps{\prod_{k=1}^NV_{\alpha_k}(z_k)}_{\gamma}\coloneqq \lim\limits_{\eps\to0} e^{-2\ps{Q,c}}\expect{\prod_{k=1}^NV_{\alpha_k,\eps}(z_k)}=\prod_{k=1}^N\norm{z_k}_+^{-4\Delta_{\alpha_k}}\prod_{k<l}e^{\ps{\alpha_k,\alpha_l}G(z_k,z_l)}.
			\end{split}
		\end{equation}
		Using Equation~\eqref{eq:green} (see also the proof of Lemma~\ref{lemma:cov_correl} below), we see that under the neutrality condition~\eqref{eq:neutrality}, the correlation functions do not depend on the background metric $g$ in that
		\begin{equation}\label{eq:correl_0}
			\begin{split}
				&\ps{\prod_{k=1}^NV_{\alpha_k}(z_k)}_{\gamma}=\prod_{k<l}\norm{z_k-z_l}e^{-\ps{\alpha_k,\alpha_l}}.
			\end{split}
		\end{equation}
		
		For future reference, let us state a basic property of such correlation functions associated to conformal covariance, corresponding to the one that can be found \textit{e.g.} in~\cite{DKRV, Toda_construction}. To this end we introduce the shorthand	$\X^g\coloneqq \X-2Q\ln\norm{\cdot}_+$ (in particular $\X^g$ coincides with $\X$ over $\D$).
		\begin{lemma}\label{lemma:cov_correl}
			Assume the neutrality condition~\eqref{eq:neutrality} to hold.
			Let $\psi$ be a M\"obius transform of the sphere $\psi$ and $F=F(\phi)$ be continuous, bounded over $\mathrm{H}^{-1}(\hat\C,g)$, and depending only on $\phi-m_g(\phi)$.
			\begin{equation}
				\expect{F\left(\X^g\right)\prod_{k=1}^NV_{\alpha_k}(z_k)}=\prod_{k=1}^N\norm{\psi'(z_k)}^{2\Delta_{\alpha_k}}\expect{F(\X^g\circ\psi+Q\ln\norm{\psi'})\prod_{k=1}^NV_{\alpha_k}(\psi(z_k))}.
			\end{equation}
		\end{lemma}
		\begin{proof}
			To start with using Girsanov's theorem~\ref{thm:girsanov} we have
			\begin{align*}
				\expect{F(\X^g)\prod_{k=1}^NV_{\alpha_k}(z_k)}=\prod_{k=1}^N\norm{z_k}_+^{-4\Delta_{\alpha_k}}\prod_{k<l}e^{\ps{\alpha_k\alpha_l}G(z_k,z_l)}\expect{F\left(\X^g+\sum_{k=1}^N\alpha_k G(z_k,\cdot)\right)}.
			\end{align*}
			Besides thanks to Equation~\eqref{eq:green} we can regroup metric dependent terms:
			\begin{align*}
				\prod_{k=1}^N\norm{z_k}_+^{-4\Delta_{\alpha_k}+\sum_{l\neq k}\ps{\alpha_k,\alpha_l}}\prod_{k<l}\norm{z_k-z_l}^{-\ps{\alpha_k\alpha_l}}\expect{F\left(\X^g+\sum_{k=1}^N\alpha_k G(z_k,\cdot)\right)}.
			\end{align*}
			We conclude that under the neutrality condition~\eqref{eq:neutrality}
			\begin{align*}
				\expect{F(\X^g)\prod_{k=1}^NV_{\alpha_k}(z_k)}=\prod_{k<l}\norm{z_k-z_l}^{\ps{\alpha_k\alpha_l}}\expect{F\left(\X^g+\sum_{k=1}^N\alpha_k G(z_k,\cdot)\right)}.
			\end{align*}
			Let us now denote $\phi(z)\coloneqq\ln\frac{\norm{\psi'(z)}^2\norm{\psi(z)}^{-4}_+}{\norm{z}_+^{-4}}\cdot$ Along the same lines
			\begin{align*}
				&\expect{F\left(\X^g\circ\psi+Q\ln\norm{\psi'}\right)\prod_{k=1}^NV_{\alpha_k}(\psi(z_k))}=\\
				&\prod_{k=1}^N\norm{\psi(z_k)}_+^{-4\Delta_{\alpha_k}}\prod_{k<l}e^{\ps{\alpha_k\alpha_l}G(\psi(z_k),\psi(z_l))}\expect{F\left(\X^g\circ\psi+Q\ln\norm{\psi'}+\sum_{k=1}^N\alpha_kG(\psi(z_k),\psi(\cdot))\right)}.
			\end{align*}
			Using the transformation rule for the GFF under M\"obius transforms~\eqref{eq:Green_Mobius} yields
			\begin{align*}
				&\prod_{k=1}^N\norm{\psi'(z_k)}^{2\Delta_{\alpha_k}}\expect{F\left(\X^g\circ\psi+Q\ln\norm{\psi'}\right)\prod_{k=1}^NV_{\alpha_k}(\psi(z_k))}=\\
				&\prod_{k=1}^Ne^{\Delta_{\alpha_k}\phi(z_k)}\prod_{k<l}e^{-\frac14\ps{\alpha_k,\alpha_l}(\phi(z_k)+\phi(z_l))}\prod_{k=1}^N\norm{z_k}_+^{-4\Delta_{\alpha_k}}\prod_{k<l}e^{\ps{\alpha_k\alpha_l}G(z_k,z_l)}\times\\
				&\hspace{2cm}\expect{F\left(\X\circ\psi-2Q\ln \norm{\cdot}_++\frac Q2\ln\phi-\frac14\sum_{k=1}^N\alpha_k(\phi(z_k)+\phi(\cdot))+\sum_{k=1}^N\alpha_kG(z_k,\cdot)\right)}.
			\end{align*}
			Using the neutrality condition the latter reduces to
			\begin{align*}
				&\prod_{k=1}^N\norm{\psi'(z_k)}^{2\Delta_{\alpha_k}}\expect{F\left(\X^g\circ\psi+Q\ln\norm{\psi'}\right)\prod_{k=1}^NV_{\alpha_k}(\psi(z_k))}=\\
				&\prod_{k<l}\norm{z_k-z_l}^{-\ps{\alpha_k\alpha_l}}\expect{F\left(\X\circ\psi-2Q\ln \norm{\cdot}_++\sum_{k=1}^N\alpha_kG(z_k,\cdot)-\frac14\sum_{k=1}^N\alpha_k\phi(z_k))\right)}.
			\end{align*}
			The proof is complete using Equation~\eqref{eq:GFF_Mobius} since $F$ depends only on $\X-m_g(\X)$ so that we can discard the constant terms $m_{g_\psi}(\X)$ and $\sum_{k=1}^N\alpha_k\phi(z_k)$.
		\end{proof}

		This neutrality condition is often relaxed by demanding the following weaker assumption to be satisfied for some non-negative integers $n_1,\cdots,n_r$ 
		\begin{equation}\label{eq:neutrality_gen}
			\sum_{k=1}^N\alpha_k-2Q=-\sum_{i=1}^r n_i \gamma e_i.
		\end{equation}
		Free-field correlation functions are then higher-rank generalizations of Dotsenko-Fateev integrals, that we call \textit{$\g$-Dotenko-Fateev integrals}. They are given, provided that the following is finite, by
		\begin{equation}\label{eq:dot_fat}
			\begin{split}
				\ps{\prod_{k=1}^NV_{\alpha_k}(z_k)}_{\gamma}\coloneqq\int_{\C^{n_1}}\cdots\int_{\C^{n_r}}&\ps{\prod_{\substack{1\leq i\leq r\\1\leq j\leq m_i}}V_{\gamma e_i}(x_{i,j})\prod_{k=1}^NV_{\alpha_k}(z_k)}_{\gamma}\prod_{\substack{1\leq i\leq r\\1\leq j\leq n_i}}\d^2x_{i,j}
			\end{split}
		\end{equation}
        where the integrand satisfies the neutrality condition~\eqref{eq:neutrality} and is thus defined by Equation~\eqref{eq:correl_0}.
		When $\gamma\in(0,\sqrt 2)$ such integrals are finite as soon as the following requirement is met for $1\leq i\leq r$\footnote{This can be seen by interpreting these integrals as integer moments of correlated Gaussian Multiplicative Chaos measures, which exist as soon as the assumptions from~\cite[Lemma 4.1]{Toda_construction} are satisfied.}:
		\begin{equation}\label{eq:fin_int}
			\begin{split}
				& \Re\left(\ps{\alpha_k,e_i}\right)<\frac2\gamma\qt{and}n_i<\frac4{\gamma^2\norm{e_i}^2}\min\Big(1,\Re\left(\ps{Q-\alpha_k,\gamma e_i}\right)\Big)\qt{for all $1\leq k\leq N$.}
			\end{split}
		\end{equation}
		\begin{remark}
		    As discussed in Section~\ref{sec:VOA} the free-field correlation functions defined under the weaker neutrality condition~\eqref{eq:neutrality} enjoy Heisenberg symmetry.
            This is no longer the case under the general neutrality condition~\eqref{eq:neutrality_gen}. However the extra Vertex Operators that appear in these $\g$-Dotsenko-Fateev integrals are the ones used to define the screening operators thanks to which $W$-algebras are constructed. This accounts for the fact that this extension of the free-field correlation functions breaks the Heisenberg symmetry but still preserves $W$-symmetry. This explains why the statement of the Ward identities~\ref{thm:ward_gen} below is made in terms of $W$-currents rather than Heisenberg ones.
		\end{remark}
		
		\subsubsection{Definition of the W-descendants}
		Before actually proving Ward identities we define the descendant fields associated to the Vertex Operators. To this end, for $\norm{z}>\norm{w}$ and any $s=d+1$:
		\begin{equation}
			\W s(z)\mc V_\alpha(w)=\sum_{i=1}^s\frac{\W s_{i-s}\mc V_\alpha(w)}{(z-w)^i}+reg.
		\end{equation}
		where the descendant fields admit an expansion of the form (for some $l(\lambda)$-linear forms $\Wb^{(s)}_{i-s,\lambda}$)
		\begin{equation}
			\W s_{i-s}\mc V_\alpha(w)=\sum_{\substack{\lambda\in\mc T_1\\ \norm{\lambda}=i-s}}:\Wb^{(s)}_{i-s,\lambda}(\partial^{\lambda_1}\Phi(w),\cdots,\partial^{\lambda_l}\Phi(w))\mc V_\alpha(w):.
		\end{equation}
		By combining Proposition~\ref{prop:VO_prim} above with Theorem~\ref{thm:VOA_Proba} we see that if we set 
		\begin{equation}
			\W s_{i-s}V_\alpha(z)\coloneqq	\sum_{\substack{\lambda\in\mc T_1\\ \norm{\lambda}=i-s}}:\Wb^{(s)}_{i-s,\lambda}(\partial^{\lambda_1}\X(z),\cdots,\partial^{\lambda_l}\X(z))V_\alpha(z):
		\end{equation} 
		then under the assumptions of Proposition~\ref{prop:VO_prim}:
		\begin{equation}
			\begin{split}
				&U_0\left(\prod_{k=1}^N\W {s_k}_{i_k-s_k}V_{\alpha_k}(z_k)F(\X)\right)= \W {s_1}_{i_1-s_1} \mc V_{\alpha_1}(z_1)\cdots \W {s_N}_{i_N-s_N}\mc V_{\alpha_N}(z_N)U_0F.
			\end{split}
		\end{equation}	
		
		As such we can define correlation functions with descendant fields by setting for any spin $s$ of $\g$, under the neutrality condition~\eqref{eq:neutrality_gen} and assuming Equation~\eqref{eq:fin_int} to hold:
		\begin{equation}\label{eq:W_correl_gen}
			\begin{split}
				&\ps{\W s(z)\prod_{k=1}^NV_{\alpha_k}(z_k)}_{\gamma,\eps}\coloneqq\\
				&\int_{\C_\eps^{n_1}}\cdots\int_{\C_\eps^{n_r}}\ps{\W s(z)\prod_{\substack{1\leq i\leq r\\1\leq j\leq m_i}}V_{\gamma e_i}(x_{i,j})\prod_{k=1}^NV_{\alpha_k}(z_k)}_{\gamma}\prod_{\substack{1\leq i\leq r\\1\leq j\leq n_i}}\d^2x_{i,j}
			\end{split}
		\end{equation}
		 where $\C_\eps\coloneqq \C\setminus B(z,\eps)\bigcup_{k=1}^NB(z_k,\eps)$. Likewise, regularized descendant fields are defined by
		\begin{equation}\label{eq:Wdesc_correl}
			\begin{split}
				&\ps{\W s_{-n}V_{\alpha_1}(z_1)\prod_{k=2}^NV_{\alpha_k}(z_k)}_{\gamma,\eps}\coloneqq \\
				&\int_{\C_\eps^{n_1}}\cdots\int_{\C_\eps^{n_r}}\ps{\W s_{-n}V_{\alpha_1}(z_1)\prod_{\substack{1\leq i\leq r\\1\leq j\leq m_i}}V_{\gamma e_i}(x_{i,j})\prod_{k=2}^NV_{\alpha_k}(z_k)}_{\gamma}\prod_{\substack{1\leq i\leq r\\1\leq j\leq n_i}}\d^2x_{i,j}.
			\end{split}
		\end{equation}
        The integrands satisfy the neutrality condition~\eqref{eq:neutrality} and are defined by~\eqref{eq:correl_0} (see also Lemma~\ref{lemma:cov_correl}). 
		
		

		
		
		\subsection{Free-field correlation functions and Ward identities}
		The Ward identities are a fundamental manifestation of the symmetries of the model at the level of correlation functions. 
		We prove here that the free-field correlation functions defined above are solutions of such Ward identities.
		
		\subsubsection{Ward identities: a first take}
		To start with we only assume the weaker neutrality condition~\eqref{eq:neutrality} to hold. Then Ward identities take the form:
		\begin{proposition}\label{prop:ward}
			Take $s$ a spin of $\g$ and $\alpha_1,\cdots,\alpha_N$ in $\h$ satisfying the neutrality condition~\eqref{eq:neutrality}. Then for $z,z_1,\cdots,z_n$ distinct the following Ward identity holds true:
			\begin{equation}\label{eq:ward_free}
				\ps{\W s(z)\prod_{k=1}^NV_{\alpha_k}(z_k)}_{\gamma}=\sum_{k=1}^N\sum_{i=1}^s\frac{1}{(z-z_k)^i}\ps{\W s_{i-s}V_{\alpha_k}(z_k)\prod_{l\neq k}^NV_{\alpha_l}(z_l)}_{\gamma}.
			\end{equation}
			More generally under the same assumptions
			\begin{equation}
				\ps{\W s_{-s}V_{\alpha_1}(z_1)\prod_{k=2}^NV_{\alpha_k}(z_k)}_{\gamma}=\sum_{k=2}^N\sum_{i=1}^s\frac{1}{(z_1-z_k)^i}\ps{V_{\alpha_1}(z_1)\W s_{i-s}V_{\alpha_k}(z_k)\prod_{l\neq k}^NV_{\alpha_l}(z_l)}_{\gamma}.
			\end{equation}  
		\end{proposition}
		\begin{proof}
			By Lemma~\ref{lemma:cov_correl} and Proposition~\ref{prop:mobius} we can always assume that $z,z_1,\cdots z_N$ belong to $\D$, with $\norm{z}>\norm{z_1},\cdots,\norm{z_n}$.
			Using Gaussian integration by parts~\eqref{eq:Gauss_IPP} for any $\bm u\in\h^p$ and $\bm n\in\N_{>0}^p$
            \begin{align*}
                &\expect{:\prod_{i=1}^p\partial^{n_i}\ps{u_i,\X(z)}:\prod_{k=1}^NV_{\alpha_k}(z_k)}=\sum_{k=1}^N\ps{u_1,\alpha_k}\partial^{n_1}_zG(z,z_k)\expect{:\prod_{i=2}^p\partial^{n_i}\ps{u_i,\X(z)}:\prod_{k=1}^NV_{\alpha_k}(z_k)}.
            \end{align*}
            Since $\partial_zG(z,z_k)=\frac{1}{2(z_k-z)}$, $\ps{\W s(z)\prod_{k=1}^NV_{\alpha_k(z_k)}}_\gamma$ is a linear combination of the 
			\[
			\frac{1}{(z-z_1)^{i_1}\cdots(z-z_n)^{i_n}}\ps{\prod_{k=1}^NV_{\alpha_k(z_k)}}_\gamma\qt{with}i_1+\cdots+i_n=s.
			\]
			Now thanks to Theorem~\ref{thm:VOA_Proba} and Proposition~\ref{prop:VO_prim_bis}:
			\begin{equation*}
				\ps{\W s(z)\prod_{k=1}^NV_{\alpha_k(z_k)}}_{\gamma}=\cav \W s[\Phi](z)\mc V_{\alpha_1}(z_1)\cdots\mc V_{\alpha_N}(z_N)\vac_{\V},
			\end{equation*}
			so that we can use the OPE between a current $\W s$ and a vertex operator, which is of the form
			\begin{equation}\label{eq:OPE_WV}
				\W s(z)\mc V_{\alpha_1}(z_1)=\sum_{i=1}^s\frac{1}{(z-z_1)^i}\W s_{i-s}\mc V_{\alpha_1}(z_1)+reg.
			\end{equation}
			where the term $reg.$ denotes a quantity that is continuous in a neighborhood of $z_1$ (this expansion follows from the fact that $\W s_{n}\prim=0$ for any $n\geq 1$, which is a consequence of $\A_n\prim=0$ for all $n\geq 1$). 
			In particular this allows to identify the residues at the poles $z=z_1$ in the probabilistic expansion (using Gaussian integration by parts) of $\ps{\W s(z)\prod_{k=1}^NV_{\alpha_k}(z_k)}_{\gamma}$. We can proceed in the same way for the other insertions $z_k$, concluding for the validity of Equation~\eqref{eq:ward_free}.
			
			The second assertion is a consequence of the fact that in the OPE~\eqref{eq:OPE_WV} the regular term as $z\to z_1$ is given by $\W s_{-s}\mc V_\alpha(z)+l.o.t.$ where the lower order term vanishes as $z\to z_1$. We can then use the first case to deduce the second one by identification of this regular term.
		\end{proof}

		\subsubsection{Ward identities: the general case}
		We have just shown that under the weaker neutrality condition~\eqref{eq:neutrality} the free-field correlation functions satisfy Ward identities, encoding the symmetries of the model. We extend this result to the general neutrality condition~\eqref{eq:neutrality_gen}:
		\begin{theorem}\label{thm:ward_gen}
			Assume the neutrality condition~\eqref{eq:neutrality_gen} to hold, and that the corresponding integral~\eqref{eq:dot_fat} is absolutely convergent. Then for $z,z_1,\cdots,z_n$ distinct and any spin $s$ of $\g$ the limits as $\eps\to0$ of Equations~\eqref{eq:W_correl_gen} and~\eqref{eq:Wdesc_correl} exist and are such that
			\begin{equation}
				\ps{\W s_{-s}V_\alpha(z)\prod_{k=1}^NV_{\alpha_k}(z_k)}_{\gamma}=\sum_{k=1}^N\sum_{i=1}^s\frac{1}{(z-z_k)^i}\ps{V_\alpha(z)\W s_{i-s}V_{\alpha_k}(z_k)\prod_{l\neq k}V_{\alpha_l}(z_l)}_{\gamma}.
			\end{equation}
		\end{theorem}
        By taking $\alpha=0$ this specializes to the Ward identities for the $W$-current:
        \begin{equation}
			\ps{\W s(z)\prod_{k=1}^NV_{\alpha_k}(z_k)}_{\gamma}=\sum_{k=1}^N\sum_{i=1}^s\frac1{(z-z_k)^i}\ps{\W s_{i-s}V_{\alpha_k}(z_k)\prod_{l\neq k}V_{\alpha_l}(z_l)}_{\gamma}.
		\end{equation}
		\begin{proof}
			Following Proposition~\ref{prop:ward} we know that
			\begin{align*}
				\ps{\W s(z)\prod_{k=1}^NV_{\alpha_k}(z_k)}_{\gamma,\eps}&=\sum_{k=1}^N\left(\sum_{i=1}^s\frac{\left(\W s_{i-s}\right)^{(z_k)}}{(z-z_k)^i}\right)\ps{\prod_{l=1}^NV_{\alpha_l}(z_l)}_{\gamma,\eps}\\
				&+\int_{\C_\eps^p} \sum_{k=1}^p\left(\sum_{i=1}^s\frac{\left(\W s_{i-s}\right)^{(x_k)}}{(z-x_k)^i}\right)\ps{\prod_{m=1}^pV_{\beta_m}(x_m)V\prod_{l=1}^NV_{\alpha_l}(z_l)}_{\gamma}\prod_{m=1}^p\d^2x_m
			\end{align*}
			where $\beta_m$ is one of the $\gamma e_i$. Now thanks to Lemma~\ref{lemma:derWV} the integrand is a total derivative. Hence using Stokes' formula the last line is given by (with $\Wb^{(s)}_{\beta_k,j}V_{\beta_k}\coloneqq \Wb^{(s)}_{i,j}V_{\gamma e_i}$ for $\beta_k=\gamma e_i$)
			\begin{equation}\label{eq:to_prove}
				\sum_{k=1}^p\int_{\C_\eps^{p-1}} \prod_{m\neq k}\d^2x_m\oint_{\partial\C_\eps}\frac{\i\d\bar x_k}{2}\sum_{j=1}^{s-1}\frac{1}{(z-x_k)^{j}}\ps{:\Wb^{(s)}_{\beta_k,j}V_{\beta_k}(x_k):\prod_{m\neq k}V_{\beta_m}(x_m)\prod_{l=1}^NV_{\alpha_l}(z_l)}_{\gamma}.
			\end{equation}
			We prove in Lemma~\ref{lemma:int_zero} below that this boundary term vanishes in the limit. This shows that 
			\begin{align*}
				\lim\limits_{\eps\to 0}\ps{\W s(z)\prod_{k=1}^NV_{\alpha_k}(z_k)}_{\gamma,\eps}-\sum_{k=1}^N\left(\sum_{i=1}^s\frac{\left(\W s_{i-s}\right)^{(z_k)}}{(z-z_k)^i}\right)\ps{\prod_{l=1}^NV_{\alpha_l}(z_l)}_{\gamma,\eps}=0.
			\end{align*}
			Therefore to finish with the proof it remains to show that the descendants appearing in the above expansion do indeed admit a well-defined limit as $\eps\to0$. This is the content of Lemma~\ref{lemma:def_desc} below.
			The proof when we consider $\W s_{-s}V_\alpha(z)$ instead of $\W s(z)$ remains the same by Proposition~\ref{prop:ward}.	
		\end{proof}
		
		We now show that the two limits in Equations~\eqref{eq:W_correl_gen} and~\eqref{eq:Wdesc_correl} are indeed well-defined. And to start with we explain how to discard the remainder terms from the proof of Theorem~\ref{thm:ward_gen}.
		\begin{lemma}\label{lemma:int_zero}
			The integral given by Equation~\eqref{eq:to_prove} vanishes as $\eps\to0$.
		\end{lemma}
		\begin{proof}
            The integrals in $x_k$ that appear in Equation~\eqref{eq:to_prove} range over $\partial B(z,\eps)\bigcup_{l=1}^N\partial B(z_l,\eps)$. For the sake of simplicity we focus here on the ones over $\partial B(z,\eps)$, the reasoning remaining valid for the other integrals.  
	        Without loss of generality we assume that $k=1$ and set
            \[
                F(x_1)\coloneqq\frac{1}{(z-x_1)^j}\ps{:\Wb^{(s)}_{\beta_1,j}V_{\gamma e_1}(x_1):\prod_{m=2}^pV_{\beta_m}(x_m)\prod_{l=1}^NV_{\alpha_l}(z_l)}_{\gamma}
            \]
            where $x_1$ belongs to $\partial\C_\eps$ and the $x_k$ are fixed and distinct in $\C_\eps$ so that $F$ is indeed well-defined. 
            Then, by Gaussian integration by parts~\eqref{eq:Gauss_IPP}, $F(x_1)$ is a linear combination of the
            \[
                \frac1{(z-x_1)^j}\prod_{l=2}^{N+p}\frac1{(x_1-x_l)^{n_l}}\prod_{1\leq k<l\leq N+p}\norm{x_k-x_l}^{-\ps{\beta_k,\beta_l}}
            \]
            where $x_{p+m}\coloneqq z_m$ and $\beta_{p+m}\coloneqq\alpha_m$, while $\bm n\in\N^{N+p-1}$ satisfies $\norm{\bm n}=s-j-1$. 
            
            Now, for $\rho>0$ small enough, let us split $\C_\eps$ between $\C_{\eps,\rho}\coloneqq \C_\eps\setminus B(z,\rho)$ and the annulus $A\coloneqq\left\{\eps<\norm{\cdot-z}<(1+\rho)\eps\right\}$. The integral over $\C_\eps^{p-1}$ becomes a sum of integrals over powers of $A$ and $\tilde\C_\eps$: without loss of generality we assume that $x_2,\cdots,x_q$ belong to $A$ and $x_{q+1},\cdots,x_p$ to $\tilde\C_\eps$.
            Over $A$, for $1\leq a\leq q$ and $q+1\leq l\leq N+p$ we use the power series expansion of $(x_a-x_l)^{-\delta}=(z-x_l)^{-\delta}\left(1-\frac{x_a-z}{x_l-z}\right)^{-\delta}$ (since $\norm{x_a-z}<\norm{x_l-z}$), which gives us
            \begin{align*}
                &\oint_{\partial B(z,\eps)}\int_{A^{q-1}}F(x_1)\frac{\i\d\bar x_1}{2}\prod_{a=2}^q\d^2x_a=\sum_{\bm n,\bm k,\tilde{\bm k}}c_{\bm n,\bm k,\tilde{\bm k}}(x_{q+1},\cdots,x_p)\oint_{\partial B(z,\eps)}\int_{A^{q-1}}\\
                &\prod_{l=1}^q(z-x_l)^{k_l}(\bar z-\bar x_l)^{\tilde k_l}\frac1{(z-x_1)^j}\prod_{a=2}^{q}\frac1{(x_1-x_a)^{n_a}}\prod_{1\leq a<b\leq q}\norm{x_a-x_b}^{-\ps{\beta_a,\beta_b}}\frac{\i\d\bar x_1}{2}\prod_{a=2}^q\d^2x_a
            \end{align*}
            with the sum ranging over tuples of integers with $\norm{\bm n}=s-j-1$ and $l(\bm k),l(\tilde{\bm k})\leq q$.
            Here the coefficients $c_{\bm n,\bm k,\tilde{\bm k}}$ are such that the sum is absolutely convergent, and viewed as functions of $x_{q+1},\cdots,x_p$ are uniformly bounded over $\C_{\eps,\rho'}$ for any $\rho'>\rho$. Now the last integral is seen to be non-zero only if $\norm{\bm k}=\norm{\tilde{\bm k}}+\tilde{\norm{\bm n}}+j+1$, where $\tilde{\norm{\bm n}}=\sum_{a=1}^qn_a$. Under this assumption, a change of variables $x_a\leftrightarrow \eps x_a$ gives a scaling factor of
            $\eps^{2+2\norm{\tilde{\bm k}}+\delta}$ where $\delta=2q-\sum_{1\leq a<b\leq q}\ps{\beta_a,\beta_b}>0$ since the integral $\int_{A^{q-1}}\ps{V_{\beta_1}(z)\prod_{m=2}^pV_{\beta_m}(x_m)\prod_{l=1}^NV_{\alpha_l}(z_l)}_{\gamma}\prod_{m=2}^q\d^2x_m$ is absolutely convergent (which follows from absolute convergence of~\eqref{eq:dot_fat}). As a consequence we see that for any $\rho'>\rho$
            \[
                \oint_{\partial B(z,\eps)}\int_{A^q}\int_{\C_{\eps,\rho'}}F(x_1)\frac{\i\d\bar x_1}{2}\prod_{m=2}^p\d^2x_m=\mc O\left(\eps^2\right).
            \]
            Now for $x_a\in A$ and $x_m\in\C_{\eps,\rho}$ close to $\partial B(z,\rho)$ the possible singularity that may arise is purely artificial since it comes from the local expansion for $x_1,\cdots,x_q$ in a neighborhood of $z$, and can thus be easily lifted. This allows to extend the previous limiting behavior to the case $\rho'=\rho$, which thus allows us to conclude for the proof.
		\end{proof}
		The same reasoning allows to define correlation functions that contain a descendant field:
		\begin{lemma}\label{lemma:def_desc}
			Take $s$ a spin of $\g$ and $0\leq n\leq s-1$.
			Then under the neutrality condition~\eqref{eq:neutrality_gen} the following exists and is finite:
			\begin{equation}
				\ps{\W s_{-n}V_{\alpha_1}(z_1)\prod_{l=2}^NV_{\alpha_l}(z_l)}_{\gamma,\eps}\coloneqq\lim\limits_{\eps\to 0}\ps{\W s_{-n}V_{\alpha_1}(z_1)\prod_{l=2}^NV_{\alpha_l}(z_l)}_{\gamma,\eps}.
			\end{equation}
		\end{lemma}
		\begin{proof}
			We first justify that all we need to show is that for all $0\leq n\leq d$ and some $F_{i-s,l}^{(x_l)}$:
			\begin{equation}\label{eq:last_one}
				\begin{split}
					&\ps{:\W{s}_{-n}V_{\alpha_1}(z_1):\prod_{m=1 }^pV_{\beta_m}(x_m)\prod_{ k=2}^NV_{\alpha_l}(z_l)}_{\gamma}\\
					&=\sum_{l=1}^p\partial_{x_l}\left(\sum_{i=1}^{s-1}\frac{\left(F_{i-s,l}\right)^{(x_l)}}{(z_1-x_l)^i}\right)\ps{\prod_{m=1 }^pV_{\beta_m}(x_m)\prod_{k=1}^NV_{\alpha_k}(z_k)}_{\gamma,\eps}+reg.
				\end{split}
			\end{equation} 
			where $reg.$ denotes a term which is smooth for $z_1$ away from the $z_k$, $k\geq 2$ (and in particular near $x_l$ ofr $1\leq l\leq s-1$). Indeed we can then use Stokes' formula to transform the singularity around $x_l=z_1$ into an contour integral over $\partial B(z_1,\eps)$. This integral then vanishes as $\eps\to 0$ in the very same fashion as in the proof of Lemma~\ref{lemma:int_zero} above. Therefore it only remains that Equation~\eqref{eq:last_one} holds true. For this purpose we rely on the fact that, in agreement with Proposition~\ref{prop:VO_prim_bis}
			\begin{align*}
				\mc V_{\alpha,-}(z) \left(\W s_{-n}\left[\Phi_\alpha\right] \mc V_{\gamma e_i}[\Phi_\alpha](w)\right)\mc V_{\alpha,+}(z),\qt{with} \Phi_\alpha=\Phi+\alpha\ln\frac{1}{\norm{\cdot-z}}
			\end{align*}
			and where in $\mc V_{\alpha_-}$ (resp. $\mc V_{\alpha,+}$) we have gathered all creation (resp. annihilation) operators. Now we already know that $\Phi_\alpha$ satisfies the same OPEs as $\Phi$. As a consequence since Equation~\eqref{eq:last_one} holds true when $\alpha=0$ we deduce that it remains valid for any $\alpha$.
		\end{proof}

		\subsubsection{Global Ward identities}
		The local Ward identities express the local symmetries enjoyed by the correlation functions. Their \textit{global} counterparts, which depend on the underlying geometry, are obtained thanks to the covariance of the currents under global conformal transformations:
		\begin{theorem}\label{thm:ward_global}
			In the same setting as Theorem~\ref{thm:ward_gen}, for all $0\leq n\leq 2s-2$:
			\begin{equation}
				\sum_{k=1}^N\sum_{i=0}^{s-1}\binom{n}{i}z_k^{n-i}\ps{\W s_{i+1-s}V_{\alpha_k}(z_k)\prod_{l\neq k}V_{\alpha_l}(z_l)}_{\gamma}=0.
			\end{equation}
		\end{theorem}
		Such identities are crucial to derive BPZ-type differential equations for correlation functions, a key input in the computation of the structure constants of Liouville and Toda CFTs. This method for deriving the structure constants was introduced by Teschner~\cite{Teschner_DOZZ} and further developed in~\cite{FZZ, PT02, Hos, FaLi1}. The mathematical implementation of these techniques led to the proof of the DOZZ formula~\cite{KRV_DOZZ}, the computation of the structure constants for boundary Liouville CFT~\cite{ARS, ARSZ} and of the Fateev-Litvinov formula for the $\mathfrak{sl}_3$ Toda CFT~\cite{Toda_correl2}.
		\begin{proof}
			From the conformal covariance property of the free-field correlation functions (Lemma~\ref{lemma:cov_correl})
            \begin{align*}
				&\ps{\W s(z)\prod_{k=1}^NV_{\alpha_k}(z_k)}=\prod_{k=1}^N\norm{\psi'(z_k)}^{2\Delta_{\alpha_k}}\ps{\W s[\X\circ\psi+Q\ln\norm{\psi'}](z)\prod_{k=1}^NV_{\alpha_k}(\psi(z_k))}.
			\end{align*}
			In particular the conformal covariance of the $\W s$ current from Proposition~\ref{prop:mobius} thus leads to
			\begin{align*}
				&\ps{\W s(z)\prod_{k=1}^NV_{\alpha_k}(z_k)}=\psi'(z)^{s}\prod_{k=1}^N\norm{\psi'(z_k)}^{2\Delta_{\alpha_k}}\ps{\W s(\psi(z))\prod_{k=1}^NV_{\alpha_k}(\psi(z_k))}.
			\end{align*}
			Combining this equality with the local Ward identities as stated in Equation~\eqref{eq:ward_free} yields
			\begin{align*}
				&\sum_{k=1}^N\sum_{i=1}^s\frac1{(z-z_k)^i}\ps{\W s_{i-s}V_{\alpha_k}(z_k)\prod_{l\neq k}^NV_{\alpha_l}(z_l)}_{\gamma}=\psi'(z)^{s}\prod_{k=1}^N\norm{\psi'(z_k)}^{2\Delta_{\alpha_k}}\ps{\W s(\psi(z))\prod_{k=1}^NV_{\alpha_k}(\psi(z_k))}.
			\end{align*}
			In particular if we take $\psi(x)\coloneqq \frac1x$ then we see that as $z\to\infty$:
			\begin{align*}
				\sum_{k=1}^N\sum_{i=1}^s\frac1{(z-z_k)^i}\ps{\W s_{i-s}V_{\alpha_k}(z_k)\prod_{l\neq k}^NV_{\alpha_l}(z_l)}_{\gamma}&=(-z)^{-2s}\prod_{k=1}^N\norm{-z_k}^{-2\Delta_{\alpha_k}}\ps{\W s(\frac 1z)\prod_{k=1}^NV_{\alpha_k}(\frac1{z_k})}\\
				&\sim C z^{-2s}\qt{for some constant} C.
			\end{align*}
			In particular all the terms scaling like $z^{-n}$ for $1\leq n\leq 2s-1$ in the expansion of the left-hand side must vanish. One checks that these are given, as expected, by
			\begin{align*}
				z^{-n}\sum_{k=1}^N\sum_{i=0}^{s-1}\binom{n}{i}z_k^{n-i}\ps{\W s_{i+1-s}V_{\alpha_k}(z_k)\prod_{l\neq k}V_{\alpha_l}(z_l)}_{\gamma}.
			\end{align*}
		\end{proof}
		
		\subsection{An application to Dotsenko-Fateev integrals associated to $B_2$}\label{subsec:Mukhin_Varchenko} 
		The Ward identities satisfied by free-field correlation functions are very constraining, and translate in concrete terms the symmetries of the model. We now explain how additional comprehension of the representation theory of $W$-algebras provides additional constraints on them. For this we focus on the case where the underlying Lie algebra is $\g=B_2$, in which case we first provide some new results concerning the representation theory of the associated $W$-algebra (that we denote here $\mc WB_2$) and explain how this leads to integrability properties for Dotsenko-Fateev integrals associated to $B_2$. 
		
		\subsubsection{Some conventions}
		The Lie algebra $\mathfrak{so}_{3}\simeq B_2$ is not simply laced: it admits a long and a short root. By convention we choose to denote by $e_1$ the longest root and by $e_2$ the short root. The scalar product on these roots is then given by
		\begin{equation}
			\ps{e_1,e_1}=2,\quad\ps{e_2,e_2}=1,\quad\ps{e_1,e_2}=-1.
		\end{equation}
		They can be represented in the canonical basis of $\R^2$ by $e_1=\left(\sqrt{2},0\right)$ and $e_2=\left(-\frac{1}{\sqrt2},\frac{1}{\sqrt2}\right)$. The dual basis is given by $\omega_1^\vee=(\frac1{\sqrt 2},\frac1{\sqrt 2})$ and $\omega_2^\vee=(0,\sqrt2)$, and the fundamental weights are $\omega_1=\omega_1^\vee$ and $\omega_2=\frac12\omega_2^\vee$. We further introduce weights $h_i$ and $\tilde h_i$ by setting
		\begin{equation}
			\begin{split}
				h_1=\omega_2,\quad h_2=h_1-e_2,\quad h_3=h_2-e_1,\quad h_4=h_3-e_2\\
				\tilde h_1=\omega_1,\quad \tilde h_2=\tilde h_1-e_1,\quad \tilde h_3=\tilde h_2-2 e_2,\quad \tilde h_4=\tilde h_3-e_1.            
			\end{split}
		\end{equation}
		They satisfy $\sum_{i=1}^4h_i=\sum_{i=1}^4\tilde h_i=0$. All of these are depicted in Figure~\ref{fig:weyl} below.

		\begin{center}
			\includegraphics[height=6.5cm]{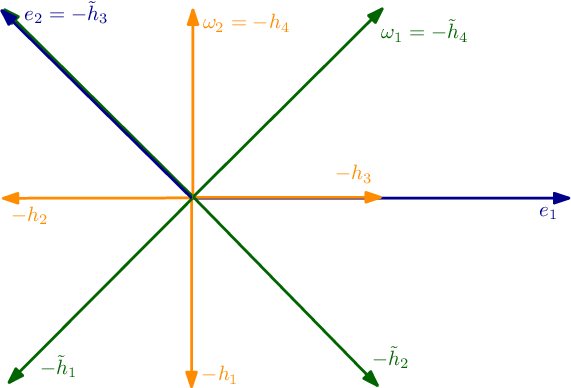}
			\captionof{figure}{The root system associated to $B_2$}
			\label{fig:weyl}
		\end{center}
		
		The exponents of $B_2$ are given by $1$ and $3$: for generic values of $\bm c$ the $\mc WB_2$ algebra is generated by two currents $\L$ and $\Wb$, with respective spins $2$ and $4$. The current of spin $2$ is the stress-energy tensor, while that of spin $4$ is associated to higher-spin symmetry; using Mathematica we find explicit expressions for them (gathered in Appendix~\ref{appendix:WB2}) in terms of derivatives of the field $\Phi$.
		
		\subsubsection{Representation theory of $\mc WB_2$}
		Let us consider the \textit{free-field} $W$-algebra modules $\Wmod$:
		\begin{equation}
			\begin{split}
				&\Wmod\coloneqq\bigoplus_{n\geq0}\left(\Wmod\right)^{(n)},\quad\left(\Wmod\right)^{(n)}\coloneqq\bigoplus_{\substack{\lambda_1\in\mc T_{s_1},\cdots\lambda_r\in\mc T_{s_r}\\ \norm{\lambda_1}+\cdots+\norm{\lambda_l}=n}}\text{span}\left\{\Wb^{(s_r)}_{-\lambda_r}\cdots\Wb^{(s_1)}_{-\lambda_1}\prim\right\}.
			\end{split}
		\end{equation}
		Likewise we can consider other $W$-algebras modules called \textit{Verma modules} $\Verma$. From~\cite[Equation (248)]{Arakawa_rep} (Verma modules are denoted there by $\bm{\mathrm M}(\gamma_{\bar \lambda})$) we know that such modules admit a Poincaré-Birkhoff-Witt basis in the sense that there is an isomorphism
		\begin{equation}
			\Verma\simeq \text{span}\left\{\bm w^{(s_r)}_{-\lambda_r}\cdots \bm w^{(s_1)}_{-\lambda_1},\quad \lambda_r,\cdots,\lambda_1\in\mc T_1\right\}.
		\end{equation}
		Then for generic (\textit{i.e.} for $\alpha$ in some dense open subset of $\h$) values of $\alpha$ this module is isomorphic to a Verma module $\Verma$ for the associated $W$-algebra (see for instance~\cite[Theorem 7.7.1]{Arakawa_rep}).
		However when $\alpha$ takes some specific values then this is no longer the case and the module $\Wmod$ is actually isomorphic to a quotient of this Verma module by its maximal proper submodule~\cite[Theorem 5.3.1]{Arakawa_rep}. This is due to the existence of so-called \textit{singular vectors}, which are the highest-weights that generate this proper submodule. For instance if we take $\prim=\vac$ then according to Theorem~\ref{thm:def_W} 
		\begin{equation}
			\begin{split}
				&\Wmod\simeq \bigslant{\Verma}{\mc I_0},\quad\text{$\mc I_0$ left ideal generated by the }\bm w_{-n}^{(s_i)}\text{ for }1\leq i\leq r\text{, } n\leq s_i-1,
			\end{split}
		\end{equation}
        and for generic values of $\bm c$.
		In agreement with the physics literature, we may call the corresponding vectors $\Wb_{-n}^{(s_i)}\vac$ \textit{null vectors} since they are indeed null:  $\Wb_{-n}^{(s_i)}\vac=0$ for $n\leq s_i-1$.
		
		Using the probabilistic framework introduced in this document we describe explicitly the structure of this submodule in the case where the underlying Lie algebra is $B_2$. Namely:
		\begin{theorem}\label{thm:Verma}
			Let $\Wmod$ be the free-field module generated by acting with the generators of the $W$-algebra $B_2$ on the primary $\prim$. 
			Then for generic value of $\bm c$ we have the following:
			\begin{itemize}
				\item Singular vectors at level $1$: for generic $\kappa\in\C$ and $\alpha=\kappa\omega_i^\vee$ for $i=1,2$, we have a complex
				\begin{equation}
					0\longrightarrow \prescript{\g}{}{\mc V_\beta}\overset{s} {\longrightarrow}\Verma\overset{\iota}{\longrightarrow} \Wmod\longrightarrow 0
				\end{equation}
				where $\beta=\alpha+\gamma e_2$ if $i=1$ and $\beta=\alpha+\gamma e_1$ if $i=2$, $\iota:\bm w_{-\lambda_1}\bm l_{-\lambda_2}\mapsto\Wb_{-\lambda_1}\L_{-\lambda_2}\prim$ and where $s$ is defined by setting 
				\begin{equation}\label{eq:deg1}
					\begin{split}
						&s(\vert\beta\rangle)=\frac\kappa\gamma\left(-3 \gamma ^2+\gamma  \kappa -8\right) \bm l_{-1} +\bm w_{-1}\qt{for}i=1,\\
						&s(\vert\beta\rangle)=\left(-4 \gamma ^2-\frac{32}{\gamma ^2}+4 \gamma  \kappa +\frac{12 \kappa }{\gamma }-2 \left(\kappa ^2+12\right)\right)\bm l_{-1}+\bm w_{-1}\qt{if} i=2.
					\end{split}
				\end{equation}
				In particular we have an embedding $\Wmod\xhookrightarrow{}  \bigslant{\Verma}{\prescript{\g}{}{\mc V_\beta}}$.
				\item Singular vectors at level $3$: if we specialize to $\alpha=-\gamma\omega_1$ or $\alpha=-\frac2\gamma\omega_2^\vee$ then
				\begin{equation}
					\Wmod\xhookrightarrow{} \bigslant{\Verma}{\mc I_3}
				\end{equation}
				where $\mc I_2$ is the left ideal generated by the image of~\eqref{eq:deg1} together with  
				\begin{equation}
					\begin{split}
						&\bm w_{-2}-\left(\frac{16}{\gamma ^2}+10\right) \bm l_{-1}^2- \left(4\gamma ^2+8\right) \bm l_{-2},\quad\frac{16 }{\gamma ^2}\bm l_{-1}^3+12 \bm_{-1}\bm l_{-2}-12 \bm l_{-3}+\bm w_{-3}\text{ for }\kappa=-\gamma,\\ 
						&\bm w_{-2}-\left(2\gamma ^2+14\right) \bm l_{-1}^2-\left(\frac{28}{\gamma ^2}+2\right) \bm l_{-2},\quad -2\gamma ^2\bm l_{-1}^3-12\bm l_{-1}\bm l_{-2}+12 \bm l_{-3}+\bm w_{-3}\text{ for }\kappa=-\frac2\gamma.
					\end{split}
				\end{equation}
				\item Singular vectors at level $4$: Finally let us take $\alpha=-\gamma\omega_2$ or $\alpha=-\frac2\gamma\omega_1$. Then we have
				\begin{equation}
					\Wmod\xhookrightarrow{} \bigslant{\Verma}{\mc I_4},
				\end{equation}
				where $\mc I_4$ is generated by the image of~\eqref{eq:deg1} together with the three relations
				\begin{equation}
					\begin{split}
						&\bm w_{-2}+\left(\frac{16}{\gamma ^2}+14\right) \bm l_{-1}^2+\left(\left(\frac{7 \gamma ^2}{2}+2\right) \bm l_{-2}\right);\quad\bm w_{-3}-\frac{32} {\gamma ^2}\bm l_{-1}^3-12 \bm l_{-1}\bm l_{-2}+\frac{3}{2} \left(\gamma ^2+4\right) \bm l_{-3};\\
						&\bm w_{-4}+\frac{64}{\gamma ^4} \bm l_{-1}^4+\frac{32}{\gamma ^2}\bm l_{-1}^2\bm l_{-2}-\frac{4 \left(\gamma ^2+8\right) }{\gamma ^2}\bm l_{-1}\bm l_{-3}+2\bm l_{-2}^2-\frac{1}{2} \left(\gamma ^2+16\right) \bm l_{-4}
					\end{split}
				\end{equation}
				if $\alpha=-\gamma\omega_2$, while for $\alpha=-\frac2\gamma\omega_1$:
				\begin{equation}
					\begin{split}
						-2\left(\gamma ^2+7\right) \bm l_{-1}^2+\left(-\frac{28}{\gamma ^2}-2\right) \bm l_{-2}+\bm w_{-2};\quad 4 \gamma ^2 \bm l_{-1}^3+12 \bm l_{-1}\bm l_{-2}-\frac{6 \left(\gamma ^2+2\right) \bm l_{-3}}{\gamma ^2}+\bm w_{-3}\\
						-\gamma ^4 \bm l_{-1}^4-4 \gamma ^2 \bm l_{-1}^2\bm l_{-2}+4\left(\gamma ^2+1\right) \bm l_{-1}\bm l_{-3}-2\bm l_{-2}^2-\frac{4 \left(\gamma ^4+4 \gamma ^2+7\right)}{\gamma ^2}\bm l_{-4}+\bm w_{-4}.
					\end{split}
				\end{equation}
			\end{itemize}
		\end{theorem}
		\begin{proof}
			The proof of this statement is based on explicit computations conducted based on the definition of the currents generating the $W$-algebra using the screening operators. Namely we first provide a free-field realization of the $\mc WB_2$ algebra which we write explicitly in terms of the probabilistic framework introduced above: these explicit expressions can be found in Appendix~\ref{appendix:WB2}. We then use this explicit construction to derive explicit expressions for the singular vectors. This is done using a Mathematica code that we describe in more details in Appendix~\ref{appendix:WB2}. A consequence of this code is the existence of such singular vectors, which in turn define a proper submodule of $\Verma$. Hence in the three cases considered we have a natural embedding $\Wmod\xhookrightarrow{} \bigslant{\Verma}{\mc I}$.
		\end{proof}

		\subsubsection{From representation theory to differential equations}
		Let us assume that $\alpha=\kappa\omega_1^\vee$ for $\kappa\in\R$. Then according to the previous statement we have an embedding of the free-field module $\Wmod$ into a quotient $\bigslant{\Verma}{\mc I}$, where $\mc I$ is the left ideal generated by a singular vector at the level one of the form $\bm w_{-1}+\lambda \bm l_{-1}$.
		This implies that we get a null vector for the free-field in the sense that
		\begin{equation}
			\Wb_{-1}\prim+\lambda\L_{-1}\prim =0.
		\end{equation}
		Using the probabilistic representation of the latter, as a straightforward consequence we get that under the neutrality condition~\eqref{eq:neutrality_gen} and as soon as the correlation functions make sense, we have
		\begin{equation}
			\ps{\Wb_{-1}V_\alpha(z)\prod_{k=2}^NV_{\alpha_k}(z_k)}_\gamma+\lambda\partial_z\ps{V_{\alpha}(z)\prod_{k=1}^NV_{\alpha_k}(z_k)}_\gamma=0.
		\end{equation}
		In particular the existence of a singular vector allows to express the insertion of a $\Wb_{-1}$ descendant of $V_\alpha$ as a differential operator acting on the correlation functions.
		
		But in order to show that some free-field correlation functions are solutions of a differential equation, we need to consider singular vectors at higher level. Let $\alpha=-\gamma\omega_2$: in that case we have
		\begin{equation}\label{eq:full_deg}
			\Wb_{-n}\prim+\sum_{\norm{\lambda}=n}c_{\lambda} \L_{-\lambda}\prim=0\qt{for any}1\leq n\leq 4.
		\end{equation}
		Then in particular this implies that for free-field correlation functions we can express the insertion of descendants $\Wb_{-n}V_{\alpha^*}$ for $1\leq n\leq 4$ in terms of differential operators:
		\begin{equation}
			\begin{split}
				&\ps{\Wb_{-n}V_{\alpha^*}(z)\prod_{k=2}^NV_{\alpha_k}(z_k)}_\gamma+\sum_{\norm{\lambda}=n}c_{\lambda}\mc D_\lambda(\bm\alpha,\bm z) \ps{V_{\alpha^*}(z)\prod_{k=2}^NV_{\alpha_k}(z_k)}_\gamma=0,\\
				&\mc D_{\lambda}(\bm\alpha,\bm z)\coloneqq\mc D_{\lambda_l}(\bm\alpha,\bm z)\cdots\mc D_{\lambda_1}(\bm\alpha,\bm z)\qt{with}\mc D_p(\bm\alpha,\bm z)\coloneqq \sum_{k=2}^N\frac{-\partial_{z_k}}{(z_k-z)^{p-1}}+\frac{(p-1)\Delta_{\alpha_k}}{(z_k-z)^p}\cdot
			\end{split}
		\end{equation}
		We can combine the latter with the local Ward identities from Theorem~\ref{thm:ward_gen}. Taking $n=4$ yields
		\begin{equation*}
			\sum_{k=2}^N\sum_{i=1}^4\frac{1}{(z-z_k)^i}\ps{V_{\alpha^*}(z)\Wb_{i-4}V_{\alpha_k}(z_k)\prod_{l\neq k}V_{\alpha_l}(z_l)}_{\gamma}+\sum_{\norm{\lambda}=4}c_{\lambda}\mc D_\lambda(\bm\alpha,\bm z) \ps{V_{\alpha^*}(z)\prod_{k=1}^NV_{\alpha_k}(z_k)}_\gamma=0.
		\end{equation*}
		
		We also know thanks to the global Ward identities that there are some non-trivial relations between the different quantities $\ps{V_\alpha(z)\W s_{i-s}V_{\alpha_k}(z_k)\prod_{l\neq k}V_{\alpha_l}(z_l)}_{\gamma}$,  $2\leq k\leq N$ and $1\leq i\leq 3$, whose cardinality is $3(N-1)$. Namely we have $7$ such relations, which (as soon as the system of linear equations is non-degenerate), allows to reduce the number of unknowns appearing in the above equation to $3(N-1)-7$ plus the $W$-descendants $\W s_{-i}V_\alpha(z)$ for $1\leq i\leq 3$, for which we already know that they can be express in terms of Virasoro descendants. In particular if we take $N=4$ then this system of $7$ linear equations is non-degenerate and we obtain
		\begin{equation*}
			\sum_{i=1}^{2}F_i(\bm\alpha,\bm z)\ps{V_{\alpha^*}(z)\Wb_{-i}V_{\alpha_1}(z_1)\prod_{k=2}^3V_{\alpha_k}(z_k)}_{\gamma}+\tilde{\mc D}(\bm\alpha,\bm z) \ps{V_{\alpha^*}(z)\prod_{k=1}^2V_{\alpha_k}(z_k)}_\gamma=0
		\end{equation*}
		where $F_i(\bm\alpha,\cdot)$ is a rational function and $\tilde D(\bm\alpha,\bm z)$ a Fuchsian differential operator (computed in Appendix~\ref{appendix:WB2}). Under the additional assumption that $\alpha_1=\-\gamma\omega_1^\vee$ we have extra null vectors
		\begin{equation}\label{eq:semi_deg}
			\Wb_{-n}\vert\alpha_1\rangle+\sum_{\norm{\lambda}=n}c_{\lambda} \L_{-\lambda}\vert\alpha_1\rangle=0\qt{for any} n=1,2
		\end{equation}
		so that we can further simplify the latter to express it in terms of a differential equation satisfied by the corresponding four-point correlation function.
		To summarize, we obtain the following:
		\begin{theorem}
			Let $\alpha^*=-\gamma\omega_2$ and $\alpha_2=-\gamma\omega_1$. Then under the neutrality condition~\eqref{eq:neutrality_gen} and provided that it is well-defined, the correlation function $\ps{V_{\alpha^*}(z)V_{\alpha_1}(0)V_{\alpha_2}(1)V_{\alpha_3}(\infty)}_\gamma$ is a solution of a Fuchsian differential equation of order $4$:
			\begin{equation}\label{eq:PDE}
				\mc D(\bm\alpha,\bm z)\mc H(z)=0,\quad\mc H(z)\coloneqq \norm{z}^{\ps{\alpha^*,\alpha_1}}\norm{z-1}^{\ps{\alpha^*,\alpha_2}}\ps{V_{\alpha^*}(z)V_{\alpha_1}(0)V_{\alpha_2}(1)V_{\alpha_3}(\infty)}_\gamma.
			\end{equation}
			Here $\mc D(\bm\alpha,\bm z)$ is the Fuchsian differential operator
			\begin{equation}
				\begin{split}
					\mc D(\bm\alpha,\bm z)&=z^2\prod_{i=1}^4\left(z\partial_z+A_i\right)+\prod_{i=1}^4\left(z\partial_z+B_i-1\right)\\
					&-z\left(\prod_{i=1}^4\left(z\partial_z+\tilde A_i\right)+\prod_{i=1}^4\left(z\partial_z+\tilde B_i-1\right)\right).
				\end{split}
			\end{equation}
			The coefficients $\tilde A_i$ and $\tilde B_i$ are described in Appendix~\ref{appendix:WB2} while
			\begin{equation}
				\begin{split}
					A_i=\frac\gamma2\ps{\alpha^*+\alpha_1+\alpha_2-Q,h_1}+\frac\gamma2\ps{\alpha_3-Q,h_i};\\
					B_i=1+\frac\gamma2\ps{\alpha_2-Q,h_1-h_{i+1}}
				\end{split}
			\end{equation}
			The Riemann scheme that is associated to this Fuchsian differential equation is given by
			\begin{equation}
				\left\{\begin{matrix}
					x=0 & x=1 & x=\infty\\
					0 & 0 & A_1\\
					1-B_1 & 1 & A_2\\
					1-B_2 & 2+\frac{5\gamma^2}{4} & A_3\\
					1-B_3 & 3+\frac{3\gamma^2}{4} & A_4
				\end{matrix}\right\}.
			\end{equation}
			In particular the index of rigidity is equal to $-2$: this Fuchsian system is \textbf{not} rigid. 
			
			Finally a similar statement holds true if we consider more generally $\alpha^*\in\left\{-\gamma\omega_2,-\frac2\gamma\omega_1^\vee\right\}$ and either $\alpha_2\in\left\{-\gamma\omega_1,-\gamma\omega_2,-\frac2\gamma\omega_1^\vee,-\frac2\gamma\omega_2^\vee\right\}$ with $\alpha_1$ and $\alpha_3$ generic, or $\alpha_2$ generic but $\alpha_1=\kappa_1\omega_i^\vee$ and $\alpha_3=\kappa_2\omega_i^\vee$. The corresponding system is not rigid as well.
		\end{theorem}
		\begin{remark}
			Hypergeometric functions ${}_n F_{n-1}$ are naturally associated with the Lie algebra $\mathfrak{sl}_n$ in the sense that under analogous assumptions, four-point correlation functions defined from the Lie algebra $\mathfrak{sl}_n$ are expected to be solutions of the hypergeometric differential equation
			\begin{equation}
				\mc D(z){}_n F_{n-1}(z)=0,\qt{where}\mc D(z)= z\prod_{i=1}^n\left(z\partial_z+A_i\right)-\prod_{i=1}^n\left(z\partial_z+B_i-1\right),
			\end{equation}
			and with the $A_i$, $B_i$ being given by similar expressions (but now involving weights of the first fundamental representation of $\mathfrak{sl}_n$).
			We expect the solutions of the above differential equation~\eqref{eq:PDE} to play a similar role in the case where the underlying Lie algebra is of type $B$. In this respect they provide natural generalizations of the hypergeometric function in the setting of the Lie algebra $B_2$, or more generally $B_n$ by considering a straightforward generalization of the above.
		\end{remark}

		\subsubsection{From differential equations to Selberg integrals}
		Based on the knowledge of this differential equation satisfied by such four-point correlation functions, we should now be in position to derive explicit expressions for Selberg integrals. The following discussion should be understood at the heuristic, non-rigorous level, and should be thought of as a roadmap for computing Dotsenko-Fateev integrals associated to $B_2$. We plan to make this approach rigorous in a work in progress.
		
		To do so we will rely on two different set of tools, coming either from Fuchsian theory~\cite{Fuchsian} or based on probabilistic OPEs. To start with, let us introduce the three-point correlation function
		\begin{equation}
			C_\gamma(\alpha_1,\alpha_2,\alpha_3)\coloneqq\ps{V_{\alpha_1}(0)V_{\alpha_2}(1)V_{\alpha_3}(\infty)}_\gamma.
		\end{equation}
		Under the neutrality condition~\eqref{eq:neutrality_gen} it admits an integral representation in the form of Dotsenko-Fateev type integrals: computing the latter thus follows from the computation of $C_\gamma(\alpha_1,\alpha_2,\alpha_3)$.
		
		For this, on the one hand, we can describe the space of solutions for Equation~\eqref{eq:PDE} using Fuchsian theory. Namely we know that it admits a basis of solutions near the origin of the form
		\[
		z^{1-B_i}\bar z^{1-B_j}\mc H_i(z)\mc H_j(\bar z)
		\]
		where the  $\mc H_i$ are power series in $z$, generalizing the hypergeometric function as discussed above. The fact that free-field correlation functions are real-valued implies that it has no monodromy around the singular point $z=0$, which allows to write an expansion of the form
		\begin{equation}\label{eq:expansion}
			\mc H(z)=\sum_{i=1}^4 \lambda_i \norm{z}^{2(1-B_i)}\norm{\mc H_i(z)}^2
		\end{equation}
		for some coefficients $\lambda_i$ to be determined. To determine these coefficients we can rely on the fact that Equation~\eqref{eq:PDE} admits a basis of solutions around the singular point $z=\infty$. In particular we know that we can write an equality of the form (for generic values of the $B_i$)
		\begin{equation}
			z^{1-B_i}\mc H_i(z)=\sum_{j=1}^4\Lambda_{i,j} (-z)^{A_j}\tilde{\mc H}_j\left(\frac{1}{z}\right)
		\end{equation}
		for some coefficients $\Lambda_{i,j}$\footnote{These coefficients can be recovered as residues of the Mellin transform of $\mc H_i$. As such they naturally satisfy a functional equation that generalizes the one enjoyed by the Gamma function $\Gamma$.} and power series $\tilde{\mc H}_j$. The previous expansion~\eqref{eq:expansion} then becomes
		\[
		\sum_{j,l=1}^4(-z)^{A_j}(-\bar z)^{A_l}\tilde{\mc H}_j\left(\frac{1}{z}\right)\tilde{\mc H}_j\left(\frac{1}{\bar z}\right)\sum_{i=1}^4\lambda_i\Lambda_{i,j}\Lambda_{i,l}.
		\]
		Like before single-valuedness of the free field correlation functions implies that for generic $A_i$
		\begin{equation}
			\sum_{i=1}^4\lambda_i\Lambda_{i,j}\Lambda_{i,l}=0\qt{for}j\neq l.
		\end{equation}
		We expect this equation to characterize the $\lambda_i$ up to a multiplicative constant: this means that we expect to be able to write $\lambda_i=\lambda_1\mc A_i$ for $2\leq i\leq 4$, 
		where the $\mc A_i$ are completely determined by the Fuchsian equation~\eqref{eq:PDE}, and where we can evaluate $\lambda_1=C_\gamma(\alpha_1-\gamma h_1,\alpha_2,\alpha_3)$. To summarize, \textit{a priori} analysis of the Fuchsian equation~\eqref{eq:PDE} should allow to write an expansion of the form
		\begin{equation}
			\begin{split}
				\ps{V_{\alpha^*}(z)V_{\alpha_1}(0)V_{\alpha_2}(1)V_{\alpha_3}(\infty)}=\norm{z}^{-\ps{\alpha^*,\alpha_1}}\norm{z-1}^{-\ps{\alpha^*,\alpha_2}}\mc H(z),\\
				\mc H(z)= C_\gamma(\alpha_1-\gamma h_1,\alpha_2,\alpha_3)\sum_{i=1}^4 \mc A_i \norm{z}^{2(1-B_i)}\norm{\mc H_i(z)}^2.
			\end{split}
		\end{equation}
		Now we also expect to be able to get an alternative expression for such a four-point correlation function using Operator Product Expansions. Namely using the probabilistic representation of the latter and along the same lines as in~\cite{Toda_correl2} we expect to be able to write an expansion of the form
		\begin{equation}
			\begin{split}
				\mc H(z)= \sum_{i=1}^4 \mc B_\gamma^{(i)}(\alpha_1) C_\gamma(\alpha_1-\gamma h_i,\alpha_2,\alpha_3)\norm{z}^{2(1-B_i)}\norm{\mc H_i(z)}^2,
			\end{split}
		\end{equation}
		where the $B_i$ are explicit and expressed as ratios of $\Gamma$ functions.
		In the above the quantity $C_\gamma(\alpha_1-\gamma h_i,\alpha_2,\alpha_3)$ should stand for a suitable analytic continuation (in the spirit of~\cite[Theorem 1.1]{Toda_correl2}) of the probabilistically defined three-point function $C_\gamma(\alpha_1,\alpha_2,\alpha_3)$. 
		
		Combining the two alternative expansions we can conclude that the three-point correlation functions satisfy the following set of shift equations:
		\begin{equation}
			\mc B_\gamma^{(i)}(\alpha_1) C_\gamma(\alpha_1-\gamma h_i,\alpha_2,\alpha_3)= \mc A_i C_\gamma(\alpha_1-\gamma h_1,\alpha_2,\alpha_3).
		\end{equation}
		We should also be able to carry the same analysis when this time we consider $\alpha^*=-\frac2\gamma\omega_1^\vee$ instead of $-\gamma\omega_2$. This would entail the existence of a similar shift equation:
		\begin{equation}
			\tilde{\mc B}_\gamma^{(i)}(\alpha_1) C_\gamma(\alpha_1-\frac2\gamma \tilde h_i,\alpha_2,\alpha_3)= \tilde{\mc A}_i C_\gamma(\alpha_1-\frac2\gamma \tilde h_1,\alpha_2,\alpha_3).
		\end{equation}
		We expect this functional relation to characterize uniquely up to a global multiplicative constant the three-point function $C_\gamma(\alpha_1,\alpha_2,\alpha_3)$. Hence carrying rigorously and explicitly the above analysis should allow one to derive an explicit expression for the following $B_2$-Dotsenko-Fateev integrals:
		\begin{equation}
			\begin{split}
				\int_{\C^n}\prod_{i=1}^n x_i^{-\ps{\alpha,\gamma e_{x_i}}}\norm{1-x_i}^{-\ps{\beta,\gamma e_{x_i}}}\prod_{i<j}\norm{x_i-x_j}^{-\ps{\gamma e_{x_i},\gamma e_{x_j}}}\d^2x_1\cdots \d^2x_n
			\end{split}
		\end{equation}
		as soon as $\alpha+\beta+\gamma\sum_{i=1}^r n_i e_i=2Q+\gamma\omega_1+\mu,$ with $\mu\in\left\{\gamma\omega_1,\gamma\omega_2,\frac2\gamma\omega_1^\vee,\frac2\gamma \omega_2^\vee\right\}$.
		
		
		
		\appendix
		\addtocontents{toc}{\protect\setcounter{tocdepth}{0}}

        \section{Reminders on Vertex Operator Algebras}\label{appendix:VOA}
        We gather in this appendix basic notions on Vertex Operator Algebras, that we present in a probability-oriented way. Classical references on the topic are for instance~\cite{Kac_VOA, Hua_CFT, FBZ}.

        \subsection{Definition of a Vertex Operator Algebra}
        Let us consider the graded vector space $\V_{+,\bm c}$ from Section~\ref{sec:GFF}. Then any element of $\V_{+,\bm c}$ can be represented as a map
        \begin{align*}
            \vecgen\coloneqq(c,\varphi)\mapsto e^{-\ps{Q,c}}\ephi{:\prod_{k=1}^p\frac{\partial^{n_k}\ps{u_k,\X(0)}}{n_k!}:}.
        \end{align*}
        In particular, in terms of functionals of the field $\X:\D\to\a$, we only consider the value of $\X$ (as well as of its derivatives) at the origin. If we want to get more information on $\X$ by adding a dependence on an extra point $z\in\D\setminus\{0\}$, we are naturally led to considering maps of the form
        \begin{align*}
            Y\Big(\vert \bm v,\bm m\rangle;z\Big)\vert \bm u,\bm n\rangle\quad \coloneqq \quad (c,\varphi)\mapsto e^{-\ps{Q,c}}\ephi{:\prod_{k=1}^q\frac{\partial^{m_k}\ps{v_k,\X(z)}}{m_k!}::\prod_{k=1}^p\frac{\partial^{n_k}\ps{u_k,\X(0)}}{n_k!}:}.
        \end{align*}
        A natural question to ask is the meaning of the above writing. By Gaussian integration by parts the right-hand side, viewed as a function of $z$, is a formal Laurent series (that is, powers of $z$ are bounded from below but not necessarily from above) whose coefficients are given by elements of $\V_{+,\bm c}$: the right-hand side belongs to $\V_{+,\bm c}((z))$. The notation $Y\Big(\vert \bm v,\bm m\rangle;z\Big)$ can thus naturally be thought of as an element of $\text{End}(\V_{+,\bm c})[[z,z^{-1}]]$. Moreover this assignment is seen to satisfy:
        \begin{enumerate}
            \item $Y\Big(\vac;z\Big)\vert \bm u,\bm n\rangle=\vert \bm u,\bm n\rangle$ for any $\bm u,\bm n$ as above, while by Gaussian integration by parts $Y\Big(\vert \bm u,\bm n\rangle;z\Big)\vac\in\V_{+,\bm c}[[z]]$. Moreover $Y\Big(\vert \bm u,\bm n\rangle;z\Big)\vac\vert_{z=0}=\vert \bm u,\bm n\rangle$.
            \item by definition of $T$~\eqref{eq:def_T}, for any $\bm u,\bm n$ and $\bm v,\bm m$: 
            \[
                T \left(Y\Big(\vert \bm v,\bm m\rangle;z\Big)\vert \bm u,\bm n\rangle\right)-Y\Big(\vert \bm v,\bm m\rangle;z\Big) \Big(T\vert \bm u,\bm n\rangle\Big)=\partial_z Y\Big(\vert \bm v,\bm m\rangle;z\Big)\vert \bm u,\bm n\rangle.
            \]
            \item from Equation~\eqref{eq:radial_ordering} we know that for $\norm{z}>\norm{w}>0$ and any $\bm u_i,\bm n_i$, $i=1,2,3$:
            \begin{align*}
                &Y\left(\vert \bm u_1,\bm n_1\rangle;z\right)Y\left(\vert \bm u_2,\bm n_2\rangle;w\right)\vert \bm u_3,\bm n_3\rangle=\mc R\Big(Y\left(\vert \bm u_1,\bm n_1\rangle;z\right)Y\left(\vert \bm u_2,\bm n_2\rangle;w\right)\Big)\vert \bm u_3,\bm n_3\rangle\coloneqq\\
                &(c,\varphi)\mapsto e^{-\ps{Q,c}}\ephi{:\prod_{k}\frac{\partial^{n_{1,k}}\ps{u_{1,k},\X(z)}}{n_{1,k}!}::\prod_{k}\frac{\partial^{n_{2,k}}\ps{u_{2,k},\X(w)}}{n_{2,k}!}::\prod_{k}\frac{\partial^{n_{3,k}}\ps{u_{3,k},\X(0)}}{n_{3,k}!}:}.
            \end{align*}
            By Gaussian integration by parts the right-hand side can be analytically continued to a meromorphic function with singularities at $z=0$, $w=0$ and $z=w$ (the order of the pole $z=w$ being at most $\norm{\bm n_1}+\norm{\bm n_2}$). 
        \end{enumerate}
        In the vertex algebra literature, item $(1)$ is referred to as the \textit{vacuum axiom}, item $(2)$ as the \textit{translation axiom} while item $(3)$ is equivalent~\cite[Proposition 1.2.9]{FBZ} to the \textit{locality axiom}. 

        The notion of vertex algebra thus naturally arises within this setting: a vertex algebra is the data of     $(\V,\vac,T,Y)$ where $\V$ is a vector space, $\vac$ one of its elements, $T$ a linear operator $\V\to\V$ and vertex operators $Y(\cdot,z):\V\to \text{End}(\V)[[z,z^{-1}]]$. They are subject to items $(1)$, $(2)$ and $(3)$.

        \subsection{The Heisenberg vertex algebra} We fix $\g$ a finite-dimensional complex, simple Lie algebra and $\h$ its Cartan subalgebra, equipped with $\ps{\cdot,\cdot}$ a multiple of the Killing form of $\g$. The Heisenberg algebra is $\hat\h\coloneqq\h\otimes \C[t,t^{-1}]\oplus \C \bm k$ (with $\bm k$ the \textit{central element}), together with the Lie bracket $[\alpha\otimes t^n,\beta\otimes t^m]=n\delta_{n+m}\ps{\alpha,\beta}\bm k$ and $[\alpha\otimes t^n,\bm k]=[\bm k,\bm k]=0$ for all $\alpha,\beta\in\h$ and $n,m\in\Z$.

        Let us define Lie subalgebras of $\hat\h$ by $\hat\h_-\coloneqq \h\otimes t^{-1}\C[t^{-1}]$, $\hat\h_+\coloneqq \h\otimes t\C[t]$ and $\hat\h_0\coloneqq \h\otimes t^0 \oplus \C\bm k$. We further denote by $U(\hat\h)$ the universal enveloping algebra of $\hat\h$, and let $\hat\h_+\oplus\hat\h_0$ act trivially on $\C$ ($\alpha\otimes t^n$ acts by zero and $\bm k$ by identity). The Fock representation of $\hat\h$ is then $U(\hat\h)\otimes_{U(\hat\h_+\otimes\hat\h_0)}\C$. It is naturally identified with $\pi_0\coloneqq U(\hat\h_-)\otimes_\C\C$, which thus comes equipped with a structure of $\hat\h$-module
        . Moreover $\pi_0$ is spanned by the $\left(\alpha_1\oplus t^{-n_1}\right)\cdots \left(\alpha_l\oplus t^{-n_l}\right)\vac$ with $n_1,\cdots,n_l<0$ and $\alpha_1,\cdots,\alpha_l\in\h$, where $\vac\coloneqq 1\in\pi_0$ is seen as a monomial of degree zero.

        The Heisenberg vertex algebra can then be defined as the data of the vector space $\pi_0$, the vacuum $\vac\in\pi_0$, and the translation operator $T$ and vertex operators $Y$ given by
        \begin{align*}
            &T\left(\alpha_1\oplus t^{-n_1}\right)\cdots \left(\alpha_l\oplus t^{-n_l}\right)\vac\coloneqq\sum_{i=1}^l n_i\left(\alpha_1\oplus t^{-n_1}\right)\cdots \left(\alpha_i\oplus t^{-n_i-1}\right)\cdots\left(\alpha_l\oplus t^{-n_l}\right)\vac\\
            &Y\Big(\left(\alpha_1\oplus t^{-n_1}\right)\cdots \left(\alpha_l\oplus t^{-n_l}\right)\vac;z\Big)\coloneqq \quad:\prod_{i=1}^l\frac{\partial^{n_i-1}\ps{\alpha_i,j(z)}}{(n_i-1)!}:\text{ with }\ps{\alpha,j(z)}\coloneqq \sum_{n\in\Z}(\alpha\otimes t^n)z^{-n-1}.
        \end{align*}
        
		\section{Alternative definitions for $W$-algebras}\label{appendix:Walgebra}
		We have used above the definition of $W$-algebras as kernels of screening operators. There are however other approaches for defining such a notion: in this appendix we present some of these perspectives. For this purpose we will follow the historical approach and first explain the construction for $\g=\sl_3$, in which case the corresponding $W$-algebra was first unveiled by Zamolodchikov' seminal work~\cite{Za85}. We then construct $W$-algebras associated to a more general family of simple Lie algebras using the approach designed by Fateev-Lukyanov~\cite{FaLu} based on a \textit{Miura transform}. 
		
		\subsection{The $W_3$ algebra}
		To start with let us assume that $\g=\sl_3$. 
		In the $W_3$ algebra, in addition to the Virasoro algebra there is a second family of operators, corresponding to the modes of a \textit{higher-spin current} $\Wb(z)$. This additional operator satisfies the following OPEs for $\norm{z}>\norm{w}$:
		\begin{equation}
			\begin{split}
				\L(z)\Wb(w)&=\frac{3\Wb(w)}{(z-w)^2}+\frac{\partial_w\Wb(w)}{z-w}+reg.\\
				\Wb(z)\Wb(w)&=\frac{\bm c/3}{(z-w)^6}+\frac{2\L(w)}{(z-w)^4}+\frac{\partial_w\L(w)}{(z-w)^3}\\
				&+\frac{\frac3{10}\partial_w^2\L(w)+2\beta\Lambda(w)}{(z-w)^2}+\frac{\frac1{15}\partial_w^3\L(w)+\beta\partial_w\Lambda(w)}{(z-w)} + reg.,
			\end{split}
		\end{equation}
		where $\beta=\frac{16}{22+5\bm c}$ and $\Lambda(w)=:\L(w)^2:-\frac{3}{10}\partial_w^2\L(w)$. These translate as commutation relation between the modes $(\L_n,\Wb_m)_{n,m\in\Z}$ of these two currents: for instance $[\L_n,\Wb_m]=(2n-m)\Wb_{n+m}$.
		Explicit computations show that the element of $\endv[[z,z^{-1}]]$ defined by
		\begin{equation}\label{eq:W3_curr}
			\begin{split}
				\Wb(z)\coloneqq &\i\sqrt{3\beta}\times\Big[ q^2h_2(\partial^3\Phi(z))-8h_1(\partial\Phi(z))h_2(\partial\Phi(z))h_3(\partial\Phi(z))\\
				&-2q\left((h_2-h_1)(\partial^2\Phi(z))h_1(\partial\Phi(z))+(h_3-h_2)(\partial^2\Phi(z))h_3(\partial\Phi(z))\right)\Big]
			\end{split}
		\end{equation}
		does satisfy the desired OPEs. Here $q=\gamma+\frac2\gamma$ and $h_i=\omega_1-\sum_{j=1}^{i-1}e_j$.
		
		The $W_3$-algebra is then defined by the data of the graded vector space \begin{equation}
			\MW{\sl_3}\coloneqq\bigoplus_{n\geq0}\MW{\sl_3}^{(n)},\quad\MW{\sl_3}^{(n)}\coloneqq\bigoplus_{\substack{\lambda_1\in\mc T_2,\lambda_2\in\mc T_3\\ \norm{\lambda_1}+\norm{\lambda_2}=n}}\text{span}\left\{\Wb_{-\lambda_2}\L_{-\lambda_1}\vac\right\}
		\end{equation}
		together with the assignment
		\begin{equation}
			\begin{split}
				&Y\left(\Wb_{-\lambda_2}\L_{-\lambda_2}\vac,z\right)\coloneqq\quad:\frac{\partial^{\lambda_2-3}\Wb(z)}{(\lambda_2-3)!}\cdots\frac{\partial^{\lambda_1-2}\L(z)}{(\lambda_1-2)!}:\qt{where}\\ &\frac{\partial^{\lambda_2-3}\Wb(z)}{(\lambda_2-3)!}\coloneqq \frac{\partial^{\lambda_2^{l(\lambda_2)}-3}\Wb(z)}{(\lambda_2^{l(\lambda_2)}-3)!}\cdots\frac{\partial^{\lambda_2^1-3}\Wb(z)}{(\lambda_2^1-3)!}\cdot
			\end{split}
		\end{equation}	
		From the property~\eqref{eq:adj_A} the representation of the $W_3$ algebra thus constructed can be shown by explicit computations to be unitary in the sense that for $\gamma\in(0,\sqrt 2)$:
		\begin{equation}\label{eq:dual_W}
			\L_n^*=\L_{-n}\qt{and}\Wb_m^*=\Wb_{-m}\qt{for any}n\in\Z.
		\end{equation}

		\subsection{$W$-algebras and the Miura transform}
		It may seem a bit mysterious at first sight how to come up with such an expression for the higher-spin current $\Wb(z)$. One can do the computations manually to find an expression for which the algebra closes, as done for instance in~\cite{Za85}. A systematic method was soon after proposed by Fateev-Lukyanov~\cite{FaLu} to compute more generally the currents generating the $W$-algebra associated to $\g=A_n$ (corresponding to $\sl_{n+1}$) and $D_n$ (for $\mathfrak{o}_{2n}$). 
		
		Let us illustrate how this Miura transform is defined in the case where $\g=\sl_n$, in which case $r=n-1$. Recall that the weights $h_k$, $1\leq k\leq n$, are defined by setting $h_k=\omega_1-\sum_{i=1}^{k-1}e_i$; then we can write the following equality, understood at the level of differential operators\footnote{The multiplicative factor $2^{\frac n2}$ stems from our normalization of the commutation relations of the Heisenberg algebra and of the coupling constant.}:
		\begin{equation}
			2^{\frac n2}\left(\frac q2\partial+\ps{h_{n},\partial\varphi}\right)\cdots\left(\frac q2\partial+\ps{h_{1},\partial\varphi}\right)=\sum_{s=0}^n\W s[\varphi]\left(\frac{q}{2}\partial\right)^{n-s}.
		\end{equation}
		For instance $\W0[\varphi]=\mathds 1$ while $\W1[\varphi]=\sum_{k=1}^n\ps{h_k,\partial\varphi}=0$. 
		The currents $(\W s(z))_{2\leq s\leq n}$ are
		\begin{equation}
			\W s(z)\coloneqq \W s[\Phi](z).
		\end{equation} 
		Moreover the current of spin $s=2$ is
		\begin{equation*}
			\W2(z)=q\sum_{k=1}^{n-1}(n-k)\ps{h_k,\partial^2\Phi(z)}+2\sum_{k<l}\ps{h_k,\partial\Phi(z)}\ps{h_l,\partial\Phi(z)}.
		\end{equation*}
		Since for $\g=\sl_n$, $\rho=\sum_{k=1}^{n-1}(n-k)h_k$ and $\sum_{k\neq l}\ps{h_k,u}\ps{h_l,v}=-\ps{u,v}$ we can write that
		\begin{equation}
			\W 2(z)=\ps{Q,\partial^2\Phi(z)}-\ps{\partial\Phi(z),\partial\Phi(z)}.
		\end{equation}
		This is nothing but the usual definition of the \textit{stress-energy tensor} $\L(z)$.
		
		In the case of a Lie algebra of type $D$,\footnote{A Miura transform associated to any simple Lie algebra is constructed in~\cite{DS}, while in~\cite{FaLu} the authors also define currents from the Miura transform associated to $B_n$ (for $\mathfrak{o}_{2n+1}$) and $C_n$ (for $\mathfrak{sp}_{2n}$). However this does not lead to the definition of the $W$-algebra associated to $\g$ in these cases.} the form of the Miura transform is more involved~\cite{FaLuDn} due to the form of the exponents of the Lie algebra $D_n$. Indeed these are given by $1,3,\cdots,2n-3$ and $n-1$, so that there exist currents with spins $\W {2k}$ for $1\leq k\leq n-1$ that may be constructed using a transform similar as above (though involving pseudo-differential operators): 
		\begin{equation}
			\begin{split}
				&2^{n}\left(\frac q2\partial+\ps{h_{1},\partial\varphi}\right)\cdots\left(\frac q2\partial+\ps{h_{n},\partial\varphi}\right)\left(\frac q2\partial\right)^{-1}\left(\frac q2\partial-\ps{h_{n},\partial\varphi}\right)\cdots\left(\frac q2\partial-\ps{h_{1},\partial\varphi}\right)\\
				&=\sum_{s=0}^{2n-1}\W s[\varphi]\left(\frac{q}{2}\partial\right)^{2n-1-s}+f\left(\frac q2\partial\right)^{-1}f
			\end{split}
		\end{equation}
		but there is in addition a current $\W{n}$ of spin $n$ which is given by
		\begin{equation}
			\W n[\varphi]\coloneqq f= -(-2)^{m}\left(\frac q2\partial+\ps{h_{1},\partial\varphi}\right)\cdots\left(\frac q2\partial+\ps{h_{m},\partial\varphi}\right)1.
		\end{equation}
		Here the $h_i$ satisfy $\ps{h_i,h_j}=\delta_{i,j}$.
		In the case where $m$ is even, then there are actually two currents with spin $n$: this is due to the fact that in that case the exponent $n-1$ has multiplicity $2$. The current of spin two is always given by the stress-energy tensor in that $\W2(z)=\L(z)$ (for $n\geq 3$).
		
		When $\g$ is of the type $A$ or $D$, a representation of the $W$-algebra associated to $\g$ is then 
		\begin{equation}
			\MW{\g}\coloneqq\bigoplus_{n\geq0}\MW{\g}^{(n)},\quad\MW{\g}^{(n)}\coloneqq\bigoplus_{\substack{\lambda_1\in\mc T_{s_1},\cdots\lambda_r\in\mc T_{s_r}\\ \norm{\lambda_1}+\cdots+\norm{\lambda_r}=n}}\text{span}\left\{\Wb^{(s_r)}_{-\lambda_r}\cdots\Wb^{(s_1)}_{-\lambda_1}\vac\right\}
		\end{equation}
		where the $s_i$, $i=1,\cdots,l$ are the spins of $\g$, together with the Vertex Operators
		\begin{equation}
			Y\left(\Wb^{(s_r)}_{-\lambda_r}\cdots\Wb^{(s_1)}_{-\lambda_1}\vac,z\right)=\quad :\prod_{i=1}^r\prod_{j=1}^{l(\lambda_i)}\frac{\partial^{\lambda_i^j-s_i}\Wb^{(s_i)}(z)}{(\lambda_i^j-s_i)!}:.
		\end{equation}	
		The advantage of the Miura transform is that it gives an explicit construction for the currents $\W s$ generating the $W$-algebra. However and to the best of our knowledge there is no such expression when the underlying Lie algebra is not of the type $A$ or $D$ apart from specific cases~\cite{KW, ABCDEFG}.
		

		
		\section{Explicit expressions for the $\mc WB_2$ algebra}\label{appendix:WB2}
		The explicit expressions provided in this section are derived using Mathematica. The code should be provided, hopefully, anytime soon on the webpage of the author. In the meantime it can be sent upon request to anyone interested.
		\subsection{Expression of the currents and of the singular vectors}
		The definition of the $W$-algebra using screening operators is very convenient when it comes to describing explicit expressions for the currents. Despite not being as explicit as the Miura transform, it still provides an algorithmic way of defining these currents.
		
		\subsubsection{Expression of the $\Wb$ current}
		To be more specific, in the case of $\g=B_2$ let us consider the action of the screening operators $Q_i$ on the Fock space $\V_+$. First of all the action of a screening operator on any element of the Fock space can be implemented  algorithmically by using the explicit expression of the latter in terms of the generators of the Heisenberg algebra together with their action on an element of the Fock space. Then for generic values of $\gamma$ we know from Theorem~\ref{thm:def_W}
		\[
		\bigcap_{i=1}^2\ker_{\V_+^{(4)}}(Q_i)=\text{span}\left\{\L_{-2}^2\vac, \L_ {-4}\vac,\Wb_{-4}\vac\right\}
		\]
        since the exponents of $B_2$ are $1$ and $3$ (thus spins $2$ and $4$).
		The left-hand side can be described explicitly using Mathematica, it remains to identify in the right-hand side to identify the corresponding term $\Wb_{-4}$. 
	    By doing so we arrive to:
        \small
		\begin{equation}
			\begin{split}
				&e^{\ps{Q,c}}\Wb_{-4}\vac=12 \gamma ^3 (\varphi_{4,1}+2 \varphi_{4,2})+\gamma ^2 \left(\varphi_{2,1}^2-8 \varphi_{2,2}^2-12 (\varphi_{1,1}+2 \varphi_{1,2}) \varphi_{3,2}\right)\\
				&-2 \gamma \left(3 \varphi_{2,1} \varphi_{1,1}^2+2 \varphi_{1,2} (3 \varphi_{2,1}-4 \varphi_{2,2}) \varphi_{1,1}+2 \varphi_{2,1} \varphi_{1,2}^2-61 \varphi_{4,1}-8 \varphi_{1,2}^2 \varphi_{2,2}-120 \varphi_{4,2}\right)\\
				&+\left(\varphi_{1,1}^4+4 \varphi_{1,2} \varphi_{1,1}^3-8 \left(\varphi_{1,2}^3+9 \varphi_{3,2}\right) \varphi_{1,1}-4 \left(\varphi_{1,2}^4+36 \varphi_{3,2} \varphi_{1,2}+3 \varphi_{2,1}^2+18 \varphi_{2,2}^2+12 \varphi_{2,1} \varphi_{2,2}\right)\right)\\
				&-\frac8\gamma \left((2 \varphi_{2,1}-\varphi_{2,2}) \varphi_{1,1}^2+2 \varphi_{1,2} (2 \varphi_{2,1}-3 \varphi_{2,2}) \varphi_{1,1}-48 \varphi_{4,1}-6 \varphi_{1,2}^2 \varphi_{2,2}-95 \varphi_{4,2}\right)\\
				&-\frac{16}{\gamma^2}  \left(2 \varphi_{2,1}^2+8 \varphi_{2,2} \varphi_{2,1}+9 \varphi_{2,2}^2+6 (\varphi_{1,1}+2 \varphi_{1,2}) \varphi_{3,2}\right) +\frac{384}{\gamma^3} (\varphi_{4,1}+2 \varphi_{4,2}).
			\end{split}
		\end{equation}
		\normalsize
		\subsubsection{Derivation of the singular vectors}
		Based on the knowledge of this current we can now look for singular vectors. And to start with we need to understand the action of $\Wb$ when acting on a generic free-field module, which is done by using Proposition~\ref{prop:VO_prim_bis}. This provides us with an efficient way of computing the action of the operators $\L_{-n}$ and $\Wb_{-m}$ on any free-field module.
		
		We can therefore find explicit expressions for the descendants $\L_{-\lambda}\prim$ and $\Wb_{-\lambda}\prim$. This allows us to search for which values of $\alpha$ we can find linear combinations of such descendants for which we obtain zero. This is how we end up with the expressions provided in Theorem~\ref{thm:Verma}.

		\subsection{The differential equation with one singular vector} 
		As explained in Subsection~\ref{subsec:Mukhin_Varchenko} we can translate representation theoretical results into concrete differential equations satisfied by a family of four-point correlation functions. For this the first step is to express the four-point correlation with one descendant $\ps{\Wb_{-4}\V_{\alpha^*}(z)V_{\alpha_1}(z_1)V_{\alpha_2}(z_2)V_{\alpha_3}(z_3)}$ using the local Ward identities from Theorem~\ref{thm:ward_global}. Next we can invert the global Ward identities~\ref{thm:ward_global} to reduce the number of descendants, which is done by inverting the matrix equation
		\begin{equation}
			\text{Desc}\cdot\text{MatWard}=\text{ResWard}\quad\text{ with Desc}=\left(\Wb_{-3}^{(z)},\Wb_{-2}^{(z)},\Wb_{-1}^{(z)},\cdots,\Wb_{-2}^{(z_3)},\Wb_{-2}^{(z_3)},\Wb_{-1}^{(z_3)}\right),
		\end{equation}
        \small
		\begin{equation}
			\begin{split}
				&\text{MatWard}=\begin{pmatrix}
					1 & z & z^2 & \cdots & z^6 & 1 & 0 & 0 & 0 & 0\\
					0 & 1 & 2z & \cdots & 6z^5 & 0 & 1 & 0 & 0 & 0\\ 
					0 & 0 & 1 & \cdots & \frac{5\times6}2 z^4 & 0 & 0 & 1 & 0 & 0 \\
					1 & z_1 & z_1^2 & \cdots & z_1^6 & 0 & 0 & 0 & 0 & 0\\
					\vdots & \vdots & \vdots & \cdots & \vdots & \vdots & \vdots& \vdots &\vdots &\vdots \\
					1 & z_2 & z_2^2 & \cdots & z_2^6 & 0 & 0 & 0 & 0 & 0\\
					\vdots & \vdots & \vdots & \cdots & \vdots& \vdots & \vdots & \vdots  &\vdots  &\vdots   \\
					1 & z_3 & z_3^2 & \cdots & z_3^6  & 0 & 0& 0 &0 &0 \\
					0 & 1 & 2z_3 & \cdots & 6z_3^5& 0 & 0& 0 & 1 &0 \\
					0 & 0 & 1 & \cdots & \frac{5\times6}2 z_3^4& 0 & 0& 0 &0 &1 \\
				\end{pmatrix},\\
				&\text{ResWard}=\left(0,0,0,-\sum_{k=1}^4w(\alpha_k),\cdots,-\sum_{k=1}^4 \frac{4\times 5\times 6}{3!} z_k^{3},\Wb_{-3}^{(z)},\Wb_{-2}^{(z)},\Wb_{-1}^{(z)},\Wb_{-2}^{(z_3)},\Wb_{-1}^{(z_3)}\right).
			\end{split}
		\end{equation}
        \normalsize
		Having inverted this system of equations we can then specialize to $z_1=0$, $z_3=1$ and take the limit as $z_2\to \infty$. 
		By doing so we obtain that the following correlation function
		\[
		f(z):=\ps{V_{\alpha^*}(z)V_{\alpha_1}(0)V_{\alpha_2}(1)V_{\alpha_3}(\infty)}
		\]
		where $\alpha^*=-\frac\gamma2\omega_1^\vee$ is such that
		\begin{equation*}
			\begin{split}
				&\frac{ P_4(z) }{2 (z-1)^4 z^4}f(z)+\frac{2 P_3(z)}{\gamma ^2 (z-1)^3 z^3}f'(z)+\frac{2P_2(z)}{\gamma ^2 (z-1)^2 z^2}f^{(2)}(z)+\frac{64 (3 z-2)}{\gamma ^2 (z-1) z}f^{(3)}(z)-\frac{64}{\gamma^4}f^{(4)}(z)\\
				&=\frac{z \Wb_{-2}^{(1)}+3 z\Wb_{-1}^{(1)}-\Wb_{-2}^{(1)}-2\Wb_{-1}^{(1)}}{(z-1)^3 z^3}f(z).
			\end{split}
		\end{equation*}
		Here the $P_k$ are polynomials of order $k$, that admit the explicit expressions:
		\small\begin{align*}
			&P_2(z)=8 z^2 \left(-2 \Delta_{\alpha^*}+2 \Delta_{\alpha^3}+11 \gamma ^2+8\right)+8 z (2 \Delta_{\alpha^*}+2 \Delta_{\alpha^2}-2 \Delta_{\alpha^1}-2 \Delta_{\alpha^3}-11-\frac{121\gamma ^2}8  )+16 \Delta_{\alpha^1}+38 \gamma ^2+40\
		\end{align*}
		\normalsize
		\begin{align*}
			&P_3(z)= 8z^3 \left(3 \gamma ^2+4\right)\left(-\Delta_{\alpha^*}+\Delta_{\alpha^3}+\gamma ^2+1\right)\\
			&-2 z^2 \left(\gamma ^2 (-21 \Delta_{\alpha^*}-11 \Delta_{\alpha^2}+11 \Delta_{\alpha^1}+21 \Delta_{\alpha^3}+69)-24 (\Delta_{\alpha^*}+\Delta_{\alpha^2}-\Delta_{\alpha^1}-\Delta_{\alpha^3}-2)+27 \gamma ^4\right)\\
			&+z \left(-2 \gamma ^2 (9 \Delta_{\alpha^*}+9 \Delta_{\alpha^2}-19 \Delta_{\alpha^1}-9 \Delta_{\alpha^3}-49)-16 (\Delta_{\alpha^*}+\Delta_{\alpha^2}-5 \Delta_{\alpha^1}-\Delta_{\alpha^3}-5)+35 \gamma ^4\right)\\
			&-2 \left(\gamma ^2+2\right)\left(8 \Delta_{\alpha^1}+3 \gamma ^2+4\right)        
		\end{align*}
		\footnotesize
		\begin{align*}
			P_4(z)=z^4 &\left(-15 \gamma ^2 \Delta_{\alpha^3}+\Delta_{\alpha^*} \left(8 \Delta_{\alpha^3}+15 \gamma ^2+32\right)-4 \Delta_{\alpha^*}^2-4 \Delta_{\alpha^3}^2+2 w(\alpha^*)-32 \Delta_{\alpha^3}-2 w(\alpha_3)\right)\\
			+z^3 &\left(45 \gamma ^2 \Delta_{\alpha^3}-\Delta_{\alpha^*} \left(-8 \Delta_{\alpha^2}+8 \Delta_{\alpha^1}+16 \Delta_{\alpha^3}+45 \gamma ^2+100\right)+8 \Delta_{\alpha^*}^2+8 \Delta_{\alpha^3}^2-6 w(\alpha^*)+8 \Delta_{\alpha^2}-6 w(\alpha_2)\right.\\
			&\left.-8 \Delta_{\alpha^1}+2 w(\alpha_1)-8 \Delta_{\alpha^2} \Delta_{\alpha^3}+8 \Delta_{\alpha^1} \Delta_{\alpha^3}+100 \Delta_{\alpha^3}+6 w(\alpha_3)\right)\\
			+z^2 &\left(20 \gamma ^2 \Delta_{\alpha^2}-4 \gamma ^2 \Delta_{\alpha^1}-41 \gamma ^2 \Delta_{\alpha^3}+\Delta_{\alpha^*} \left(-8 \Delta_{\alpha^2}+16 \Delta_{\alpha^1}+8 \Delta_{\alpha^3}+41 \gamma ^2+88\right)-4 \Delta_{\alpha^*}^2-4 \Delta_{\alpha^2}^2\right.\\
			&\left.-4 \Delta_{\alpha^1}^2-4 \Delta_{\alpha^3}^2+6 w(\alpha^*)+76 \Delta_{\alpha^2}+6 w(\alpha_2)+8 \Delta_{\alpha^2} \Delta_{\alpha^1}+4 \Delta_{\alpha^1}-6 w(\alpha_1)+8 \Delta_{\alpha^2} \Delta_{\alpha^3}-16 \Delta_{\alpha^1} \Delta_{\alpha^3}-88 \Delta_{\alpha^3}-6 w(\alpha_3)\right)\\
			+z &\left(-11 \gamma ^2 \Delta_{\alpha^2}+4 \gamma ^2 \Delta_{\alpha^1}+11 \gamma ^2 \Delta_{\alpha^3}-\Delta_{\alpha^*} \left(8 \Delta_{\alpha^1}+11 \gamma ^2+20\right)\right.\\
			&\left.+8 \Delta_{\alpha^1}^2-2 w(\alpha^*)-20 \Delta_{\alpha^2}-2 w(\alpha_2)-8 \Delta_{\alpha^2} \Delta_{\alpha^1}+6 w(\alpha_1)+8 \Delta_{\alpha^1} \Delta_{\alpha^3}+20 \Delta_{\alpha^3}+2 w(\alpha_3)\right)\\
			&-2 \left(2 \Delta_{\alpha^1}^2-2 \Delta_{\alpha^1}+w(\alpha_1)\right).
		\end{align*}
		\normalsize
		In the above $\Delta_\alpha$ is the conformal weight while $w(\alpha)$ is the quantum weight associated to the current $\Wb$, namely $\Wb\prim=w(\alpha)\prim$.
		
		By choosing $\alpha_2$ to be semi-degenerate we can express the $\Wb_{-i}^{(1)}$ descendants that appear on the right-hand side above in terms of differential operators acting on $f(z)$. By doing so we arrive at the differential equation $\mc D(\bm\alpha,\bm z)\mc H(z)=0$ satisfied by $\mc H(z)\coloneqq \norm{z}^{\ps{\alpha^*,\alpha_1}}\norm{z-1}^{\ps{\alpha^*,\alpha_2}}f(z)$, with the $A_k$, $B_k$ given in Theorem~\ref{thm:B2} while the $\tilde A_i$, $\tilde B_i$ are such that
		\begin{align*}
			&\prod_{k=1}^4\left(x+\tilde A_i\right)+\prod_{k=1}^4\left(x+\tilde B_i-1\right)=2x^4+x^3 \left(-2 \gamma ^2+4 B_1+4 B_2\right)+ax^2+bx+c,
		\end{align*}
		\footnotesize
		\begin{align*}
			a\coloneqq &A_1 \left(-A_1+B_1+B_2-2-\gamma ^2\right)+A_2 \left(-A_2+B_1+B_2-2-\gamma ^2\right)+(B_1+B_2-2-\gamma^2)^2\\
			&-(B_1+B_2-2-\gamma^2)^2+\frac{\gamma ^4}{16}-\frac{13 \gamma ^2}{4}-3+(B_1+B_2)(3-2\gamma^2)+2(B_1^2+B_2^2)+5B_1B_2\\
			b\coloneqq& -(A_1^2+A_2^2)\left(B_1+B_2+\frac{\gamma ^2}2+2\right)+(A_1+A_2) \left((B_1+ B_2-\gamma^2)(B_1+B_2+\frac{\gamma^2}2)-(4+3 \gamma ^2)\right)\\
			&+B_1^2 B_2+B_1 B_2^2+3(B_1+B_2)^2-\frac{B_1 B_2(4+3\gamma^2)}{2}+(B_1+B_2)\frac{9  \gamma ^4-28 \gamma ^2-16}{16}-\frac{9 \gamma ^6+44 \gamma ^4+112 \gamma ^2+64}{32}
		\end{align*}
		\begin{align*}
			c\coloneqq& -(A_1^4+A_2^4)\frac{\left(3 \gamma ^2+4\right) }{4 \gamma ^2}-A_1^2A_2^2\frac{\left(\gamma ^2+4\right) }{2 \gamma ^2}+\frac{(B_1+B_2-2-\gamma^2) \left(3 \gamma ^2+4\right)}{2 \gamma ^2} (A_1^3+A_2^3)\\
			&+\left(\frac{(B_1+B_2) \left(\gamma ^2+4\right)}{2 \gamma ^2}-\frac{\left(3 \gamma ^2+4\right)(2+\gamma^2)}{2\gamma^2}\right) (A_1^2A_2+A_2^2A_1)+\frac{\left(\gamma ^2+4\right)(B_1+B_2-2-\gamma^2)^2}{2 \gamma ^2} A_1A_2\\
			&+\left(-\frac{(B_1^2+B_2^2+4B_1B_2) \left(3 \gamma ^2+4\right)}{4 \gamma ^2}+\frac{(B_1+B_2)(9 \gamma ^4+26 \gamma ^2+24)}{4 \gamma ^2}-\frac{132 \gamma ^6+816 \gamma ^4+1477 \gamma ^2+768}{128 \gamma ^2}\right)(A_1^2+A_2^2)\\
			&+\left(\frac{(B_1^2B_2+B_2^2B_1) \left(3\gamma ^2+4\right)}{2\gamma ^2}-\frac{(B_1^2+B_2^2+4B_1B_2)\left(3 \gamma ^4+6 \gamma ^2+8\right)}{4 \gamma ^2}\right)(A_1+A_2)\\
			&+\left(\frac{(B_1+B_2)(33 \gamma ^6+172 \gamma ^4+304 \gamma ^2+192)}{32 \gamma ^2}-\frac{9 \gamma ^8+94 \gamma ^6+296 \gamma ^4+352 \gamma ^2+128}{32 \gamma ^2}\right)(A_1+A_2)\\
			&-\frac{\left(3 \gamma ^2+4\right)}{4 \gamma ^2}B_1^2B_2^2 +\frac{\left(3 \gamma ^2+4\right)(\gamma^2+1)}{4 \gamma ^2}(B_1^2B_2+B_2^2B_1)-\frac{\left(9 \gamma ^6+60 \gamma ^4+48 \gamma ^2+64\right)}{64 \gamma ^2}(B_1^2+B_2^2) \\
			&-\frac{\left(3 \gamma ^6+22 \gamma ^4+32 \gamma ^2+16\right)B_1B_2}{4 \gamma ^2}+\frac{9 \gamma ^8+114 \gamma ^6+248 \gamma ^4+288 \gamma ^2+128}{64 \gamma ^2}(B_1+B_2) \\
			&-\frac{27 \gamma ^{10}+612 \gamma ^8+2912 \gamma ^6+5248 \gamma ^4+3840 \gamma ^2+1024}{1024 \gamma ^2}\cdot
		\end{align*}
        \normalsize
		We notice that for $\gamma^2=-\frac43$ we can take $\tilde A_k=A_k$ and $\tilde B_k=B_k$.

		\section{Some proofs from Section~\ref{sec:VOA}}\label{appendix:proofs}
		In this final Appendix we gather the proofs omitted in the body of the document.
		
		\subsection{Proof of Lemma~\ref{lemma:to_prove1}}
		
		\begin{proof}
			The proof is based on a recursive application of Gaussian integration by parts~\eqref{eq:IPP_D}: by definition of $U_0$ and of $\X=\X_\D+P\varphi+c$, we see that we need to understand expressions
			\begin{align*}
				U_0\left(\prod_{k=1}^p\ps{u_k,\partial^{n_k}\X_\eps(z)}F\right)=e^{-\ps{Q,c}}\ephi{\prod_{k=1}^p\ps{u_k,\partial^{n_k}\left(\X_\D+P\varphi\right)_\eps(z)}F\left(\X_\D+P\varphi+c\right)}.
			\end{align*}
			Here $p$ is positive, the $u_k$ are vectors in $\a$ and the $n_k$ positive integers, while $z$ belongs to $\D$. For the sake of simplicity we will assume that $F=e^{\ps{\X+c,f}}$ for some $f\in\mathcal E_\delta$, the generic case where $f\in\mathcal{F}_\delta$ being dealt with by differentiating with respect to $f$. Associated to $f$ we introduce the notation $f_n\coloneqq \int_\D x^nf(x)\d^2x$ for $n\geq0$ and $f_n\coloneqq \int_\D \bar x^{-n} f(x)\d^2x$ for $n<0$. Note that only finitely many of such modes are non-zero. We further denote $\bm a_n\coloneqq \frac \i2\left(f_n+2n\varphi_{-n}\mathds1_{n\leq0}\right)$. 
			
			To start with Gaussian computations give
			\begin{align*}
				\partial_{m}\ephi{F\left(\X_\D+P\varphi+c\right)}&=\partial_{m} \left(e^{\ps{c+\sum_{n>0} \left(\varphi_n x^n+\varphi_{-n}\bar x^n\right),f}_\D}e^{\frac12\ps{f,G_\D f}_\D}\right)
				=&f_m\ephi{F\left(\X_\D+P\varphi+c\right)}
			\end{align*}
			so that, by recalling Equation~\eqref{eq:def_An}, $\bm a_n U_0 F=\A_n U_0F$ for any $n\in\Z$. Moreover the commutation relations of the Heisenberg algebra~\eqref{eq:comm_he} show that \[
			\ps{u,\bm a_n}\ps{v,\bm a_m}U_0F=\left(\ps{u,\A_n}\ps{v,\A_m}-\ps{u,v}\frac{n}{2}\delta_{n,-m}\mathds 1_{m\geq 0}\right) U_0 F=:\ps{u,\A_n}\ps{v,\A_m}:U_0F.
			\]
			More generally for such a $F$ we see that we have the following:
			\begin{equation}\label{eq:aA}
				\prod_{k=1}^p\ps{u_k,\bm a_k}U_0F=\prod_{k=1}^p:\ps{u_k,\A_k}:U_0F.
			\end{equation}
			
			Now let us consider the case where $p=1$ in Equation~\eqref{eq:voa_prob} and look at expressions of the form
			\begin{align*}
				U_0\left(\ps{u,\partial^mX_\eps(z)}F\right)=e^{-\ps{Q,c}}\ephi{\ps{u,\partial^m\left(\X_\D+P\varphi\right)_\eps}F\left(\X_\D+P\varphi+c\right)}.
			\end{align*}
			In virtue of Equation~\eqref{eq:IPP}, in $e^{\beta\norm{c}}\Lro$ with $\beta$ large enough:
			\begin{align*} 	
				&\lim\limits_{\eps\to 0}\ephi{\ps{u,\partial(\X_\D)_\eps(z)}F\left(\X_\D+P\varphi+c\right)}=\int_{\D}\partial_zG_\D(z,x)\ps{u,f(x)}\d^2x\ephi{F\left(\X\right)}.
			\end{align*}
			Since $f$ is supported inside $\delta \D$, for any $x$ in $\delta\D$ we can expand
			\[
			\partial_zG_\D(z,x)=-\frac12\left(\frac{1}{z-x}+\frac{\bar x}{1-z\bar x}\right)=-\frac12\sum_{n\geq 0}\left(x^nz^{-n-1}+\bar x^{n+1}z^n\right).
			\]
			As a consequence we have the equality (the series is actually a finite sum)
			\begin{align*}
				\int_{\D}\partial_zG_\D(z,x)\ps{u,f(x)}\d^2x&=-\frac12\sum_{n\geq 0}\int_\D \left(x^nz^{-n-1}+\bar x^{n+1}z^n\right)\ps{u,f(x)}\d^2x\\
				&=-\frac12\sum_{n\geq 0} \left(\ps{u,f_n}z^{-n-1}+\ps{u,f_{-n-1}}z^n\right).
			\end{align*}
			In a similar fashion $\ps{u,P\varphi(z)}=\sum_{n\geq 0}\ps{u,\varphi_n}z^n$, so that by definition of $\bm a_{-n}$ we have\\
            $\lim\limits_{\eps\to 0}U_0\left(\ps{u,\partial\X_\eps(z)}F\right)=\sum_{n\in\Z}z^{n-1}\ps{u,\i\bm a_{-n}}U_0F$. As a consequence we see that
			\begin{align*}
				\lim\limits_{\eps\to 0}U_0\left(\ps{u,\partial^m\X_\eps(z)}F\right)&=\sum_{n\in\Z}(-n-m+1)_{m-1} z^{-n-m}\ps{u,\i\A_n}U_0F\\
				&=\partial^m\ps{u,\Phi(z)}U_0(F)= m! Y(\ps{u,\varphi_{-m}},z).
			\end{align*}
			
			We then consider for $z\in\D$ the product
			\begin{align*}
				&\ephi{\prod_{k=1}^p\ps{u_k,\partial^{m_k}\X_\eps(z)}F\left(\X_\D+P\varphi+c\right)}.
			\end{align*}
			Via the same reasoning as above we can write that
			\begin{align*}
				&\ephi{\ps{u_1,\partial^{m_1}\X_\eps(z)}\prod_{k=2}^p\ps{u_k,\partial^{m_k}\X_\eps(z)}F\left(\X_\D+P\varphi+c\right)}\\
				&=\sum_{l=2}^p\expect{\ps{u_1,\partial^{m_1}\left(\X_\D\right)_\eps(z)}\ps{u_l,\partial^{m_l}\left(\X_\D\right)_\eps(z)}}\ephi{\prod_{k\neq 1,l}\ps{u_k,\partial^{m_k}\left(\X\right)_\eps(z)}F\left(\X\right)}\\
				&+\ps{u_1,\partial^{m_1}P\varphi(z)}\ephi{\prod_{k=2}^p\ps{u_k,\partial^{m_k}\X_\eps(z)}F\left(\X\right)}\\
				&-\frac12\sum_{n\geq 0} \left(\ps{u_1,f_n}z^{-n-1}+\ps{u_1,f_{-n-1}}z^n\right)\ephi{\prod_{k=2}^p\ps{u_k,\partial^{m_k}\X_\eps(z)}F\left(\X\right)}
			\end{align*}
			up to terms that will vanish in the $\eps\to 0$ limit.
			The terms on the first line will be canceled thanks to the definition of the Wick products. Hence
			\begin{align*}
				\lim\limits_{\eps\to0}U_0\left(:\prod_{k=1}^p\ps{u_k,\partial^{m_k}\X_\eps(z)}:F\right)&=\prod_{k=1}^p\left(\sum_{n\in\Z}(-n-m_k+1)_{m_k-1} z^{-n-m_k}\ps{u_k,\i\bm a_n}\right)U_0F.
			\end{align*}
			Now thanks to Equation~\eqref{eq:aA} this is actually equal to
			\begin{align*}
				:\prod_{k=1}^p\left(\sum_{n\in\Z}(-n-m_k+1)_{m_k-1} z^{-n-m_k}\ps{u_k,\i\bm A_n}\right):U_0F.
			\end{align*}
			Rewriting the latter in terms of the current $\Phi$ we conclude that, as expected,
			\begin{align*}
				\lim\limits_{\eps\to0}U_0\left(:\prod_{k=1}^p\ps{u_k,\partial^{m_k}\X_\eps(z)}:F\right)&=:\prod_{k=1}^p\ps{u_k,\partial^{m_k}\Phi(z)}:U_0F.
			\end{align*}
		\end{proof}

		\subsection{Proof of Proposition~\ref{prop:ope_prob}}
		
		\begin{proof}
			We proceed along the same lines as in the proof of Theorem~\ref{thm:VOA_Proba}, the only difference being that there will be extra terms of the form $\ps{u_k,u_l}\partial_{w_k}^{m_k}\partial_{w_l}^{m_l}G_{\D}(w_k,w_l)$ for $k\neq l$ coming from Gaussian integration by parts, and that will not be absorbed by Wick products. More precisely for $z\neq w$
			\begin{align*}
				&\ephi{\ps{u,\partial^{m}\X_\eps(w)}:\prod_{k=1}^q\ps{u_k,\partial^{m_k}\X_\eps(z)}:F\left(\X_\D+P\varphi+c\right)}\\
				&=\sum_{n}(-n-m+1)_{m-1}w^{-n-m-1}\ps{u,\i\bm a_{n}}\ephi{:\prod_{k=1}^q\ps{u_k,\partial^{m_k}\X_\eps(z)}:F\left(\X\right)}\\
				&+\sum_{r=1}^q\ps{u,u_r}\partial_{w}^m\partial_{z}^{m_r}G_{\D}(z,w)\ephi{:\prod_{k\neq r}\ps{u_k,\partial^{m_k}\X_\eps(z)}:F\left(\X\right)}
			\end{align*}
			up to vanishing terms. We can use the previous computations to treat the remaining expectation terms, while we can use the mode expansion of the Green function $G_\D$ to infer that
			\begin{align*}
				U_0\Big(\ps{u,\partial^{m}\X_\eps(w)}:\prod_{k=1}^q\ps{u_k,\partial^{m_k}\X_\eps(z)}:F\left(\X\right)\Big)=\partial_{w}^{m}\ps{u,\i\phi(w)}\prod_{k=1}^q\partial_z^{m_k}\ps{u_k,\i\phi(z)}U_0F\\
				+\sum_{r=1}^q\ps{u,u_r}\partial_{w}^{m}\partial_z^{m_r}G_\D(z,w)\prod_{k\neq r}\partial_z^{m_k}\ps{u_k,\i\phi(z)}U_0F
			\end{align*}
            in agreement with Theorem~\ref{thm:VOA_Proba}, and
			where $\phi(z)\coloneqq \bm a_0\log z-\sum_{n\in\Z}\frac1n\bm a_n z^{-n}$. 
			Now
			\begin{align*}
				&\partial_{w}\partial_zG_\D(z,w)=\left\{\begin{matrix}
					-\sum_{n\geq 1}\frac n2w^{n-1}z^{-n-1}\qt{if}\norm{z}>\norm{w}\\
					-\sum_{n\geq 1}\frac n2z^{n-1}w^{-n-1}\qt{if}\norm{w}>\norm{z}
				\end{matrix}\right.\qt{so that we get}\\
				&\left[\partial_{w}^{m}\ps{u,\i\phi(w)}\prod_{k=1}^q\partial_z^{m_k}\ps{u_k,i\phi(z)}-\sum_{\substack{n_0\geq1\\n_r\in\Z}}\frac{n_0}2\sum_{r=1}^q\ps{u,u_r}\delta_{n_0,-n_r}z^{n_0-1}w^{n_r-1}\prod_{k\neq r}\partial_z^{m_k}\ps{u_k,i\phi(z)}\right]U_0F
			\end{align*}
			provided $\norm{w}>\norm{z}$.
			Moreover for any $n_0,\cdots,n_q$
			\begin{align*}
				&\left(\ps{u,\i\bm a_{n_0}}\prod_{k=1}^q\ps{u_k,\i\bm a_{n_k}}-\frac{n_0}{2}\mathds 1_{n_0\geq 1}\sum_{r=1}^q\ps{u,u_r}\delta_{n_0,-n_r}\prod_{k\neq r}\ps{u_k,\i\bm a_{n_k}}\right)U_0F\\
				&=\ps{u,\i\A_{n_0}}:\prod_{k=1}^q\ps{u_k,\i\A_{n_k}}:U_0F.
			\end{align*}
			This allows to conclude that, as soon as $\norm{w}>\norm{z}$,
			\begin{align*}
				U_0\Big(\ps{u,\partial^{m}\X(w)}:\prod_{k=1}^q\ps{u_k,\partial^{m_k}\X(z)}:F\left(\X\right)\Big)=\ps{u,\partial^{m}\Phi(w)}:\prod_{k=1}^q\ps{u_k,\partial^{m_k}\Phi(z)}:U_0F.
			\end{align*}
			The proof works in the same fashion if $\norm{w}<\norm{z}$ concluding for the proof in the case $p=1$.
			
			For $p\geq 2$ we rely on the two analog formulas for Wick products and normally ordered ones as described by Equations~\eqref{eq:wick_gen} and~\eqref{eq:normal_ord_gen}. This allows to extend the reasoning from the $p=1$ case to the general one as well as for general values of $n$.
		\end{proof}
		
		\subsection{Proof of Proposition~\ref{prop:VO_prim}}
		
		\begin{proof}
			We use the same notations as in the proofs of Theorem~\ref{thm:VOA_Proba} and Proposition~\ref{prop:ope_prob}, and take $F(\X)=e^{\ps{\X,f}}$ with $f\in\mc E_\delta$. To start with when $\alpha_k\in\a$ Girsanov's theorem~\ref{thm:girsanov} gives
			\begin{align*}
				U_0\left(\prod_{k=1}^NV_{\alpha_k}(z_k)F(\X)\right)&=\prod_{k=1}^N:e^{\ps{\alpha_k,P\varphi(z_k)+c}}:\prod_{k<l}e^{\ps{\alpha_k,\alpha_l}G_\D(z_k,z_l)}U_0\left(F\left(\X+\sum_{k=1}^N\alpha_kG_\D(z_k,\cdot)\right)\right)\\
				&=\prod_{k=1}^N:e^{\ps{\alpha_k,P\varphi(z_k)+c}}:e^{\ps{\alpha_k G_\D(z_k,\cdot),f}}\prod_{k<l}e^{\ps{\alpha_k,\alpha_l}G_\D(z_k,z_l)}U_0F.
			\end{align*}
			To extend the equality to $\alpha_k\in\h$, note that for $\eps>0$ the map $\bm\alpha\mapsto U_0\left(\prod_{k=1}^NV_{\alpha_k,\eps}(z_k)F(\X_\eps)\right)$ is analytic in $\ps{\alpha_k,u}$ for any $u\in\h$ and $1\leq k\leq n$. The same applies to $\bm\alpha\mapsto\prod_{k=1}^N:e^{\ps{\alpha_k,P\varphi_\eps(z_k)+c}}:e^{\ps{\alpha_k G_{\D,\eps}(z_k,\cdot),f}}\prod_{k<l}e^{\ps{\alpha_k,\alpha_l}G_{\D,\eps}(z_k,z_l)}U_0(F(\X_\eps))$, where $G_{\D,\eps}$ is covariance kernel of $\X_{\D,\eps}$. Hence the equality holds at the regularized level for $\bm\alpha$ complex too, showing the result by taking $\eps\to0$.

			On the other hand, thanks to the explicit expression of $F$
			\begin{align*}
				&e^{-\ps{\alpha,\partial_c+Q}\log\norm{z}}U_0F=e^{-\ps{Q,c}}e^{-\ps{\alpha,\partial_c}\log\norm{z}}e^{\ps{c,f}}\ephi{F(\X_\D)}=e^{-\ps{\alpha,f_0}\log\norm{z}}U_0F\qt{while}\\
				&\exp\left( \frac 1{2n}\ps{\alpha,\partial_n}z^{-n}\right)U_0F=\exp\left( \frac 1{2n}\ps{\alpha,f_n}z^{-n}\right)U_0F\qt{for any non-zero $n$. Hence}
			\end{align*}
			\begin{align*}
				:e^{\ps{\alpha,\i\A(z)}}:U_0F=e^{-\frac12\ps{\alpha,f_0}\ln z}\prod_{n\geq1}\exp\left(\ps{\alpha,\varphi_n}z^{n}\right) \prod_{n\in\Z^*}\exp\left(\frac 1{2n}\ps{\alpha,f_n}z^{-n}\right)U_0F.
			\end{align*}
			Applying $:e^{\ps{\alpha,c+\i\bar\A(\bar z)}}:$ we get 
			\begin{align*}
				:e^{\ps{\alpha,\Phi(z)}}:U_0F=e^{\ps{\alpha,c+P\varphi(z)}}(1-\norm{z}^2)^{\frac{\ps{\alpha,\alpha}}{2}} e^{-\ps{\alpha,f_0}\ln\norm{ z}}\prod_{n\in\Z^*}\exp\left(\frac 1{2n}\left(\ps{\alpha,f_n}z^{-n}+\ps{\alpha,f_{-n}}\bar z^{n}\right)\right)U_0F.
			\end{align*}
			Now since $f$ is supported inside $\delta\D$ we have the expansion:
			\begin{align*}
				\ps{\alpha G_\D(\cdot,z),f}&=\int_{\delta\D}\left(-\log\norm{z}+\sum_{n\geq1}\frac1{2n}\left(x^nz^{-n}+\bar x^n\bar z^{-n}+x^n\bar z^{n}+\bar x^nz^{n}\right)\right)\ps{\alpha,f(x)}\d^2x\\
				&=-\ps{\alpha,f_0}\log\norm{z}+ \sum_{n\geq1}\frac1{2n}\left(\ps{\alpha,f_n}z^{-n}+\ps{\alpha,f_{-n}}\bar z^{-n}+\ps{\alpha,f_n}\bar z^{n}+\ps{\alpha,f_{-n}}z^{n}\right).
			\end{align*}
			Using the fact that $\expect{\ps{\alpha,P\varphi(z)}^2}=\ps{\alpha,\alpha}\ln\frac1{1-\norm{z}^2}$ this allows to conclude that
			\begin{align*}
				:e^{\ps{\alpha,\Phi(z)}}:U_0F=:e^{\ps{\alpha,P\varphi(z)+c}}: e^{\ps{\alpha G_\D(z,\cdot),f}}U_0F=U_0\left(V_\alpha(z)F\right).
			\end{align*}
			
			When more than one primary fields are involved we may proceed in a similar fashion by writing
			\begin{align*}
				:e^{\ps{\alpha,\Phi(z)}}::e^{\ps{\beta,P\varphi(w)+c}}:=:e^{\ps{\alpha,P\varphi(z)+c}}:&:e^{\ps{\beta,P\varphi(w)+c}}:e^{\ps{\alpha,\beta}G_\D(z,w)}\\
				& e^{-\ps{\alpha,\partial_c+Q}\log\norm{z}}\prod_{n\in\Z^*}\exp\left( \frac 1{2n}\ps{\alpha,\partial_nz^{-n}+\partial_{-n}\bar z^{-n}}\right)
			\end{align*}
			via the very same reasoning and provided that $\norm{z}>\norm{w}$. An immediate recursion gives, as desired,
			\begin{align*}
				\mc V_{\alpha_1}(z_1)\cdots\mc V_{\alpha_N}(z_n)U_0F&=\prod_{k=1}^N:e^{\ps{\alpha_k,P\varphi(z_k)+c}}:e^{\ps{\alpha_k G_\D(z_k,\cdot),f}}\prod_{k<l}e^{\ps{\alpha_k,\alpha_l}G_\D(z_k,z_l)}U_0F\\
				&=U_0\left(\prod_{k=1}^NV_{\alpha_k}(z_k)F\right).
			\end{align*}
			
			As for our second claim, it follows from the first one by taking a sequence of functionals $F_\eps$ of the form $F_\eps(\X)=\prod_{k=1}^p\ps{\X,(-\partial)^{m_k}\rho_\eps}$ where $(\rho_\eps)_{\eps>0}$ is a sequence of smooth approximations of unity with support in $B(0,\eps)$.
		\end{proof}
		

		
		\bibliography{biblio}

\begin{thebibliography}{10}

\bibitem{Za_E8}
{A.B. Zamolodchikov}.
\newblock {Integrals of motion and S-matrix of the (scaled) $T=T_c$ Ising model
  with magnetic field}.
\newblock {\em {International Journal of Modern Physics A}}, 04(16):4235--4248,
  1989.

\bibitem{ACBK}
D.~Adame-Carrillo, D.~Behzad, and K.~Kyt\"ol\"a.
\newblock {Fock space of local fields of the discrete GFF and its scaling limit
  bosonic CFT }.
\newblock {\em Preprint,
  \href{http://arxiv.org/abs/2404.15490}{\textup{\nolinkurl{arXiv:2404.15490}}}},
  2024.

\bibitem{AGT}
L.F. Alday, D.~Gaiotto, and Y.~Tachikawa.
\newblock {Liouville Correlation Functions from Four-Dimensional Gauge
  Theories}.
\newblock {\em {Letters in Mathematical Physics}}, 91:167--197, 2010.

\bibitem{AS_DOZZ}
M.~Ang, G.~Cai, X.~Sun, and B.~Wu.
\newblock {Integrability of the conformal loop ensemble: imaginary DOZZ formula
  and beyond}.
\newblock {\em Preprint,
  \href{http://arxiv.org/abs/2107.01788}{\textup{\nolinkurl{arXiv:2107.01788}}}},
  2021.

\bibitem{ARS}
M.~Ang, G.~Remy, and X.~Sun.
\newblock {FZZ formula of boundary Liouville CFT via conformal welding}.
\newblock {\em Journal of the European Mathematical Society}, 27(3):1209--1266,
  2025.

\bibitem{ARSZ}
M.~Ang, G.~Remy, X.~Sun, and T.~Zhu.
\newblock {Derivation of all structure constants for boundary Liouville CFT}.
\newblock {\em Preprint,
  \href{http://arxiv.org/abs/2305.18266}{\textup{\nolinkurl{arXiv:2305.18266}}}},
  2023.

\bibitem{Aomoto}
K.~Aomoto.
\newblock {On the complex Selberg integral}.
\newblock {\em The Quarterly Journal of Mathematics}, 38(4):385--399, 12 1987.

\bibitem{Arakawa_rep}
T.~Arakawa.
\newblock {Representation theory of ${\mathcal{W}}$-algebras}.
\newblock {\em Inventiones mathematicae}, 169:219–320, 2007.

\bibitem{Arakawa_intro}
T.~Arakawa.
\newblock {Introduction to W-Algebras and Their Representation Theory}.
\newblock In {\em {Perspectives in Lie Theory}}, pages 179--250. Springer
  International Publishing, Cham, 2017.

\bibitem{AF}
T.~Arakawa and E.~Frenkel.
\newblock {Quantum Langlands duality of representations of
  ${\mathcal{W}}$-algebras}.
\newblock {\em Compositio Mathematica}, 155(12):2235–2262, 2019.

\bibitem{BFFOrW3}
J.~Balog, L.~Fehér, P.~Forg\'acs, L.~O'Raifeartaigh, and A.~Wipf.
\newblock {Kac-Moody realization of $\mathcal W$-algebras}.
\newblock {\em {Physics Letters}}, 244B:435, 1990.

\bibitem{BGKR}
G.~Baverez, C.~Guillarmou, A.~Kupiainen, and R.~Rhodes.
\newblock {Semigroup of annuli in Liouville CFT}.
\newblock {\em Preprint,
  \href{http://arxiv.org/abs/2403.10914}{\textup{\nolinkurl{arXiv:2403.10914}}}},
  2024.

\bibitem{BGKRV}
G.~Baverez, C.~Guillarmou, A.~Kupiainen, R.~Rhodes, and V.~Vargas.
\newblock {The Virasoro structure and the scattering matrix for Liouville
  conformal field theory}.
\newblock {\em Probability and Mathematical Physics}, 5(2):269--320, 2024.

\bibitem{BJ_SLE}
G.~Baverez and A.~Jego.
\newblock {The CFT of SLE loop measures and the Kontsevich--Suhov conjecture}.
\newblock {\em Preprint,
  \href{http://arxiv.org/abs/2407.09080}{\textup{\nolinkurl{arXiv:2407.09080}}}},
  2024.

\bibitem{BPZ}
A.A. Belavin, A.M. Polyakov, and A.B. Zamolodchikov.
\newblock Infinite conformal symmetry in two-dimensional quantum field theory.
\newblock {\em Nuclear Physics B}, 241(2):333 -- 380, 1984.

\bibitem{Borcherds}
R.~Borcherds.
\newblock {Vertex algebras, Kac-Moody algebras, and the Monster}.
\newblock In {\em {Proceedings of the National Academy of Sciences of the
  United States of America}}, volume~83, pages 3068--3071, 1986.

\bibitem{Selberg}
A.~Borodin and I.~Corwin.
\newblock {Macdonald processes}.
\newblock {\em {Probability Theory and Related Fields}}, 158:225--400, 2014.

\bibitem{magnet}
{Borthwick, D. and Garibaldi, S.}
\newblock {Did a 1-Dimensional Magnet Detect a 248-Dimensional Lie Algebra?}
\newblock {\em {Notices of the AMS}}, 58(8):1055--1066, September 2011.

\bibitem{BouSch}
P.~Bouwknegt and K.~Schoutens.
\newblock W symmetry in conformal field theory.
\newblock {\em Physics Reports}, 223(4):183--276, 1993.

\bibitem{Toda_correl1}
B.~Cercl\'e.
\newblock {Three-point correlation functions in the $\mathfrak{sl}_3$ Toda
  theory I: Reflection coefficients}.
\newblock {\em Probability Theory and Related Fields}, 188:89--158, 2024.

\bibitem{Toda_correl2}
B.~Cercl\'e.
\newblock {Three-point correlation functions in the $\mathfrak{sl}_3$ Toda
  theory II: the Fateev-Litvinov formula}.
\newblock {\em Journal of the European Mathematical Society}, to appear, 2022.

\bibitem{Toda_OPEWV}
B.~Cercl\'e and Y.~Huang.
\newblock {Ward identities in the $\mathfrak{sl}_3$ Toda conformal field
  theory}.
\newblock {\em Communications in Mathematical Physics}, 393:419--475, 2022.

\bibitem{Toda_construction}
B.~Cercl\'e, R.~Rhodes, and V.~Vargas.
\newblock {Probabilistic construction of Toda conformal field theories}.
\newblock {\em Annales Henri Lebesgue}, 6:31--64, 2023.

\bibitem{CH_construction}
B.~Cerclé and N.~Huguenin.
\newblock {Boundary Toda Conformal Field Theory from the path integral}.
\newblock {\em Preprint,
  \href{http://arxiv.org/abs/2402.02888}{\textup{\nolinkurl{arXiv:2402.02888}}}},
  2024.

\bibitem{DKRV}
F.~David, A.~Kupiainen, R.~Rhodes, and V.~Vargas.
\newblock Liouville {Q}uantum {G}ravity on the {R}iemann {S}phere.
\newblock {\em Communications in Mathematical Physics}, 342:869--907, 2016.

\bibitem{DF1}
V.~S. Dotsenko and V.~A. Fateev.
\newblock {Four Point Correlation Functions and the Operator Algebra in the
  Two-Dimensional Conformal Invariant Theories with the Central Charge c
  \ensuremath{<} 1}.
\newblock {\em Nucl. Phys. B}, 251:691--734, 1985.

\bibitem{DF2}
Vl.S. Dotsenko and V.A. Fateev.
\newblock Conformal algebra and multipoint correlation functions in 2d
  statistical models.
\newblock {\em Nuclear Physics B}, 240(3):312--348, 1984.

\bibitem{DS}
V.~G. Drinfeld and V.~V. Sokolov.
\newblock {Lie algebras and equations of Korteweg-de Vries type}.
\newblock {\em J. Sov. Math.}, 30:1975--2036, 1984.

\bibitem{dubedat}
J.~Dub\'edat.
\newblock {SLE and the Free Field: partition functions and couplings}.
\newblock {\em {Journal of the AMS}}, 22 (4):995--1054, 2009.

\bibitem{FZZ}
V.~Fateev, A.~Zamolodchikov, and Al. Zamolodchikov.
\newblock {Boundary Liouville Field Theory I. Boundary State and Boundary
  Two-point Function}.
\newblock {\em Preprint,
  \href{http://arxiv.org/abs/0001012}{\textup{\nolinkurl{arXiv:0001012}}}},
  2000.

\bibitem{FaLi1}
V.A. Fateev and A.V. Litvinov.
\newblock {Correlation functions in conformal Toda field theory. I.}
\newblock {\em JHEP}, 11:002, 2007.

\bibitem{FaLu}
V.A. Fateev and S.~L. Lukyanov.
\newblock {The Models of Two-Dimensional Conformal Quantum Field Theory with
  Z(n) Symmetry}.
\newblock {\em Int. J. Mod. Phys. A}, 3:507, 1988.

\bibitem{FF_DS}
B.~Feigin and E.~Frenkel.
\newblock {Quantization of the Drinfeld-Sokolov reduction}.
\newblock {\em Phys. Lett. B}, 246:75--81, 1990.

\bibitem{FF_KM}
B.~Feigin and E.~Frenkel.
\newblock {Affine Kac-Moody algebras at the critical level and Gelfand-Dikii
  algebras}.
\newblock {\em Int. J. Mod. Phys. A}, 7S1A:197--215, 1992.

\bibitem{FF_QG}
B.~Feigin and E.~Frenkel.
\newblock {Integrals of motion and quantum groups}.
\newblock {\em Lect. Notes Math.}, 1620:349--418, 1996.

\bibitem{yellow_book}
P.~Francesco, P.~Mathieu, and D.~Sénéchal.
\newblock {\em Conformal Field Theory}.
\newblock Springer-Verlag New York, 1996.

\bibitem{FBZ}
E.~Frenkel and D.~Ben-Zvi.
\newblock {\em Vertex Algebras and Algebraic Curves: Second Edition}.
\newblock American Mathematical Society, 2004.

\bibitem{FKW}
E.~Frenkel, V.~Kac, and M.~Wakimoto.
\newblock {Characters and fusion rules for $W$-algebras via quantized
  Drinfel'd-Sokolov reduction}.
\newblock {\em Communications in Mathematical Physics}, 147(2):295 -- 328,
  1992.

\bibitem{FLM89}
I.~Frenkel, J.~Lepowsky, and A.~Meurman.
\newblock {\em Vertex Operator Algebras and the Monster}.
\newblock Academic Press, 1989.

\bibitem{FHV}
I.~B. Frenkel, Y.-Z. Huang, and J.~Lepowsky.
\newblock {\em On Axiomatic Approaches to Vertex Operator Algebras and
  Modules}, volume 104.
\newblock American Mathematical Society, 1993.

\bibitem{FdV}
H.~Freudenthal and H.~de~Vries.
\newblock {\em Linear lie groups}.
\newblock Pure and Applied Mathematics. Academic Press, 1969.

\bibitem{FS87}
D.~Friedan and S.~Shenker.
\newblock The analytic geometry of two-dimensional conformal field theory.
\newblock {\em Nuclear Physics B}, 281(3):509--545, 1987.

\bibitem{Gin_CFT}
P.~H. Ginsparg.
\newblock {Applied Conformal Field Theory}.
\newblock In {\em {Les Houches Summer School in Theoretical Physics: Fields,
  Strings, Critical Phenomena}}, 9 1988.

\bibitem{GKR_CILT}
C.~Guillarmou, A.~Kupiainen, and R.~Rhodes.
\newblock {Compactified imaginary Liouville theory}.
\newblock {\em Preprint,
  \href{http://arxiv.org/abs/2310.18226}{\textup{\nolinkurl{arXiv:2310.18226}}}},
  2023.

\bibitem{GKRV_Segal}
C.~Guillarmou, A.~Kupiainen, R.~Rhodes, and V.~Vargas.
\newblock {Segal's axioms and bootstrap for Liouville theory}.
\newblock {\em Preprint,
  \href{http://arxiv.org/abs/2112.14859}{\textup{\nolinkurl{arXiv:2112.14859}}}},
  2021.

\bibitem{GKRV}
C.~Guillarmou, A.~Kupiainen, R.~Rhodes, and V.~Vargas.
\newblock {Conformal bootstrap in Liouville Theory}.
\newblock {\em Acta Mathematica}, 233(1):33--194, 2024.

\bibitem{GRW}
C.~Guillarmou, R.~Rhodes, and B.~Wu.
\newblock {Conformal Bootstrap for surfaces with boundary in Liouville CFT.
  Part 1: Segal axioms}.
\newblock {\em Preprint,
  \href{http://arxiv.org/abs/2408.13133}{\textup{\nolinkurl{arXiv:2408.13133}}}},
  2024.

\bibitem{Hos}
K.~Hosomichi.
\newblock {Bulk boundary propagator in Liouville theory on a disc}.
\newblock {\em JHEP}, 11:044, 2001.

\bibitem{Hua_CFT}
Y.-Z. Huang.
\newblock {\em Two-Dimensional Conformal Geometry and Vertex Operator
  Algebras}, volume 148.
\newblock Birkhäuser Boston, MA, 1997.

\bibitem{Fuchsian}
K.~Iwasaki, H.~Kimura, S.~Shimomura, and M.~Yoshida.
\newblock {\em {From Gauss to Painlevé: A Modern Theory of Special
  Functions}}, volume~16 of {\em {Aspects of Mathematics}}.
\newblock Vieweg+Teubner Verlag Wiesbaden, 1991.

\bibitem{Janson}
S.~Janson.
\newblock {\em {Gaussian Hilbert Spaces}}.
\newblock Cambridge Tracts in Mathematics. Cambridge University Press, 1997.

\bibitem{Kac_VOA}
V.~Kac.
\newblock {\em {Vertex Algebras for Beginners: Second Edition}}.
\newblock University Lecture Series, volume 10, 1998.

\bibitem{KRW}
V.~Kac, SS. Roan, and M.~Wakimoto.
\newblock Quantum reduction for affine superalgebras.
\newblock {\em Communications in Mathematical Physics}, 241:307--342, 2003.

\bibitem{Kah}
J.-P. Kahane.
\newblock Sur le chaos multiplicatif.
\newblock {\em Annales des sciences math{\'{e}}matiques du Qu{\'{e}}bec}, 1985.

\bibitem{Kang-Makarov}
{Kang, N.-G. and Makarov, N. G.}
\newblock {Gaussian free field and conformal field theory}.
\newblock {\em {Astérisque}}, no. 353, 2013.

\bibitem{KW}
H.~G. Kausch and G.~M.~T. Watts.
\newblock {Quantum Toda theory and the Casimir algebra of B(2) and C(2)}.
\newblock {\em Int. J. Mod. Phys. A}, 7:4175--4187, 1992.

\bibitem{ABCDEFG}
C.~A. Keller, N.~Mekareeya, J.~Song, and Y.~Tachikawa.
\newblock {The ABCDEFG of Instantons and W-algebras}.
\newblock {\em JHEP}, 03:045, 2012.

\bibitem{KRV_loc}
A.~Kupiainen, R.~Rhodes, and V.~Vargas.
\newblock Local {C}onformal {S}tructure of {L}iouville {Q}uantum {G}ravity.
\newblock {\em Communications in Mathematical Physics}, 2018.

\bibitem{KRV_DOZZ}
A.~Kupiainen, R.~Rhodes, and V.~Vargas.
\newblock {Integrability of Liouville theory: proof of the DOZZ formula}.
\newblock {\em Annals of Mathematics}, 191(1):81--166, 2020.

\bibitem{FaLuDn}
Sergei Lukyanov and V.A. Fateev.
\newblock {Exactly Soluble Models of Conformal Quantum Field Theory Associated
  with the Simple Lie Algebra $D_n$}.
\newblock {\em Soviet Journal of Nuclear Physics}, 49:5, 06 1988.

\bibitem{MSW}
J.~Miller, S.~Sheffield, and W.~Werner.
\newblock {Simple conformal loop ensembles on Liouville quantum gravity}.
\newblock {\em Ann. Probab.}, 50(3):905 -- 949, 2022.

\bibitem{MT}
K.~Mimachi and T.~Takamuki.
\newblock {A generalization of the beta integral arising from the
  Knizhnik-Zamolodchikov equation for the vector representations of types $B_n$
  , $C_n$ and $D_n$}.
\newblock {\em Kyushu Journal of Mathematics}, 59(1):117--126, 2005.

\bibitem{MuVa}
E.~Mukhin and A.~Varchenko.
\newblock {\em {Remarks on critical points of phase functions and norms of
  Bethe vectors}}, volume~27 of {\em {Advanced Studies in Pure Mathematics}}.
\newblock Arrangements - Tokyo 1998, 2000.

\bibitem{Neretin}
Y.A. Neretin.
\newblock {On the Dotsenko–Fateev complex twin of the Selberg integral and
  its extensions}.
\newblock {\em The Ramanujan Journal}, 64:37--55, 2024.

\bibitem{Backbone}
P.~Nolin, W.~Qian, X.~Sun, and Z.~Zhuang.
\newblock {Backbone exponent for two-dimensional percolation}.
\newblock {\em Preprint,
  \href{http://arxiv.org/abs/2309.05050}{\textup{\nolinkurl{arXiv:2309.05050}}}},
  2023.

\bibitem{Importance}
{P. Forrester and S. Warnaar}.
\newblock {The importance of the Selberg integral}.
\newblock {\em Bulletin of the American Mathematical Society}, 45:489--534,
  2007.

\bibitem{Pol81}
A.~Polyakov.
\newblock {Quantum Geometry of bosonic strings}.
\newblock {\em Physics Letters B}, 103:207:210, 1981.

\bibitem{PT02}
B.~Ponsot and J.~Teschner.
\newblock {Boundary Liouville field theory: boundary three-point function}.
\newblock {\em Nuclear Physics B}, 622(1):309--327, 2002.

\bibitem{RY91}
D.~Revuz and M.~Yor.
\newblock {\em {Continuous Martingales and Brownian Motion}}.
\newblock Springer, Berlin, 1991.

\bibitem{Seg04}
G.~Segal.
\newblock The definition of conformal field theory.
\newblock In {\em Topology, Geometry, and Quantum Field Theory. Proc. Oxford
  2002}. Oxford Univ. Press 2004, 2004.

\bibitem{She07}
S.~Sheffield.
\newblock Gaussian free field for mathematicians.
\newblock {\em {Probability Theory and Related Fields}}, 139:521, 2007.

\bibitem{Sh_CLE}
S.~Sheffield.
\newblock {Exploration trees and conformal loop ensembles}.
\newblock {\em Duke Mathematical Journal}, 147(1):79 -- 129, 2009.

\bibitem{SW}
S.~Sheffield and W.~Werner.
\newblock {Conformal loop ensembles: the Markovian characterization and the
  loop-soup construction}.
\newblock {\em Annals of Mathematics}, 176:1827--1917, 2012.

\bibitem{TaVa}
V.~Tarasov and A.~Varchenko.
\newblock {Selberg-Type Integrals Associated with $\mathfrak{sl}_3$}.
\newblock {\em Letters in Mathematical Physics}, 65:173--185, 2003.

\bibitem{Teschner_revisited}
J.~Teschner.
\newblock Liouville theory revisited.
\newblock {\em Classical and Quantum Gravity}, 18(23):R153--R222, nov 2001.

\bibitem{Teschner_DOZZ}
Jörg Teschner.
\newblock {On the Liouville three-point function}.
\newblock {\em Physics Letters B}, 363(1):65--70, 1995.

\bibitem{Tsukada}
H.~Tsukada.
\newblock {\em String Path Integral Realization of Vertex Operator Algebras},
  volume~91.
\newblock American Mathematical Society, 1991.

\bibitem{War_Sel}
S.~O. Warnaar.
\newblock {A Selberg integral for the Lie algebra $A_n$}.
\newblock {\em Acta Mathematica}, 203:269--304, 2009.

\bibitem{Za85}
A.~B. {Zamolodchikov}.
\newblock {Infinite additional symmetries in two-dimensional conformal quantum
  field theory}.
\newblock {\em Theoretical and Mathematical Physics}, 65(3):1205--1213,
  December 1985.

\end{thebibliography}
		\bibliographystyle{plain}
	\end{document}